\documentclass[journal]{IEEEtran} \onecolumn 
\usepackage{url}

\usepackage[utf8]{inputenc}
\usepackage[T1]{fontenc}
\usepackage{hyperref}              
\usepackage{cite}             
\usepackage{array}
\usepackage{multirow}
\usepackage{lscape}
\usepackage[cmex10]{amsmath}  
\usepackage{mleftright}       
\mleftright                   

\usepackage{amsfonts,bm,graphicx,tikz,subcaption,algorithm,algorithmic}
\usepackage{amsthm}
\usepackage{amsmath,amssymb}                          
\usepackage{booktabs}         
\usepackage{ bbold }                            
\newcommand{\norm}[1]{%
	\| #1 \|}

\newcommand{\gm}{{\tiny{GM}}}

\newcommand{\sigmax}{\sigma_{1}^*{}}
\newcommand{\sigmin}{\sigma_{r}^*{}}

\newcommand{\x}{\bm{x}}
\newcommand{\y}{\bm{y}}
\newcommand{\g}{\bm{g}}
\renewcommand{\b}{\bm{b}}

\newcommand{\X}{\bm{X}}
\newcommand{\Y}{\bm{Y}}
\newcommand{\A}{\bm{A}}
\newcommand{\U}{{\bm{U}}}
\newcommand{\V}{\bm{V}}
\newcommand{\B}{\bm{B}}
\newcommand{\xstar}{\x^*}
\newcommand{\bstar}{\b^*}

\newcommand{\Xstar}{\X^*}
\newcommand{\Ustar}{\U^*}
\newcommand{\Vstar}{\V^*}
\newcommand{\Bstar}{\B^*}
\newcommand{\bSigma}{\bm{\Sigma}^*}

\newcommand{\Sig}{\bm{\Sigma}}

\newtheorem{theorem}{Theorem}[section]
\newtheorem{cor}[theorem]{Corollary}
\newtheorem{claim}[theorem]{Claim}
\newtheorem{remark}[theorem]{Remark}
\newtheorem{definition}{Definition}
\newtheorem{assu}{Assumption}
\newtheorem{lemma}[theorem]{Lemma}
\newtheorem{fact}[theorem]{Fact}

\renewcommand\thetheorem{\arabic{section}.\arabic{theorem}}

\newcommand{\cS}{\mathcal{S}}

\renewcommand{\S}{\bm{S}}

\newcommand{\ik}{{ki}}

\newcommand{\Xhat}{\hat\X}

\newcommand{\tC}{\tilde{C}}
\newcommand{\tc}{\tilde{c}}

\newcommand{\sigmamin}{\sigma_{min}}

\newcommand{\pul}{\mathcal{P}_{\U_{\ell}}}	
\newcommand{\pulbst}{\mathcal{P}_{\U_{\ell_{best}}}}
\newcommand{\pustar}{\mathcal{P}_{\Ustar}}

\renewcommand{\d}{\bm{d}}
\newcommand{\dstar}{\d}
\newcommand{\tL}{\tilde{L}}

\begin{document}

\newcommand{\ben}{\begin{enumerate}} \newcommand{\een}{\end{enumerate}} \newcommand{\bi}{\begin{itemize}} \newcommand{\ei}{\end{itemize}}
	
	\newcommand{\indic}{\mathbb{1}}
	\newcommand{\E}{\mathbb{E}}
	\newcommand{\SD}{\bm{SD}}
	\newcommand{\eps}{\epsilon}
	\newcommand{\tdelta}{\tilde\delta}
	
	\newcommand{\I}{\bm{I}}
	\newcommand{\M}{\bm{M}}
	
	\renewcommand{\a}{\bm{a}}
	\newcommand{\e}{\bm{e}}

	\newcommand{\G}{\bm{G}}
	
	\newcommand{\Err}{\mathrm{Err}}
	\newcommand{\GradErr}{\mathrm{GradErr}}
	\newcommand{\Bone}{{\B^{(1)}}}
	\newcommand{\Xone}{\X^{(1)}}
	\newcommand{\Bell}{{\B^{(\ell)}}}
	\newcommand{\Xell}{\X^{(\ell)}}
	\newcommand{\bell}{\b^{(\ell)}}
	\newcommand{\xell}{\x^{(\ell)}}

	\newcommand{\Uhat}{\hat\U}
	\renewcommand{\P}{\mathcal{P}}
	\newcommand{\D}{\bm{D}}
	
	\newcommand{\BGM}{\B}
	\newcommand{\XGM}{\X}
	\newcommand{\xkGM}{\x_k}
	\newcommand{\bkGM}{\b_k}
\newcommand{\Sigmastar}{\bm\Sigma^*}

	\newcommand{\Jgood}{\mathcal{J}_{good}}

\newcommand{\thresh}{\omega_\gm}
\newcommand{\z}{\bm{z}}
\newcommand{\Z}{\bm{Z}}
\newcommand{\bPhi}{\bm{\Phi}}
\newcommand{\tz}{\tilde{\z}}

\newcommand{\m}{\widetilde{m}}
\newcommand{\q}{\widetilde{q}}
\newcommand{\qell}{q_\ell}
\newcommand{\totl}{q_\ell}
\newcommand{\trnc}{\mathrm{trunc}}
\newcommand{\sigmarstar}{\sigma_r^*}
\newcommand{\sigmarplusstar}{\sigma_{r+1}^*}
\newcommand{\sigmaonestar}{\sigma_1^*}
\renewcommand{\Xhat}{\X}

\title{Byzantine-Resilient Federated PCA and Low Rank Matrix Recovery}	
	\title{Byzantine-Resilient Federated PCA \\ and Low Rank Column-wise Sensing}	
	\author{%
		\IEEEauthorblockN{Ankit Pratap Singh and Namrata Vaswani}
		\IEEEauthorblockA{Iowa State University, Ames, IA, USA}
\thanks{An early version of the first part of this work will be presented at ISIT 2024 \cite{altgdmin_isit}. A Byzantine resilient federated few shot learning solution using the ideas introduced in this work will be presented at ICML 2024 \cite{altgdmin_icml}.
}
	}
	
	\maketitle

\newcommand{\approxgm}{{\tiny{\text{approxGM}}}}

\begin{abstract}
This work considers two related learning problems in a federated attack-prone setting -- federated principal components analysis (PCA) and federated low rank column-wise sensing (LRCS).  The node attacks are assumed to be Byzantine which means that the attackers are omniscient and can collude.  We introduce a novel provably Byzantine-resilient communication-efficient and sample-efficient algorithm, called Subspace-Median, that solves the PCA problem and is a key part of the solution for the LRCS problem. We also study the most natural Byzantine-resilient solution for federated PCA, a geometric median based modification of the federated power method, and explain why it is not useful.
Our second main contribution is a complete alternating gradient descent (GD) and minimization (altGDmin) algorithm for Byzantine-resilient horizontally federated LRCS and sample and communication complexity guarantees for it. Extensive simulation experiments are used to corroborate our theoretical guarantees. The ideas that we develop for LRCS are easily extendable to other LR recovery problems as well.
	\end{abstract}

\section{Introduction}
Federated learning is a setting where multiple entities/nodes/clients collaborate in solving a machine learning (ML) problem. Each node can only communicate with a central server or service provider that we refer to as ``center'' in this paper. The data observed or measured at each node/client is stored locally and should not be shared with the center. Summaries of it can be shared with the center. The center typically aggregates the received summaries and broadcasts the aggregate to all the nodes \cite{kairouz2021advances}.
One of the challenges in this setup is adversarial attacks on the nodes. In this work we assume Byzantine attacks, i.e., the adversarial nodes are omniscient and can collude \cite{guerraoui2018hidden, chen2018draco, krum, alistarh2018byzantine, yin2018byzantine}. ``Omniscient'' means that the attacking nodes have knowledge of all the data at every node and the exact algorithm (and all its parameters) implemented by every node, including center, and can use this information to design the worst possible attacks at each algorithm iteration.

This work develops provably Byzantine resilient algorithms for solving two related problems -- federated principal components analysis (PCA) and horizontally-federated low rank (LR) column-wise sensing or LRCS -- in a communication- and sample-efficient fashion. 
The first goal in solving both problems is to reliably estimate the subspace spanned by the top $r$ singular vectors of an unknown symmetric $n \times n$ matrix, $\bPhi^*$. In case of PCA, $\bPhi^*$ is the population covariance matrix of the available data. For each $\ell=1,2,\dots,L$, node $\ell$ observes an $n \times \totl$  data matrix $\D_\ell$ which can be used to compute an estimate $\bPhi_\ell:=\D_\ell \D_\ell^\top /\totl$ of $\bPhi^*$. 
PCA is well known to have a large number of applications in scientific visualization and as a pre-processing step for speeding up various ML tasks. LRCS finds applications in accelerated dynamic MRI \cite{dyn_mri1,lrpr_gdmin_mri_jp}, multi-task linear representation learning and few shot learning \cite{du2020few, collins2021exploiting,altgdmin_icml},  and federated sketching \cite{hughes_icip_2012,hughes_icml_2014,lee2019neurips}.

\subsection{Existing Work}

\subsubsection{Byzantine-resilient federated machine learning (ML)}
There has been a large amount of recent work on Byzantine-resilient federated ML algorithms, some of which come with provable guarantees \cite{dist_adversarial,krum,yin2018byzantine,xie2019zeno,alistarh2018byzantine,cao2020fltrust,allen2020byzantine,wu2020federated,defazio2014saga,acharya2022robust, pillutla2019robust,data2021byzantine,li2019rsa,ghosh2019robust, regatti2022byzantine,lu2022defense,cao2020fltrust,cao2019distributed,xie2019zeno}. Some of the theoretical guarantees are asymptotic, and almost all of them analyze the standard gradient descent (GD) algorithm or stochastic GD.
Typical solutions involve replacing the mean/sum of the gradients from the different nodes by a different robust statistic, such as geometric median (of means) \cite{dist_adversarial}, trimmed mean, coordinate-wise mean \cite{yin2018byzantine} or Krum  \cite{krum}.

One of the first non-asymptotic results for Byzantine attacks is \cite{dist_adversarial}.  This used the geometric median (GM) of means to replace the regular mean/sum of the partial gradients from each node. Under standard assumptions  (strong convexity, Lipschitz gradients, sub-exponential-ity of sample gradients, and an upper bound on the fraction of Byzantine nodes), it provided an exponentially decaying bound on the distance between the estimate at the $t$-th iteration and the unique global minimizer. 
In follow-up work \cite{yin2018byzantine}, the authors studied the coordinate-wise mean and the trimmed-mean estimators and developed guarantees for both convex and non-convex problems. Because these works used coordinate-wise estimators, their results needed smoothness and convexity along each dimension. This is a stronger, and sometimes impractical, assumption.
Another interesting series of works \cite{alistarh2018byzantine,allen2020byzantine} provides non-asymptotic guarantees for Byzantine resilient stochastic GD. This work develops an elaborate median based algorithm to detect and filter out the Byzantine nodes. The theoretical analysis assumes that the norm of the sample gradients is bounded by a constant that does not depend on the gradient dimension. This can be a restrictive assumption.
With this assumption, they are able to obtain sample complexity guarantees that do not depend on the signal dimension. These works also assume that the set of Byzantine nodes is the same for all GD iterations, while the work of \cite{dist_adversarial} allowed this set to vary at each iteration.

Most of the above works considered the homogeneous data setting; this means that the data that is i.i.d. (independent and identically distributed) across all nodes. More recent work has focused on heterogeneous distributions (data is independent across nodes but is not identically distributed) and proved results under upper bounds on the amount of heterogeneity  \cite{pillutla2019robust,data2021byzantine,li2019rsa,ghosh2019robust}. 
Other works rely on detection methods to handle heterogeneous gradients \cite{regatti2022byzantine,lu2022defense,cao2020fltrust,cao2019distributed,xie2019zeno}.  These assume the existence of a trustworthy root/validation dataset at the central server that is used for detecting the adversarial gradients.

\subsubsection{Work on robust PCA and subspace learning and tracking, and other robust estimation problems}
There is much work also on solving the robust PCA problem using the low rank plus sparse (L+S) \cite{rpca,rpca2,rpca_gd,rrpcp_dynrpca,rrpcp_jsait} or other models, on robust subspace learning \cite{novel_m_estimator,li2020nonconvex}, and on robust subspace tracking problems \cite{narayanamurthy2022federated,rekavandi2023robust,diakonikolas2017being}. The review article \cite{rrpcp_review} provides a comprehensive summary of the older work. In addition there is other related work that uses the median or vector medians for other types of outlier robust algorithms, e.g., \cite{li2020non}. However, there are two key differences between all these works and the problem that we study in this paper. 
(1) All of these works assume that the outlier or the attack is on the observed or measured data. In security literature, such attacks (in which only the data can be corrupted) are referred to as ``data-poisoning'' attacks. The algorithms from these works cannot be used to deal with Byzantine attacks which involve corruption of the (intermediate and/or final) algorithm estimates sent by some nodes.
(2) Secondly, almost all of these are designed for the centralized setting. A possible way to extend any of these ideas to the federated setting is for the nodes to share their raw data with the center and for the center to implement the same algorithm as that developed in these works. However, this would not be communication-efficient\footnote{One exception is the work of \cite{narayanamurthy2022federated} that considers the federated setting. However this has two important limitations: (i) it assumes that, for the initialization step, the data is outlier-free; (ii) and, it requires a much larger number of observed samples than what traditional LR matrix completion literature needs.}.
To distinguish from the L+S, or other, model-based robust PCA work, here, we use the term ``resilient'' to denote attack-resilience.%

\subsubsection{Work on robust statistics - robust mean and robust covariance estimation}
There is a large amount of existing work in the general robust statistics literature, most of it is on robust mean estimation, and some on robust covariance estimation, e.g., see \cite{diakonikolas2023algorithmic,diakonikolas2019recent,diakonikolas2019robust,lai2016agnostic,gm_banach}. None of these can be extended to solving our problem in a communication-efficient fashion and most of these also have much larger sample complexities. 
As an example, the work of Minsker \cite{gm_banach}  studies the geometric median (GM), which  is one well-known approach to compute a reliable estimate of a vector-valued quantity using multiple individual estimates of it when some of these estimates may be corrupted by outliers \cite{gm_banach,dist_adversarial}. In \cite[Corollary 4.3]{gm_banach}, Minsker shows the application of GM for ``robust PCA'' -
provably accurate robust/resilient covariance estimation followed by SVD on the robust covariance estimate to compute its top  $r$ singular vectors. We refer to this solution as {\em SVD-ResCovEst}.
This approach needs order $\qell\ge n^2$ samples. Moreover, it cannot be federated efficiently because it requires that each node $\ell$ shares $\D_\ell \D_\ell^\top$ with the center. 
This has a communication cost of order $n^2$. A similar discussion applies for the result of \cite[Theorem 4.35]{diakonikolas2023algorithmic} as well.

\subsubsection{Work on the LR Column-wise Sensing (LRCS) problem}
The LRCS problem, and its phaseless measurements' extension, LR phase retrieval, have been extensively studied in recent years \cite{lrpr_it,lrpr_best,lee2019neurips,lrpr_gdmin,lrpr_gdmin_2,lrpr_gdmin_mri_jp}, mostly in centralized settings. 
The work of \cite{lrpr_gdmin,lrpr_gdmin_2} introduced a fast and communication-efficient solution to attack-free federated LRCS, called alternating GD and minimization (altGDmin). AltGDmin is initialized using spectral initialization.

\subsubsection{Federated PCA and subspace learning; no attacks}
There is also somewhat related work on federated PCA and subspace learning that does not consider any attacks or other outliers, e.g.,  \cite{fed_pca,zhang2022turning,silva2019federated}.

\subsection{Our Contributions}
A natural way to make the SVD-ResCovEst approach communication-efficient is to borrow ideas from the sketching literature and  share $\bPhi_\ell \U_{any}$ for some (possibly random) $n \times r$ matrix $\U_{any}$. This idea is,  in fact, one iteration of the power method for computing the $r$-SVD of a matrix \cite{golub,distpca_review}. It can be converted into a provably correct solution by using the GM to modify the power method.
We refer to this solution as Resilient Power Method (ResPowMeth), and obtain a set of sufficient conditions for it to work. We show that this approach is both provably resilient to Byzantine attacks and communication efficient under certain restrictive assumptions on the accuracy of the individual nodes' partial covariance estimates, which translate into a very large sample complexity for PCA: for $n$-length data vectors, ResPowMeth works if $\qell \ge C n^2r^2$. 

Our first important contribution is a novel and well-motivated solution to Byzantine-resilient federated subspace estimation, and PCA, that is both communication-efficient and sample-efficient. We refer to this as ``Subspace-Median''. Its guarantee is provided in Theorem \ref{main_res} and Corollary \ref{th:ro_pca_A}.
We show how the Subspace Median can be used to provably solve two practically useful problems: (i) Byzantine resilient federated PCA, and (ii) the initialization step of Byzantine-resilient horizontal federated LRCS. 
For the PCA problem, we show that this works well with just $\qell \ge C nr $ samples.
We also develop Subspace Median-of-Means (MoM) extensions for both problems. These help improve the sample complexity at the cost of reduced Byzantine/outlier tolerance. For all these algorithms, Theorem \ref{main_res} helps prove sample, communication, and time complexity bounds for $\eps$-accurate subspace recovery. Extensive simulation experiments corroborate our theoretical results.

Our second important contribution is a provable communication-efficient and sample-efficient complete solution to horizontally federated LRCS.
For solving it, we develop a GM-based modification of the alternating GD and minimization (altGDmin) algorithm from earlier work \cite{lrpr_gdmin}.  We use Subspace Median and Subspace MoM to make its spectral initialization step Byzantine-resilient. For the complete algorithm, we can show that it obtains an $\eps$-accurate estimate of the unknown LR matrix using only order $nr^2\log(1/\eps)$ samples per node, and with a total communication cost of only order $nr\log(1/\eps)$ per node. Both costs are comparable to what basic altGDmin needs for solving this problem in the attack-free setting \cite{lrpr_gdmin,lrpr_gdmin_2}.

The overall approach that we develop for modifying the altGDmin algorithm (use Subspace-Median for initialization and GM of gradients for the GD step), and analyzing it, can be extended to make GD-based solutions to many other similar non-convex problems in federated settings Byzantine-resilient. Some examples include vertically federated LRCS, LR matrix completion, LR phase retrieval, LR plus sparse matrix recovery (robust PCA).
Our approach for analyzing Byzantine-resilient PCA is also extendable to solving PCA for approximately LR datasets, PCA for such datasets with missing entries (see Remark \ref{gens}), and also to  subspace tracking and robust subspace tracking. We describe these in Sec. \ref{extend}.

\subsection{Novelty of our algorithmic and proof techniques}\label{discuss}\label{novelty}
While both SVD and geometric median (GM) are well known in literature, we are not aware of any notions of ``median'' for subspaces. We cannot directly use the GM on the subspace basis matrices because these do not lie in a Euclidean space, e.g., $\U$, $-\U$ specify the same subspace even though $\norm{\U-(-\U)}_F=2\sqrt{r}\neq 0$. 
The design of Subspace Median relies on the fact that the Frobenius norm of the difference between two subspace projection matrices is within a constant factor of the subspace distance between their respective subspaces.  Its analysis also uses the fact that these projection matrices are bounded by $\sqrt{r}$ in Frobenius norm.
We use these facts and Lemma \ref{gm_new_1} (GM lemma for bounded inputs) to prove our key lemma, Lemma \ref{key_lemma}. This is combined with the Davis-Kahan $\sin \Theta$ theorem to prove Theorem \ref{main_res}. This result is likely to be widely applicable in making various other subspace recovery problems Byzantine resilient.

Our analysis of the AltGDmin iterations relies heavily on the lemmas proved in \cite{lrpr_gdmin} and the overall simplified proof approach developed in \cite{lrpr_gdmin_2}. However, we need to modify this approach to deal with the fact that we compute the geometric median of the gradients from the different nodes. The GM analysis provides bounds on Frobenius norms, and hence our analysis also uses the Frobenius norm subspace distance instead of the 2-norm one; see Lemma \ref{sd_pop}. At the same time it avoids the complicated proof approach (does not need to use the fundamental theorem of calculus) of \cite{lrpr_gdmin}. The main new step is the bound on the difference between the expected values of the gradients from two good nodes conditioned on past estimates and data\footnote{As explained earlier, the conditional expectations are different at the different nodes. These can be computed and bounded easily because we assume sample-splitting.}. See Lemma \ref{new_term_lem}.  This lemma is used along with Lemma \ref{gm_new_2} (our GM lemma for potentially unbounded inputs) to obtain Lemma \ref{err_bound_lemma}. This discussion will be clearer from the proof outline provided below Theorem \ref{gdmin_thm_gmom}.

\subsection{Organization}
We define the problems, the notation, and introduce the geometric median in the next section. Sec. \ref{res_sub_estim} develops the Subspace Median and Resilient Power Method (ResPowMeth) solutions, and provides their theoretical guarantees. Sec. \ref{res_pca} develops corollaries for the resilient PCA problem, compares the three approaches -- SVD-ResCovEst \cite{gm_banach}, ResPowMeth, and Subspace Median. A summary is provided in Table \ref{table_compare}. Subspace Median of Means is also developed here. Sec. \ref{lrccs_section} develops a complete altGDmin-based solution for resilient horizontally federated LRCS. Proofs for Sections \ref{res_sub_estim} and \ref{res_pca} are provided in Sec. \ref{proofs1}.  Simulation experiments are provided in Sec. \ref{sims}. We conclude in Sec. \ref{discuss}.

\section{Problem set-up, Notation, and Geometric Median Preliminaries}

\subsection{Problem setup} \label{probdef}
We study two interrelated problems stated below. We begin by stating the subspace estimation meta problem that occurs in both problems.
We consider a federated setting with $L$ nodes, with $L$ being a numerical constant, and assume the following.

\begin{assu}[Number of Byzantine nodes]
At most $\tau L$ of the $L$ total nodes are Byzantine, with $\tau \le 0.4$ (instead of 0.4, we can use any constant $c$ that is strictly less than 0.5 here). Denote the set of good (non-Byzantine) nodes by $\Jgood$. Equivalently, this means that $|\Jgood| > (1- \tau)L$. We define a Byzantine attack below in Sec. \ref{byzdefine}.
\label{byzassu}
\end{assu}

\subsubsection{Resilient Federated Subspace Estimation}
The goal is to reliably estimate the subspace spanned by the top $r$ singular vectors of an unknown symmetric $n \times n$ matrix, $\bPhi^*$. Denote the $n \times r$ matrix formed by these singular vectors by $\Ustar$.  Our goal is thus to estimate $\mathrm{span}(\Ustar)$.  Each node $\ell \in [L]$ observes, or can compute, a symmetric matrix $\bPhi_\ell$ which is {\em an} estimate of $\bPhi^*$.
Typically, the node observes an $n \times \totl$  data matrix $\D_\ell$ and computes $\bPhi_\ell:=\D_\ell \D_\ell^\top /\totl$. 
We use $\sigmaonestar\geq ...\geq \sigma_n^*$ to denote the singular values of $\bPhi^*$.

\subsubsection{Resilient Federated PCA} \label{respca_prob}
Given $q$ data vectors $\dstar_k \in \Re^n$, that are  mutually independent and identically distributed (i.i.d.), the goal is to find the $r$-dimensional principal subspace (span of top $r$ singular vectors) of their covariance matrix, which we will denote by $\Sigmastar$.  We can arrange the data vectors into an $n \times q$ matrix, $\D:=[\dstar_1, \dstar_2, \dots \dstar_q]$. We use $\sigma^*_j$ to denote the $j$-th singular value of $\Sigmastar$.
We assume that all $\dstar_k$s are i.i.d. {\em zero mean}, sub-Gaussian vectors, with covariance matrix $\Sigmastar$ and maximum sub-Gaussian norm $K\sqrt{\|\Sigmastar\|} = K \sqrt{\sigmax}$ \cite[Chap 2]{versh_book}. 
The data is vertically federated, this means that each node $\ell$ has  $q_\ell = \q=\frac{q}{L}$ $\dstar_k$'s. Denote the corresponding sub-matrix of $\D$ by $\D_\ell$. Thus, $\D = [\D_1, \dots, \D_\ell,\dots \D_L]$.
This problem is an instance of resilient federated subspace estimation with $\bPhi_\ell = \D_\ell \D_\ell^\top / \qell$.

\subsubsection{Resilient Horizontally-Federated Low Rank Column-wise Sensing (LRCS)}
LRCS involves recovering an $n \times q$ rank-$r$ matrix $\Xstar$ from compressive linear measurements of each column, i.e., from $\y_k:= \A_k \xstar_k$, $k \in [q]$, with $\y_k \in \Re^m$ with $m \ll n$, and $\A_k$ being $m \times n$ matrices which are i.i.d. random Gaussian (each entry is i.i.d. standard Gaussian)  \cite{lrpr_gdmin}. We treat $\Xstar$ as a deterministic unknown.
Here and below $[q]:=\{1,2,\dots,q\}$. 
Let $\Y := [\y_1, \y_2, \dots, \y_q]$. Horizontal federation means that row sub-matrices of $\Y$ are observed at the different nodes. To be precise, node $\ell$ observes an $\m \times q$ rows sub-matrix of $\Y$ denoted $\Y_\ell$ with $\m = m/L$.   We assume that node $\ell$ has access to $\Y_\ell$ and $\{ (A_k)_\ell, k \in [q]\}$.  Denote the set of indices of the rows available at node $\ell$ by $\cS_\ell$. Then, $(\A_k)_\ell:= \I_{\cS_\ell}^\top \A_k$ is of size $\m \times n$ and $\Y_\ell:= \I_{\cS_\ell}^\top \Y$ is of size $\m \times q$ with $\m = m/L$, and with $(\y_k)_\ell: = (\A_k)_\ell \xstar_k$ for all $k \in [q]$
Three important applications that can be modeled as instances of LRCS are accelerated dynamic MRI reconstruction \cite{dyn_mri1,lrpr_gdmin_mri_jp}, federated sketching \cite{lrpr_gdmin_mri_jp,lee2019neurips}, and multi-task representation learning and few shot learning \cite{du2020few, collins2021exploiting,altgdmin_icml}. In the representation learning problem, horizontal federation corresponds to the setting where the $\ell$-th subset of training data for the $q$ correlated linear regression tasks is observed at node $\ell$. Few shot learning uses this learned representation (column span of $\Xstar$) for learning the regression coefficients using very few training data points (this problem is also referred to as online subspace tracking in \cite{lrpr_gdmin_mri_jp}).
In multi-coil dynamic MRI, $L$ is the number of MRI scanners, each of which observes a differently weighted subset of measurements of the human organ's image sequence. Scanners can be prone to security threats if they are connected to the internet. 
As we will see later, the initialization step for solving LRCS using an iterative algorithm  can be interpreted as an instance of resilient federated subspace estimation.

The reason we consider vertical federation for PCA but horizontal for LRCS is because these are the settings in which the data on the different nodes is i.i.d. in each case. In case of vertically federated PCA, $\D_\ell$'s are i.i.d. If we consider horizontal federation for PCA, then this is no longer true (unless we assume $\Sigmastar$ is block diagonal). For LRCS, the opposite holds because different entries of a given $\y_k$ are i.i.d.; but the different $\y_k$'s are not identically distributed. 
Guaranteeing Byzantine resilience without extra assumptions requires the different nodes' data be i.i.d. or i.i.d.-like (this means that it should be possible to obtain a uniform bound on the errors between the individual nodes' outputs and the quantity of interest each time the node output is shared with the center).
As we explain later, it is possible to use ideas similar to the ones introduced here to also solve vertically federated LRCS, but that will need extra assumptions that ensure bounded heterogeneity.

\subsubsection{Byzantine attack definition}\label{byzdefine}
{\em We use the terms Byzantine node/adversary/attack almost interchangeably. The output of a Byzantine node or adversary is the Byzantine attack.  Byzantine nodes are often also referred to as ``bad'' nodes and non-Byzantine ones as ``good'' nodes. 
}
The Byzantine attack has not been clearly mathematically defined in past work \cite{guerraoui2018hidden, chen2018draco, krum, alistarh2018byzantine, yin2018byzantine}, although there are definitions inspired by \cite{byz_general}. The following definition, taken from \cite{guerraoui2018hidden}, is the most precise one we can find.
\begin{definition}
The Byzantine adversary is an entity which controls the outputs of some of the $L$ worker nodes. 
It is omniscient, in the sense that it has a perfect knowledge of the system state at any time, i.e., it knows (i) the full state of the center (data and algorithm, including all algorithm parameters), and (ii) the full state of every node (data and algorithm, including all algorithm parameters). Different Byzantine adversaries can also collude. However, they are not omnipotent: they cannot modify the outputs of the other (non-Byzantine) nodes or of the center, or delay communication. 
\end{definition}
In our setting, this means the following. Let $\nabla_{byz}$ denote the set of outputs of all the Byzantine nodes. Then $\nabla_{byz}= g_{byz} (\{\mathrm{Data}_\ell\}_{\ell=1}^{L},\mathcal{A}))$  where $\mathcal{A}$ denotes the true algorithm being implemented at each of the non-Byzantine (good) nodes and at the center along with all its parameters; $g_{byz}(.)$ is a function that can be jointly designed by all the Byzantine nodes; and $\mathrm{Data}_\ell$ is the data observed at node $\ell$: it is $\bPhi_\ell$ or $\D_\ell$ (in case of PCA), or $\Y_\ell, (\A_k)_\ell, k\in [q]$ (in case of LRCS).

\subsection{Notation}
We use $\|.\|_F$ to denote the Frobenius norm and $\|.\|$ without a subscript to denote the (induced) $l_2$ norm (often called the operator norm or spectral norm);  $^\top$ denotes matrix or vector transpose; $|\bm{z}|$ for a vector $\bm{z}$ denotes element-wise absolute values; $\I_n$ (or sometimes just $\I$) denotes the $n \times n$ identity matrix, and $\e_k$ denotes its $k$-th column ($k$-th canonical basis vector); and $\M^\dag = (\M^\top \M)^{-1} \M^\top$. We use $\indic_{\{a\leq b\}}$ to denote the indicator function that returns $1$ if $a\leq b$ otherwise $0$.

We say $\U$ is a {\em basis matrix} if it is a tall matrix with mutually orthonormal columns; we use this to denote the subspace spanned by its columns. For a basis matrix $\U$, the projection matrix for projecting onto $\mathrm{span}(\U)$ (the subspace spanned by the columns of $\U$) is denoted $\P_\U:= \U \U^\top$ while that for projecting orthogonal to $\mathrm{span}(\U)$ is denoted $\P_{\U,\perp}:=\I - \U \U^\top$
``$r$-SVD'' to refer to the top $r$ left singular vectors (singular vectors corresponding to the $r$ largest singular values) of a matrix. 
For basis matrices, $\U_1, \U_2$, we use $\SD_F(\U_1, \U_2): = \|(\I - \U_1 \U_1{}^\top)\U_2\|_F$ as the default Subspace Distance (SD) measure between the subspaces spanned by the two matrices. In some places, we also use  $\SD_2(\U_1, \U_2) = \|(\I - \U_1 \U_1{}^\top)\U_2\|$. If both matrices have $r$ columns (denote $r$-dimensional subspaces), then $\SD_F(\U_1, \U_2) \le \sqrt{r} \SD_2(\U_1, \U_2)$.  

We use $QR(\tilde\U)$ to denote the orthonormalization of the columns of $\tilde\U$ by using QR decomposition.  For a matrix $\M$, $vec(\M)$ vectorizes it.

{\em We reuse the letters $c,C$ to denote different numerical constants in each use with the convention that $c < 1$ and $C \ge 1$.} Also the notation $a \lesssim b$ means $a \le C b$.

For $a,b$ in $(0,1)$, we use $\psi(a, b):=(1-a)\log \frac{1-a}{1-b}+a\log \frac{a}{b}$ to denote the binary KL divergence.
When computing a median of means estimator, one splits the $L$ node indices into $\tL$ mini-batches so that each mini-batch contains $\rho = L/\tL$ indices. 
For the $\ell$-th node in the $\vartheta$-th mini-batch we use the short form notation ${(\vartheta,\ell)} = { (\vartheta - 1)\rho + \ell}$, for $\ell\in[\rho]$.

Recall that Byzantine nodes are often also referred to as ``bad'' nodes and non-Byzantine ones as ``good'' nodes. We use $\Jgood \subseteq [L]$ to denote the set of non-Byzantine/good nodes and $\Jgood^c$ denotes the set of Byzantine nodes. Using Assumption \ref{byzassu}, $|\Jgood| = L - \tau L = (1-\tau)L$.

\subsection{Geometric median (GM)}\label{GM_basics}
The geometric median (GM) is one well-known approach to compute a reliable estimate of a vector-valued quantity using multiple individual estimates of it when some of these estimates may be corrupted by outliers \cite{gm_banach,dist_adversarial}. For $L$ data vectors $\{\z_1, \z_2, \dots, \z_L\}$, with each $\z_\ell \in \Re^d$, this is defined as
\[
\z_\gm^*:= GM\{\z_1,\z_2,...,\z_L\} = \arg \min_{\z\in \Re^d} \sum_{\ell=1}^{L}\left\| \z-\z_\ell \right\|
\]
This cannot be computed in closed form. Iterative algorithms exist to solve it approximately.
When we say $\z_\gm$ is a $(1+ \eps_0)$ approximate GM, for an $0< \eps_0 < 1$ we mean that
\begin{align}
\sum_{\ell=1}^{L}\left\| \z_\gm -\z_\ell \right\| \le (1+\eps_0) \sum_{\ell=1}^{L}\left\| \z_\gm^* -\z_\ell \right\| = (1+\eps_0)  \min_{\z\in \Re^d} \sum_{\ell=1}^{L}\left\| \z-\z_\ell \right\|
\label{approxgmdef}
\end{align}
There are two popular iterative solutions for computing the approximate GM. The most commonly used one in practice, Weiszfeld's algorithm \cite{weiszfeld_original,beck2015weiszfeld}, does not come with a useful iteration complexity guarantee.
The recent work of \cite{cohen2016geometric} introduced a nearly linear-time algorithm for provably computing the approximate GM, with high probability.
We provide \cite[Algorithm 1]{cohen2016geometric} in Appendix \ref{gm_append}. We state its guarantee next. All theoretical results in our work use this result.

\begin{claim}[Theorem 1 \cite{cohen2016geometric}]\label{cohen_1}
Pick an accuracy level $0 < \eps_0 < 1$.
Consider \cite[Algorithm 1]{cohen2016geometric} with input $\{\z_1,\z_2,...,\z_L\}$ and using number of iterations, $T_\gm =  C \log(\frac{L}{\eps_0})$. 
With probability at least $1-c_\approxgm$ (where $c_\approxgm<1$ is a numerical constant, e.g., 0.1), the algorithm computes $\z_\gm$ that satisfies \eqref{approxgmdef}.
Its per iteration complexity is $C Ld\log^2(\frac{L}{\eps_0})$ and total time complexity is $O(Ld\log^{3}(\frac{L}{\eps_0}))$.
\end{claim}

The use of the above result allows us to bound the iteration complexity of all our algorithms. This, in turns, allows us to get a bound on the total communication cost and the total time cost.
Although it has a simple guarantee, the algorithm \cite[Algorithm 1]{cohen2016geometric} itself is quite complicated. The authors of \cite{cohen2016geometric} have not shown any experimental results with it. To our best knowledge, nor have any other authors in follow-up work that cites it. The algorithm used in practice for approximating the GM is the Weiszfeld's algorithm initialized using the average of the $\z_\ell$'s \cite{weiszfeld_original}. This is an iteratively re-weighted least squares type algorithm.
 We provide both algorithms in Appendix \ref{gm_append}.

\begin{algorithm}[t]
	\caption{Subspace Median}\label{NEW_PCA_1}
\hspace*{\algorithmicindent} \textbf{Input} $\D_\ell$, $\ell \in [L]$; or $\bPhi_\ell$, $\ell \in [L]$. \\
    \hspace*{\algorithmicindent} \textbf{Parameters} $T_\gm$, $T_{pow}$
	\begin{algorithmic}[1]
		\STATE \underline{\textbf{Nodes $\ell=1,...,L$}}
		\STATE Compute  top $r$ singular vectors, $\Uhat_\ell$, of $\D_\ell$ (equivalently of $\bPhi_\ell := \D_\ell \D_\ell^\top$).
\\ (Can use power method with $T_{pow}$ iterations)
			
		\STATE \underline{\textbf{Central Server}}
\STATE Orthonormalize: $\U_\ell \leftarrow QR(\Uhat_\ell)$, $\ell \in [L]$
		\STATE Compute Projection Matrices: $\pul \leftarrow \U_\ell \U_\ell^\top$, $\ell \in [L]$
		\STATE Compute their GM: $\P_\gm \leftarrow approxGM\{\pul, \ \ell \in [L]\}$
\\ (Use \cite[Algorithm 1]{cohen2016geometric} with parameter $T_\gm$).
		\STATE Find $\ell_{best} = \arg \min_\ell \norm{\pul - \P_\gm}_F$
		\STATE Output $\U_{out} = \U_{\ell_{best}}$
	\end{algorithmic}
\end{algorithm}

\section{Resilient Federated Subspace Estimation} \label{res_sub_estim}

\subsection{Proposed solution: Subspace Median}
Recall that our goal is to obtain a reliably accurate estimate of $\mathrm{span}(\Ustar)$, which is an $r$-dimensional subspace in $\Re^n$, when each node computes an estimate $\U_\ell$ of it by computing the top $r$ singular vectors of $\bPhi_\ell$.. Some nodes can be Byzantine (Assumption \ref{byzassu}).  
We develop a solution approach that relies on the geometric median (GM). Notice from Sec. \ref{GM_basics} that the GM is defined for quantities whose distance can be measured using the vector $l_2$ norm (equivalently, matrix Frobenius norm). Our solution adapts the GM to use it for subspaces by using the fact that the Frobenius norm between the projection matrices of two subspaces is another measure of subspace distance: $\|\P_\U - \P_{\Ustar}\|_F = \sqrt{2}\SD_F(\U,\Ustar) $ \cite[Lemma 2.5]{spectral_dist}.

Our proposed algorithm, which we refer to as ``Subspace Median'', relies on this fact. It proceeds as follows.
Each node computes the top $r$ singular vectors of its matrix $\bPhi_\ell$, denoted $\Uhat_\ell$, and sends these to the center. If node $\ell$ is good (non-Byzantine), then $\Uhat_\ell$ already has orthonormal columns; however if the node is Byzantine, then it is not. The center first orthonormalizes the columns of all the received $\Uhat_\ell$'s using QR. This ensures that all the $\U_\ell$'s have orthonormal columns.
It then computes the projection matrices $\P_{\U_\ell}:= \U_\ell \U_\ell^\top$, $\ell \in [L]$, followed by vectorizing them, computing their GM, and converting the GM back to an $n \times n$ matrix. Denote this by $\P_\gm$. 
Finally, the center finds the $\ell$ for which $\P_{\U_\ell}$ is closest to $\P_\gm$ in Frobenius norm and outputs the corresponding $\U_\ell$.

We should mention that this last step can also be replaced by finding the top $r$ singular vectors of $\P_\gm$. However, doing this requires time of order $n^2 r \log(1/\eps)$ while finding the closest $\P_{\U_\ell}$ only needs time of order $\max(n^2,L \log L)$.

Subspace Median is summarized in Algorithm \ref{NEW_PCA_1}.  We can prove the following for it. 
\begin{lemma}[Subspace-Median]\label{key_lemma}
For a $\delta > 0$, consider Algorithm \ref{NEW_PCA_1} with $T_\gm = C\log \left(\frac{Lr}{\delta}\right)$. Assume that Assumption \ref{byzassu} holds. 
	Assume that, for at least $(1-\tau)L$ nodes, the following holds:  
	\begin{equation}
		\Pr \left(\SD_F(\Ustar,\U_{\ell})\leq\delta \right) \geq 1-p.  \notag
	\end{equation}
Then,  w.p. at least $1- c_\approxgm - \exp(-L \psi(0.4-\tau,p) )$,   	
	\begin{equation}
		\SD_F(\Ustar,\U_{out})\leq 23\delta .  \notag
	\end{equation}
Here $\psi(a, b):=(1-a)\log \frac{1-a}{1-b}+a\log \frac{a}{b}$ for $0 < a,b < 1$ is the binary KL divergence, and $c_\approxgm$ is the numerical constant from Claim \ref{cohen_1}.
\end{lemma}

\begin{proof} See Sec. \ref{proofs1}. \end{proof}

Combining this lemma with the Davis-Khan $\sin \Theta$ theorem (bounds the distance between the principal subspaces of two symmetric matrices) \cite{davis_kahan} and a guarantee for the power method  \cite{npm_hardt}, we can prove the following theorem.

\begin{remark}
We specify the power method just to have one algorithm for computing the top $r$ singular vectors of a matrix for which we can specify the time compleixty. 
It can be replaced by any other algorithm and our overall result remains the same.
\end{remark}

\begin{theorem}[Subspace-Median guarantee] \label{th:ro_pca}\label{main_res}  
Pick an $\eps<1$.
Assume that Assumption \ref{byzassu} holds and that $\sigmarstar-\sigmarplusstar \ge \Delta$ for a $\Delta > 0$.
Assume also that, for at least $(1-\tau)L$ node outputs, the following holds, for a $p>0$.
\[
\Pr\left\{\norm{\bPhi_\ell-\bPhi^*}\leq  \frac{\epsilon}{92\sqrt{r}}  \Delta \right\}\geq 1-p.
\]
Consider Algorithm \ref{NEW_PCA_1} with $T_\gm= C \log \left(\frac{L r}{\eps}  \right)$.
\ben
\item  Assume use of exact SVD at the nodes. 
Then, w.p. at least $1-c_\approxgm-\exp(-L\psi(0.4-\tau,p))$,
\[
\SD_F(\U_{out},\Ustar)\leq \eps
\]

\item Assume that the power method with $T_{pow}$ iterations is used for the SVD step. If $T_{pow} = C  \frac{\sigmarstar}{\Delta} \log ( \frac{n}{\eps}) $, then the above conclusion holds w.p. at least  $1-c_\approxgm-\exp(-L\psi(0.4-\tau,p + \frac{1}{n^{10}} ))$.

The communication cost is $nr$ per node.
The computational cost at the center is order $n^2 L  \log^3\left(\frac{L r}{\eps} \right)$.
 The computational cost at any node (when using power method) is order $ n \qell r  T_{pow} =  n\qell r \frac{\sigmarstar}{\Delta} \log (\frac{n}{\epsilon})$. 
\een
\end{theorem}

\begin{proof} See Sec. \ref{proofs1}. \end{proof}

The assumption $\sigmarstar-\sigmarplusstar \ge \Delta$ (singular value gap) is needed for ensuring that the span of $\U_\ell$ computed at any good node is an accurate estimate of the span of $\Ustar$. It also decides the time complexity of the computation ($\Delta$ appears in the required number of power method iterations.

\subsection{Alternate solution 1: SVD on Resilient Covariance Estimation (SVD-ResCovEst)}
SVD-ResCovEst is the solution studied by Minsker \cite{gm_banach} and  described earlier. It involves computing the GM of (vectorized) $\bPhi_\ell$s, followed by obtaining the principal subspace ($r$-SVD) of the GM matrix.
In a federated setting, this is communication inefficient since it requires that each node $\ell$ either shares its raw data $\D_\ell$ with the center (this is a matrix of size $n \times q_\ell$), or,  that it shares $\bPhi_\ell = \D_\ell \D_\ell^\top/\qell$ (this is of size $n \times n$). For PCA, as we explain in the next section, this is also sample inefficient; it requires $q_\ell \ge n^2/\eps^2$. See Remark \ref{svdcovest_remark}.

\subsection{Alternate solution 2: Resilient Power Method (ResPowMeth)}
A natural way to make the SVD-ResCovEst approach communication-efficient is to borrow ideas from the sketching literature and  share $\bPhi_\ell \U_{any}$ for some (possibly random) $n \times r$ matrix $\U_{any}$. This idea is,  in fact, one iteration of the power method for computing the $r$-SVD of a matrix \cite{golub,distpca_review}. It can be converted into a provably correct solution by using the GM to modify the power method. This starts with a random Gaussian initialization, $\U_{rand}$, and implements the iteration: $\U \leftarrow \mathrm{QR}(\sum_\ell \bPhi_\ell \U)$. In our GM based modification, we replace the summation by the GM.
We refer to this solution as {\em Resilient Power Method (ResPowMeth)}, and summarize it in Algorithm \ref{gm_pm}. 
As we show next, ResPowMeth works with high probability (w.h.p.) if all the $\bPhi_\ell$'s are very accurate estimates of $\bPhi^*$.
The reason it needs to make the above assumption is because it computes the GM of the node outputs $\bPhi_\ell \U$ at each iteration including the first one. At the first iteration, $\U_0$ is a randomly generated matrix and thus, w.h.p., this is a bad approximation of the desired subspace $\mathrm{span}(\Ustar)$. Consequently, the same is true for the column span of $\tilde\U_\ell^+=\bPhi_\ell \U_0$. To understand this easily, suppose $\U_0$ is almost orthogonal to $\Ustar$, i.e., $\U_0^\top \Ustar \approx \bm{0}$. Then the span of $\tilde\U_\ell^+$ will be almost orthonormal to that of $\Ustar$.
Thus, unless all the $\Phi_\ell$s are very similar, the column spans of the different $\tilde\U_\ell^+$s will not be close. As a result, the GM of their projection matrices will not be able to distinguish between the good and Byzantine ones. There is a good chance that it approximates the subspace of the Byzantine one(s). This then means that the updated $\U$ is also a bad approximation of $\mathrm{span}(\Ustar)$. The same idea repeats at the second iteration. Thus, with significant probability, the subspace estimates do not improve over iterations.
This intuition is captured in the guarantee provided next. It becomes clearer in the direct one-step analysis that we provide in Appendix \ref{direct_respow}.
\begin{algorithm}[t]
	\caption{Resilient Power Method (ResPowMeth)}\label{gm_pm}
	\hspace*{\algorithmicindent} \textbf{Parameters} $T_{pow}$, $T_\gm$, $\thresh$
	\begin{algorithmic}[1]
		\STATE \underline{\textbf{Central Server}} Randomly Initialize $\U_{rand}$ with i.i.d standard Gaussian entries. Set $\U_0 = \U_{rand}$.
		\FOR {$t\in T_{pow}$}
		\STATE \underline{\textbf{Nodes $\ell=1,...,L$}}
		\STATE Compute $\bPhi_\ell \U_{t-1}$
		\STATE \underline{\textbf{Central Server}}
		\STATE 
		$GM\leftarrow approxGM\left( \{ vec(\bPhi_\ell\U_{t-1}), \ \ell \in [L] \} \setminus \{\ell: \|\bPhi_\ell\U_{t-1}\|_F > \thresh\} \right)$
		\\ (Use \cite[Algorithm 1]{cohen2016geometric} with $T_\gm$ iterations on the set of $\bPhi_\ell\U_{t-1}$s whose Frobenius norm is below $\thresh$)
		\STATE Orthonormalize: using QR $GM \stackrel{QR}{=} \hat{\U} \bm{R}$
		\STATE Return $\U_{t}\leftarrow \hat{\U}$
		\ENDFOR
		\STATE Output $\U_{out} \leftarrow \U_{T_{pow}}$
	\end{algorithmic}
\end{algorithm}

\begin{theorem}[ResPowMeth guarantee] \label{th:ro_pm}
Assume that Assumption \ref{byzassu} holds and that $\sigmarstar-\sigmarplusstar \ge \Delta$ for a $\Delta > 0$.
Consider ResPowMeth  (Algorithm \ref{gm_pm}) with $T_{pow} = C \frac{\sigmarstar}{\Delta} \log (\frac{n}{\epsilon})$,
$T_\gm =  \log \left(\frac{L nr }{\eps}\right) $, and $\thresh = 1.1 \sigmaonestar \sqrt{r}$.
Suppose, for at least $(1-\tau)L$ node outputs, the following holds
\[
\Pr\left\{\norm{\bPhi_\ell-\bPhi^*}\leq \frac{1}{70}\min\left( \frac{\eps}{\sqrt{r}}, \frac{1}{2\sqrt{n} r}\right) \Delta \right\}\geq 1-p.
\]
Then w.p. at least $1 - c_\approxgm - c - L p -\exp(-L\psi(0.4-\tau,p)) $ 
$
\SD_F(\U_{out},\Ustar)\leq \epsilon.
$
The communication cost is $nr T_{pow} = C  nr \frac{\sigmarstar}{\Delta} \log (\frac{n}{\epsilon})$ per node.
 The computational cost at the center is $nrL \log^3\left(\frac{L nr }{\eps}  \right)  \cdot T_{pow}= nrL \frac{\sigmarstar}{\Delta} \log^3\left(\frac{L nr }{\eps}  \right) \log (\frac{n}{\eps})$.
 The computational cost at any node is $n \qell r T_{pow} = n \qell r \frac{\sigmarstar}{\Delta} \log (\frac{n}{\epsilon})$. 
\end{theorem}

\begin{proof} See Sec. \ref{proofs1}. A second more illustrative proof is provided in Appendix \ref{direct_respow}. \end{proof}

Observe that this result assumes $\|\bPhi_\ell - \bPhi^*\| \lesssim \min(\eps/\sqrt{r}, 1/(\sqrt{n} \ r)) \cdot \Delta$. 
The $1/\sqrt{n}r$ factor makes this a very stringent requirement, e.g., even to get an $\eps = 0.5$ accurate subspace estimate, we need $\|\bPhi_\ell - \bPhi^*\| \lesssim \Delta / \sqrt{n} \  r$.
On the other hand, the Subspace Median guarantee only assumes $\|\bPhi_\ell - \bPhi^*\| \lesssim (\eps/\sqrt{r})\Delta$.
As we will see in the next section, this translates into a much better sample complexity for PCA for Subspace Median than for ResPowMeth.

\begin{table*}[t]
	\begin{center}
\resizebox{0.95\linewidth}{!}{
		\begin{tabular}{|l|l|l|l|l|}
			\hline
			\textbf{Methods}$\rightarrow$ & \textbf{SVD-ResCovEst} & \textbf{ResPowMeth} &  \textbf{SubsMed}& \textbf{Basic PowMeth, no attack}   \\
& \cite[Cor 4.3]{gm_banach}, \cite[Thm 4.35]{diakonikolas2023algorithmic}  & \textbf{(Proposed modific of \cite[Cor 4.3]{gm_banach})} & \textbf{(Proposed)}  & {(Baseline)} \\
			\hline
			\textbf{Sample Comp for PCA} %
			      &$\frac{n^2}{\epsilon^2} \cdot L$ & $\max\left(n^2 r^2 ,\frac{n}{\epsilon^2}\right) \cdot L$ & $\frac{nr}{\epsilon^2} \cdot L $ & $\frac{n}{\epsilon^2}$\\
$q \geq  C  K^4 \frac{\sigmaonestar{}^2}{\Delta^2} \times$  &&&& \\
&&&&\\ \hline
			\textbf{Communic Cost}   &$n^2$ &   $nr \frac{\sigmarstar}{\Delta} \log (\frac{n}{\epsilon})$ & $nr$  & $nr \frac{\sigmarstar}{\Delta} \log (\frac{n}{\epsilon})$ \\
&&&&\\ \hline
\textbf{Compute Cost - node}    & $n^2\qell$ &   $ n \qell r \frac{\sigmarstar}{\Delta} \log (\frac{n}{\epsilon})$ &  $ n\qell r \frac{\sigmarstar}{\Delta} \log (\frac{n}{\epsilon})$ & $ n \qell r \frac{\sigmarstar}{\Delta} \log (\frac{n}{\epsilon})$ \\
&&&&\\			\hline
			\textbf{Compute Cost - center}    & $n^2 L  \log^3\left(\frac{L r}{\eps}\right)$ &   $\frac{\sigmarstar}{\Delta} nrL  \log (\frac{n}{\eps}) \log^3\left(\frac{L n}{\eps} \right) $ &  $n^2 L  \log^3\left(\frac{L r}{\eps}  \right)$& $\frac{\sigmarstar}{\Delta} nrL  \log (\frac{n}{\eps})$ \\
&&&&\\ \hline
		\end{tabular}
}
	\end{center}
\vspace{-0.1in}
\caption{\sl{
Comparing Subspace Median (SubsMed) with SVD-ResCovEst and Resilient Power Method (ResPowMeth) and with the basic power method for a no-attack setting. We obtain complexities for guaranteeing $\SD_2(\U,\Ustar) \le \eps$. SubsMed and ResPowMeth only bound $\SD_F(\U,\Ustar)$ and this is why the sample complexities for these contain an extra factor of $r$. This table summarizes the results of Corollary \ref{th:ro_pca_A} and the two remarks below it.
}
}
\vspace{-0.1in}
\label{table_compare}
\end{table*}

\section{Application 1: Resilient federated PCA} \label{res_pca}

Recall from Sec. \ref{probdef} that our goal is to reliably estimate the principal subspace of the unknown data covariance matrix $\Sigmastar$. Node $\ell$ has access to a subset of $\qell$ data vectors $\dstar_k$ arranged as columns of an $n \times \qell$ matrix $\D_\ell$.

\subsection{Subspace-Median (SubsMed) for resilient PCA}
Using its data, each node can compute the empirical covariance matrix $\hat\Sig_\ell:= \D_\ell \D_\ell^\top / \q$. This is an estimate of the true one, $\Sigmastar$.
This allows us to use Algorithm \ref{NEW_PCA_1} applied to $\D_\ell$ or $\hat\Sig_\ell$ to obtain a Byzantine resilient PCA solution, and use Theorem \ref{th:ro_pca} to analyze it.  
The sample complexity needed to get the desired bound on $\|\hat\Sig_\ell - \Sigmastar\|$ w.h.p. is obtained using \cite[Theorem 4.7.1]{versh_book}. Combining these two results, we can prove the following.

\begin{cor}[Subspace Median for PCA] \label{th:ro_pca_A}
Consider the PCA problem as defined in Sec. \ref{respca_prob}.
Assume that Assumption \ref{byzassu} holds and that $\sigmarstar-\sigmarplusstar \ge \Delta$ for a $\Delta > 0$. 
Consider Algorithm \ref{NEW_PCA_1} (SubsMed) with input $\bPhi = \D_\ell \D_\ell^\top / \qell$, and parameters as set in Theorem \ref{th:ro_pca}. If
\[
\qell:= \frac{q}{L} \geq C  K^4 \frac{\sigmaonestar{}^2}{\Delta^2} \cdot \frac{nr}{\epsilon^2},
\]
then, w.p. at least $1 - c_\approxgm -\exp(-L\psi(0.4-\tau, 2\exp(-n)+ n^{-10}))$,
$\SD_F(\U_{out},\Ustar)\leq \epsilon$.
\end{cor}

\begin{proof}
We prove it in Sec. \ref{proof_ro_pca_A}. It is  an immediate corollary of Theorem \ref{th:ro_pca} and \cite[Theorem 4.7.1]{versh_book}. 
\end{proof}

\begin{remark}[ResPowMeth for PCA]\label{respowmeth_remark}
In the setting of Corollary \ref{th:ro_pca_A}, consider Algorithm \ref{gm_pm} (ResPowMeth). If $\qell\geq C  K^4 \frac{\sigmax^2}{\Delta^2}  \cdot \max\left(\frac{n}{\epsilon^2},n^2 r^2 \right)$,
then $\SD_F(\U_{out},\Ustar)\leq \epsilon$. 
\end{remark}

\begin{remark}[SVD-CovEst for PCA]\label{svdcovest_remark}
In the setting of Corollary \ref{th:ro_pca_A}, consider SVD-ResCovEst (SVD on GM of nodes' covariance matrix estimates) studied in \cite[Corollary 4.3]{gm_banach}.
By using \cite[Corollary 4.3]{gm_banach}, and using the fact that, $\E\norm{\dstar_k}^4 - \text{trace}(\Sigma^*{}^2)\leq Cn^2K^4 \sigmax^2$ under the sub-Gaussian assumption, we can conclude the following:
If $q_\ell \ge CK^4 \frac{\sigmax^2}{\Delta^2} n^2 /\eps^2$, then, $\SD_F(\U_{out},\Ustar)\leq \epsilon$ with constant probability. 

The reason this needs $q_\ell$ of order $n^2$ is because it first obtains a resilient estimate of the entire $n \times n$ covariance matrix, followed by $r$-SVD on it. For resilient estimation, it needs to use the Frobenius norm as the error measure. 
The robust estimator studied in \cite[Theorem 4.35]{diakonikolas2023algorithmic} uses a different algorithm, but this also needs order $n^2 /\eps^2$ sample complexity and order $n^2$ communication complexity.
\end{remark}

 \begin{remark}[Generalizations of Theorem \ref{th:ro_pca}]\label{gens}
(1) Theorem \ref{th:ro_pca} also holds if the $\d_k$'s are not i.i.d., but are zero mean, independent, sub-Gaussian, and with covariance matrices that are of the form  $\E[\d_k \d_k^\top] = \Ustar \bm{S}_k \Ustar{}^\top$ with all $\bm{S}_k$'s being such that 
their $r$-th singular value gap is at least $\Delta$.

(2) We can also combine Theorem \ref{th:ro_pca} with the sample complexity bound for estimating approximately LR covariance matrices given in \cite[Corollary 5.52 and Remark 5.53]{vershynin2010introduction} to show that, in this case, a much lower sample complexity suffices. Suppose  $\dstar_k$ are i.i.d., {\em zero mean}, sub-Gaussian, have covariance matrix $\Sigmastar$, and are bounded with $\|\dstar_k\|^2 \le K^2 trace(\Sigmastar)$ and $trace(\Sigmastar) = r_0 \sigmax$ with stable rank $r_0 \ll n$ (approximately LR matrix). Then,  if $ \qell \ge C K^4 \frac{\sigmax^2}{\Delta^2} (\max(r_0, r)^2 \log n) / \eps^2$, then $\SD_F(\U_{out},\Ustar)\leq \epsilon$. Here $r_0$ is the stable rank. 

(3) We can also do the above for PCA with missing data by combining with \cite[Theorem 3.22]{spectral_init_review}.
\end{remark}

\subsection{Subspace Median-of-Means (Subspace MoM)} \label{subspace_gmom} 
As is well known, the use of median of means (MoM), instead of median, improves (reduces) the sample complexity needed to achieve a certain recovery error, but tolerates a smaller fraction of Byzantine nodes. It is thus useful in settings where the number of bad nodes is small. 
We show next how to obtain a communication-efficient and private GMoM estimator for federated PCA. Pick an integer $\tL \le L$.
In order to implement the ``mean'' step, we need to combine samples from $\rho = L/\tL$ nodes, i.e., we need to find the $r$-SVD of 
matrices $\D_{(\vartheta)} = [\D_{(\vartheta,1)}, \D_{(\vartheta,2)}, \dots, \D_{(\vartheta,\rho)}]$, for all $\vartheta \in [\tL]$. Recall that $_{(\vartheta,\ell)} = _{ (\vartheta - 1)\rho + \ell}$.
This needs to be done without sharing the entire matrix $\D_{(\vartheta,1)}$.
We do this by implementing $\tL$ different federated power methods, each of which combines samples from a different minibatch of $\rho$ nodes. The output of this step will be $\tL$ subspace estimates $\U_{(\vartheta)}$, $\vartheta \in [\tL]$. These serve as inputs to the Subspace-Median algorithm to obtain the final Subspace-MoM estimator.
We summarize the complete the algorithm in Algorithm \ref{alg:sub_gmom_pm}. We should mention that $\tL=L$ is the subspace median special case.

As long as the same set of $\tau L$ nodes are Byzantine for all the power method iterations, we can prove the following.
\begin{cor}\label{th:sub_gmom} 
	Consider Algorithm \ref{alg:sub_gmom_pm} and the setting of Corollary \ref{th:ro_pca_A}. 
Assume that the set of Byzantine nodes remains fixed for all iterations in this algorithm and the size of this set is at most $\tau L$ with  $\tau< 0.4\tL/L$. 
	If $$\frac{q}{L} = \q 	\geq  CK^4 \frac{\sigmaonestar{}^2}{\Delta^2} \frac{nr}{\epsilon^2} \cdot \frac{\tL}{L}$$ then, the conclusion of  Corollary \ref{th:ro_pca_A} holds. 
	The communication cost is $T_{pow}nr = nr  \frac{\sigmarstar}{\Delta} \log ( \frac{n}{\eps})$ per node.
The computational cost at the center is order $n^2 \tL  \log^3\left(\frac{\tL r}{\eps} \right)$.
 The computational cost at any node is order 
 $n \qell r \frac{\sigmarstar}{\Delta} \log (\frac{n}{\epsilon})$. 
\end{cor}

\subsection{Discussion and Comparisons}

\subsubsection{Comparing Subspace-Median and Subspace-MoM}
For a chosen value of $\tL \le L$, the sample complexity required by subspace-MoM reduces by a factor of $1/\rho = \tL/L$, but its Byzantine tolerance also reduces by this factor. This matches what is well known for other MoM estimators, e.g., that for gradients used in \cite{dist_adversarial}. Also, the communication cost of Subspace-MoM is larger than that of Subspace Median since it implements a power method to share samples between subsets of nodes.

\subsubsection{Comparing Subspace-Median with SVD-ResCovEst and ResPowMeth}
Consider communication cost. SVD-CovEst has a very high cost of order $n^2$ while Subspace Median and ResPowMeth have much lower costs of order $nr$ and $nr \frac{\sigmin}{\Delta}$ times a log factor respectively. 
Consider sample cost. Both SVD-CovEst and ResPowMeth have a very high sample cost of order $n^2$ and order $n^2 r^2$ respectively for $\eps=c$.  Subspace Median has a sample cost of only order $nrL$.

In terms of computation cost at the nodes, SVD-CovEst is the most expensive, while both ResPowMeth and Subspace Median have the same cost. But, at the center, Subspace Median has a  higher cost by a factor $n/(r \log^3(n/\eps))$. In many practical federated applications, the nodes are power limited, and hence their computation cost, and communication cost, needs to be low.  In terms of total algorithm speed, communication cost/time is often the main bottleneck. The computational cost at the center is a lesser concern.

\subsubsection{Comparison with standard federated power method in the no-attack setting}
Observe that, for a given normalized singular value gap, the sample complexity (lower bound on $q$) needed by the above result is order $nr L /\eps^2$ while that needed for standard PCA (without Byzantine nodes) is order $n/\eps^2$ \cite[Remark 4.7.2]{versh_book}. 
The reason we need an extra factor of $L$ is because we are computing the individual node estimates using $\q = q/L$ data points and we need each of the node estimates to be accurate (to ensure that their ``median'' is accurate). This extra factor of $L$ is needed also in other work that uses (geometric) median, e.g., \cite{dist_adversarial} needs this too.
The reason we need an extra factor of $r$ is because we need use Frobenius subspace distance, $\SD_F$, to develop and analyze the geometric median step of Subspace Median.  The bound provided by the Davis-Kahan sin Theta theorem for $\SD_F$ needs an extra factor of $\sqrt{r}$.

The per-node computational cost of standard federated PCA is $n \q r T_{pow}$ while that for SubsMed is $n \q r T_{pow} + n^2 L T_\gm$. Ignoring log factors and treating the singular value gap as a numerical constant  (ignoring $T_{pow}$ and $T_\gm$), letting $\eps=c$, and substituting the respective lower bounds on $\q$, the PCA cost is $n^2 r$ while that for SubsMed for Byzantine-resilient PCA is $\max(n^2 r^2 ,n^2 L) = n^2 \max(r^2, L)$. Thus the computational cost is only $\max(r, L/r)$ times higher.

We summarize the comparisons in Table \ref{table_compare}.

\begin{algorithm}[t]
	\caption{\small{Subspace Median-of-Means. Recall that $_{(\vartheta,\ell)} = _{ (\vartheta - 1)\rho + \ell}$. 
	}}
\label{alg:sub_gmom_pm}
	\begin{algorithmic}[1]
		\STATE \textbf{Input:} Batch $\D_{(\vartheta)} = [\D_{(\vartheta,1)}, \D_{(\vartheta,2)}, \dots, \D_{(\vartheta,\rho)}]$, $\vartheta \in [\tL]$. \\ 
		\STATE \textbf{Parameters:} $T_{pow}$
		\STATE \underline{\textbf{Central Server}}
		\STATE  Randomly initialize $\U_{rand}$ with i.i.d standard Gaussian entries. Set $\U_{(\vartheta)} = \U_{rand}$.
		\FOR {$t \in [T_{pow}]$ }
		\STATE \underline{\textbf{Nodes $\ell=1,...,L$}}
		\STATE Compute  $\tilde\U_{(\vartheta,\ell)} \leftarrow \D_{(\vartheta,\ell)} \D_{(\vartheta,\ell)}^\top \U_{(\vartheta)}$,  $\ell \in (\vartheta - 1)\rho +[\rho]$, $\vartheta \in [\tL]$. Push $\tilde\U_{(\vartheta,\ell)}$ to center.
		\STATE \underline{\textbf{Central Server}}
 \STATE  Compute $ \U_{(\vartheta)} \leftarrow QR(\sum_{\ell=1}^{\rho}\tilde\U_{(\vartheta,\ell)})$, $\vartheta \in [\tL]$
		\STATE Push $\U_{(\vartheta)}$ to nodes $\ell \in (\vartheta - 1)\rho +[\rho]$.
		\ENDFOR
		\STATE Use Algorithm \ref{NEW_PCA_1} for the input $\{\U_{(\vartheta)}\}_{\vartheta=1}^{\tL}$
		\STATE {\bf Output} $\U_{out}$.
	\end{algorithmic}
	
\end{algorithm}

\section{Application 2: Horizontally federated LRCS} \label{lrccs_section}

\subsection{Problem setting}\label{lrccs_problem}

\subsubsection{Basic problem}
The LRCS problem involves recovering an $n \times q$ rank-$r$ matrix $\Xstar=[\xstar_1, \xstar_2, \dots, \xstar_q]$, with  $r \ll \min(q,n)$, from
$
\y_k := \A_k \xstar_k,  \ k \in [q]
\label{obsmod}
$
when $\y_k$ is an $m$-length vector with $m < n$, and the measurement matrices $\A_k$ are known and independent and identically distributed (i.i.d.) over $k$. We assume that each $\A_k$ is a ``random Gaussian'' matrix, i.e., each entry of it is i.i.d. standard Gaussian.
Let
$
\Xstar \stackrel{SVD}{=} \Ustar  \bSigma \Vstar{}^\top: = \Ustar \Bstar
$
denote its reduced (rank $r$) SVD, and $\kappa:= \sigmax/\sigmin$ the condition number of $\bSigma$.
Notice that each measurement $\y_\ik$ is a global function of column $\xstar_k$, but not of the entire matrix. As explained in \cite{lrpr_gdmin}, to make it well-posed (allow for correct interpolation across columns), we need the following incoherence assumption on the right singular vectors. 

\begin{assu}[Right Singular Vectors' Incoherence]\label{right_incoh}
	We assume that $\max_k \| \bstar_k \| \leq \mu \sqrt{r/q} \sigmax$ for a constant $\mu \geq 1$.
	\label{right_incoh}
\end{assu}

\subsubsection{Horizontal Federation}
Consider the $m \times q$ measurements' matrix,
\[
\Y=[\y_1,\y_2,...,\y_q]=[\A_1\xstar_1,\A_2\xstar_2,...,\A_q\xstar_q].
\]
We assume that there are a total of $L$ nodes and each node observes a different disjoint subset of $\m$ rows of $\Y$. Denote the set of indices of the rows available at node $\ell$ by $\cS_\ell$. Thus $|\cS_\ell|=\m = m/L$. 
We assume that node $\ell$ has access to $\Y_{\ell}$ and $\{(A_k)_\ell, k \in [q]\}$.  Here $(\A_k)_\ell:= \I_{\cS_\ell}^\top \A_k$ is $\m \times n$ and $\Y_\ell:= \I_{\cS_\ell}^\top \Y$ is of size $\m \times q$ with $\m = m/L$, and with $(\y_k)_\ell: = (\A_k)_\ell \xstar_k$ for all $k \in [q]$

Observe that the sub-matrices of rows of $\Y$, $\Y_\ell$, are identically distributed, in addition to being independent. Consequently, the same is true for the partial gradients computed at the different nodes. This is why, without extra assumptions, we can make our solution Byzantine resilient. On the other hand, column sub-matrices of $\Y$ are not identically distributed. In order to obtain provable guarantees for vertical LRCS, we will need extra assumptions that bound on the amount of heterogeneity in the data (and hence in the nodes' partial gradients). This is being studied in ongoing work.

\subsubsection{Byzantine nodes}
Assumption \ref{byzassu} holds. Also, the set of Byzantine nodes may change at each AltGDmin algorithm iteration. 

\subsection{Review of Basic altGDmin \cite{lrpr_gdmin}}
We first explain the basic idea \cite{lrpr_gdmin} in the simpler no-attack setting. AltGDmin imposes the LR constraint by expressing the unknown matrix $\X $ as $\X = \U \B$ where $\U$ is an $n \times r$ matrix and $\B$ is an $r \times q$ matrix.
In the absence of attacks, the goal is to minimize
\[
f(\U,\B): = \sum_{k=1}^q \|\y_k - \U \b_k\|^2
\]
AltGDmin proceeds as follows:
\ben
\item {\em Truncated spectral initialization:} Initialize $\U$ (explained below). 
\item At each iteration, update $\B$ and $\U$ as follows:
\ben
\item {\em Minimization for $\B$:} keeping $\U$ fixed, update $\B$ by solving $\min_{\B} f(\U, \B)$. Due to the form of the LRCS measurement model, this minimization decouples across columns, making it a cheap least squares problem of recovering $q$ different $r$ length vectors. It is solved as $\b_k \leftarrow (\A_k \U)^\dag \y_k$ for each $k \in [q]$.

\item {\em GD for $\U$:} keeping $\B$ fixed, update $\U$ by a GD step, followed by orthonormalizing its columns: 
$\U^+  \leftarrow QR(\U - \eta \nabla_{\U} f(\U,\B)))$.
Here $\nabla_{\U} f(\U,\B)=\sum_{k \in [q]}\A_k^{\top} (\A_k \U \b_k - \y_k) \b_k^{\top}$
\een
\een
The use of full minimization to update $\B$ is what helps ensure that AltGDmin provably converges, and that we can show exponential error decay with a constant step size (this statement treats $\kappa$ as a numerical constant) \cite{lrpr_gdmin,lrpr_gdmin_2}. Due to the decoupling in this step, its time complexity is only as much as that of computing one gradient w.r.t. $\U$. Both steps need time of order $mqnr$.
In a federated setting, AltGDmin is also communication-efficient because each node needs to only send $nr$ scalars (gradient w.r.t $\U$) at each iteration.%

We initialize $\U$ by computing the top $r$ singular vectors of
\[
\X_0 := \sum_k \A_k^\top (\y_{k})_\trnc (\alpha) \e_k^\top, \ \text{where} \  \y_\trnc(\alpha):=(\y \circ \indic_{|\y| \le  \sqrt\alpha}) 
\]
Here $ \alpha:= 9\kappa^2\mu^2 \sum_k \|\y_k\|^2 / mq$  and $\y_\trnc(\alpha)$ is a truncated version of the vector $\y$ obtained by zeroing out entries of $\y$ with magnitude larger than $\alpha$ 
(the notation $|\y|$ means $|\y|_i = |y_i|$ for each entry $i$, the notation $\indic_{\bm{z} \le \alpha}$ returns a 1-0 vector with 1 where $\bm{z}_j < \alpha$ and zero everywhere else, and $\bm{z}_1 \circ \bm{z}_2$ is the Hadamard product between the two vectors, i.e., the ``.*'' operation in MATLAB)

Sample-splitting is assumed to prove the guarantees. This means the following: we use a different independent set of measurements and measurement matrices $\y_k,\A_k, k \in [q]$ for each new update of $\U$ and of $\B$. We also use a different independent set for computing the initialization threshold $\alpha$.

{\em All expected values used below are expectations conditioned on past estimates (which are functions of past measurement matrices and measurements, $\A_k,\y_k$). For example, $\E[(\nabla_U f)_{\ell}]$ conditions on the values of $\U, \B_\ell$ used to compute it. This is also the reason why $\E[(\nabla_U f)_{\ell}]$ is different for different nodes; see Lemma \ref{new_term_lem}.
}

\subsection{Resilient Federated Spectral Initialization}

This consists of two steps. First the truncation threshold $\alpha=\frac{\tC}{m q}\sum_{k}\sum_{i} y_{ki}^2$ which is a scalar needs to be computed. This is simple: each node computes $\alpha_\ell=\frac{\tC}{\m q}\sum_{k}\sum_{i\in\cS_\ell} (y_\ell)_{ki}^2$ and sends it to the center which computes their median.

Next, we need to compute $\U_0$ which is the matrix of top $r$ left singular vectors of $\X_0$, and hence also of $\X_0 \X_0^\top$. Node $\ell$ has data to compute the $n \times q$ matrix $(\Xhat_0)_\ell$, defined as 
\begin{align}
(\Xhat_0)_\ell := \sum_{k=1}^q (\A_k)_\ell{}^\top ((\y_k)_\ell)_\trnc \e_k^\top,
\label{Xhat0_ell}
\end{align}
Observe that $\Xhat_0 = \sum_\ell (\Xhat_0)_\ell$. If all nodes were good (non-Byzantine), we would use this fact to implement the federated power method for this case. 
However, some nodes can be Byzantine and hence this approach will not work.
For reasons similar to those explained in Sec. \ref{res_sub_estim}, (i) an obvious GM-based modification of the federated power method will not work either, and (ii) nodes cannot send the entire $(\Xhat_0)_\ell$ (this is too expensive to communicate). We instead use Subspace Median, Algorithm \ref{NEW_PCA_1}, applied to $\D_\ell =(\Xhat_0)_\ell$. This is both communication-efficient and sample-efficient. It can be shown that it will work under a sample complexity lower bound that is comparable to that needed in the attack-free setting. We summarize this in Algorithm \ref{node_init}. We can obtain a guarantee for this approach by applying Theorem \ref{main_res} with $\bPhi_\ell \equiv (\Xhat_0)_\ell (\Xhat_0)_\ell^\top / q$ and using the results from \cite{lrpr_gdmin,lrpr_gdmin_2} to ensure that the assumption needed by Theorem \ref{main_res} holds.  We directly state a guarantee for the GM of means estimator developed next.
The guarantee for
Algorithm \ref{node_init} is a special case of that for the GM of means estimator developed next with $\tL=L$, and thus its guarantee is also given by Corollary \ref{u_init_gmom} with $\tL=L$.

\begin{algorithm}[t]
	\caption{Byz-AltGDmin: Initialization using Subspace Median}\label{node_init}
	\begin{algorithmic}
\STATE	 \textbf{Input} $\Y_\ell$ and $(\A_k)_\ell$
\STATE  \textbf{Parameters} $T_\gm, T_{pow}$
		\STATE \underline{\textbf{Nodes $\ell=1,...,L$}}
		\STATE Compute $\alpha_\ell \leftarrow \frac{9 \kappa^2\mu^2}{\m q}\sum_{k} \|(\y_k)_\ell\|^2$.  	
		\STATE \underline{\textbf{Central Server}}
		\STATE $\alpha \leftarrow \text{Median}\{\alpha_\ell\}_{\ell=1}^{L}$	
		\STATE \underline{\textbf{Nodes $\ell=1,...,L$}}
		\STATE Compute $(\U_0)_{\ell}$ $\leftarrow$ top-$r$-singular vectors of $(\Xhat_0)_\ell$ defined in \eqref{Xhat0_ell}
		(use power method with $T_{pow}$ iterations).
		\STATE \underline{\textbf{Central Server}}
		\STATE Use Algorithm \ref{NEW_PCA_1} (Subspace-Median) with parameter $T_\gm$ on $(\U_0)_{\ell}$, $\ell \in [L]$. 
		\STATE {\bf Output} $\U_{out}$.
	\end{algorithmic}
\end{algorithm}

\subsection{Resilient Federated Spectral Initialization: Horizontal Subspace-MoM}\label{lrccs_gmom}
As explained earlier for PCA, the use of just (geometric) median wastes samples. Hence, we develop a median-of-means estimator. For a parameter $\tL \le L$, we would like to form $\tL$ mini-batches of $\rho = L/\tL$ nodes; w.l.o.g. $\rho$ is an integer.
In our current setting, the data is horizontally federated. This requires a different approach to combine samples than what we used for PCA in Sec. \ref{subspace_gmom}. Here, each node can compute the $n \times q$ matrix $(\Xhat_0)_\ell$. 
Combining samples means combining the rows of $(\A_k)_\ell$ and $(\y_k)_\ell$ for $\rho$ nodes to obtain $(\Xhat_0)_{(\vartheta)}$ with $k$-th column given by
$
\sum_{\ell=1}^\rho (\A_k)_{(\vartheta,\ell)}^\top (\y_{k,\trnc})_{(\vartheta,\ell)} / \rho.
$
Recall that ${(\vartheta,\ell)} = {(\vartheta-1)\rho + \ell)}$. 
To compute this in a communication-efficient and private fashion, we use a horizontally federated power method for each of the $\tL$ mini-batches. The output of each of these power methods is $\U_{(\vartheta)}, \ \vartheta \in [\tL]$. These are then input to the subspace-median algorithm, Algorithm \ref{NEW_PCA_1} to obtain the final subspace estimate $\U_{out}$.
To explain the federation details simply, we explain them for $\vartheta=1$. 
The power method needs to federate $\U \leftarrow QR( (\Xhat_0)_{(1)} (\Xhat_0)_{(1)}^\top  \U) =QR( \sum_{\ell'=1}^\rho (\Xhat_0)_{\ell'}  (\sum_{\ell=1}^\rho (\Xhat_0)_\ell^\top \U)) $. This  needs two steps of information exchange between the nodes and center at each power method iteration. In the first step, we compute $\V = \sum_{\ell \in [\rho]} (\X_0)_\ell{}^\top \U$, and in the second one we compute  $\tilde\U = \sum_{\ell \in [\rho]} (\X_0)_\ell \V$, followed by its QR decomposition.

We summarize the complete algorithm in  Algorithm \ref{alg:gmom_init}. As long as the same set of $\tau L$ nodes are Byzantine for all the power method iterations needed for the initialization step, we can prove the following result for it\footnote{This assumption can be relaxed if we instead assume that the size of the set of nodes that are Byzantine in any one initialization iteration is at most $\tau L$.}. This follows as a corollary of Theorem \ref{main_res} and the lemmas proved in \cite{lrpr_gdmin,lrpr_gdmin_2} for the attack-free case.

\begin{cor}[Initialization using Subspace-GMoM]
\label{u_init_gmom}
Consider the Initialization steps (lines 3-22) of Algorithm \ref{byzaltgdmin} 
with $T_\gm = C\log(\frac{Lr}{\delta_0})$ and $T_{pow} = C\kappa^2 \log(\frac{n}{\delta_0})$. 
Assume that Assumption \ref{right_incoh} hold. Assume also that the set of Byzantine nodes remains fixed for all iterations in this algorithm and the size of this set is at most $\tau L$ with  $\tau< 0.4\tL/L$. 
Pick a $\delta_0 < 1$ and an $\tL < L$ such that $L$ is a multiple of $\tL$.
If
$$m q\geq C \tL  \cdot \kappa^6\mu^2(n+q)r^2 / \delta_0^2,$$
then w.p. at least $1- c_\approxgm - \exp(-L\psi(0.4-\tau,\exp(-c(n+q)) + n^{-10} )) - L\exp(-\tilde{c}\m q\delta_0^2/r^2\kappa^4)$, 
	\[
		\SD_F(\Ustar,\U_{out})\leq \delta_0.
	\]
The communication cost per node is order $nr \cdot T_{pow} = \kappa^2 nr \log(\frac{n}{\delta_0})$.
\end{cor}

\begin{proof}
This follows by applying Theorem \ref{th:ro_pca} on $\bPhi_{(\vartheta)} = \sum_{\ell=1}^\rho (\Xhat_0)_{(\vartheta,\ell)} (\Xhat_0)_{(\vartheta,\ell)}^\top /  \rho$ and $\bPhi^* = \E[(\Xhat_0)_\ell |\alpha] \E[(\Xhat_0)_\ell |\alpha]^\top$ for $\vartheta \in [\tL]$ and using the results from \cite{lrpr_gdmin,lrpr_gdmin_2} to ensure that the assumption needed by Theorem \ref{main_res} holds. 

The idea is almost exactly the same as for the special case $\tL=L$. This case is simpler notation-wise and hence we provide a proof for this case in  Appendix \ref{proof_u_init_1}.
The main idea is as follows. Let $\D(\alpha)$ be the positive entries' diagonal matrix defined in  \cite[Lemma 3.8]{lrpr_gdmin}. 
We use \cite[Lemma 3.8]{lrpr_gdmin} and  \cite[Fact 3.9]{lrpr_gdmin} to show that $\E[(\X_0)_\ell|\alpha]  = \Xstar \D(\alpha)$  and to bound $\|(\X_0)_\ell  - \E[(\X_0)_\ell|\alpha] \|$.
We then use this bound to then get a bound  $\|\bPhi_\ell - \bPhi^*\|$. In the last step, we use an easy median-based modification of \cite[Fact 3.7]{lrpr_gdmin} to remove the conditioning on $\alpha$. 
\end{proof}

\subsection{Byzantine-resilient Federated AltGDmin: GDmin Iterations}
We can make the altGDmin iterations resilient as follows. 
In the minimization step, each node computes its own estimate $(\b_k)_\ell$ of $\bstar_k$ as follows:
\[
(\b_k)_\ell = ((\A_k)_\ell \U)^\dag (\y_k)_\ell , \ k \in [q]
\]
Here, $\M^\dag:= (\M^\top \M)^{-1}\M^\top$.
Each node then uses this to compute its estimate of the gradient w.r.t. $\U$ as $\nabla f_\ell= \sum_{k \in [q]} (\A_k)_\ell^\top ((\A_k)_\ell \U (\b_k)_\ell -  (\y_k)_\ell) (\b_k)_\ell^\top$. The center receives the gradients from the different nodes, computes their GM and uses this for the projected GD step. Since the gradient norms are not bounded, the GM computation needs to be preceded by the thresholding step explained in Sec. \ref{gm_unbounded}.

As before, to improve sample complexity (while reducing Byzantine tolerance), we can replace GM of the gradients by their GM of means: form $\tL$ batches of size $\rho = L/\tL$ each, compute the mean gradient within each batch, compute the GM of the $\tL$ mean gradients. Use appropriate scaling. 
We summarize the GMoM algorithm in Algorithm \ref{byzaltgdmin}. The GM case corresponds to $\tL=L$.
Given a good enough initialization, a small enough fraction of Byzantine nodes, enough samples $\m q$ at each node at each iteration, and assuming that Assumption \ref{right_incoh} holds, we can prove the following for the GD iterations.

\begin{theorem}(AltGDmin-GMoM: Error Decay)\label{gdmin_thm_gmom} \label{gdmin_thm}
	Consider the AltGDmin steps of Algorithm \ref{byzaltgdmin}
with sample-splitting, and with a step-size $\eta \le 0.5/\sigmax^2$.
Set $T_\gm = C \log \frac{L r}{\eps}$, $\thresh =  C \m \sigmin^2$.
	Assume that Assumptions \ref{byzassu} and \ref{right_incoh} holds. If, at each iteration $t$,
$$m q\geq C_1 \tL \kappa^4\mu^2(n+r)r^2 , $$  $\m>C_2\max(\log q, \log n)$; if $\tau< 0.4\tL/L$; and
	if the initial estimate $\U_0$ satisfies $\SD_F(\Ustar,\U_0)\leq \delta_0 = 0.1/\kappa^2$, then w.p. at least $1 - c_\approxgm -t\left[Ln^{-10} +\exp(-L\psi(0.4-\tau,n^{-10}))\right]$, 
	\[
	\SD_F(\Ustar,\U_{t+1})\leq \delta_{t+1}:=\left(1-(\eta\sigmax^2)\frac{0.12}{\kappa^2}\right)^{t+1}\delta_0
	\]
	and $\|\xstar_k - (\x_k)_{t+1}\| \le \delta_{t+1} \|\xstar_k\|$ for all $k \in [q]$.

The communication cost per node is order $nr \cdot T = \kappa^2 nr \log(\frac{1}{\eps})$.
\end{theorem}

\begin{proof}
Consider the $\tL=L$ (GM) special case since this is notationally simpler. The extension for the general $\tL < L$ (GM of means) case is straightforward.
The proof uses the overall approach developed in \cite{lrpr_gdmin_2} with the following changes. Let $\ell_1 := (\Jgood)_1$ be {\em a} non-Byzantine node.  
We now also need a bound on the Frobenius norm of
\[
\Err := \nabla f_\gm - \E[ \nabla f_{\ell_1} (\U,\B_{\ell_1})], \ \ \ell_1 := (\Jgood)_1
\]
that is of the form $c \delta_t \m \sigmax^2$ for a $c < 1$ w.h.p., under the assumed sample complexity bound. This type of a bound, along with assuming $\delta_0 < c/\sqrt{r} \kappa^2$, helps ensure that the algebra needed for showing exponential decay of the subspace estimation error goes through. 
We can get the above bound on $\|\Err\|_F$ using Lemma \ref{gm_new_2} if we can get a similar bound on
\[
\max_{\ell \in \Jgood} \| \nabla f_\ell - \E[ \nabla f_{\ell_1} (\U,\B_{\ell_1})] \|_F
\]
This is proved in the lemma given next.
\begin{lemma} \label{new_term_lem}
	Assume $\SD_F(\Ustar,\U)\leq \delta_t<\delta_0$. Then,
w.p. at least $1-\exp \Big((n+r)-c\epsilon_1^2\frac{\m q}{r^2\mu^2}\Big)-2\exp(\log q +r -c\epsilon_1^2\m)$, for all $\ell \in \Jgood$,
		\[
		\norm{\nabla f_{\ell}(\U,\B_\ell)-\E[\nabla f_{\ell_1}(\U,\B_{\ell_1})]}_F\leq 12.5 \epsilon_1 \delta_t \m \sigmax^2
		\]
\end{lemma}
We prove this lemma by noting that
\[
\nabla f_{\ell}(\U,\B_\ell)-\E[\nabla f_{\ell_1}(\U,\B_{\ell_1})] = \left( \nabla f_{\ell}(\U,\B_\ell)-\E[\nabla f_{\ell}(\U,\B_{\ell})] \right) + \left(\E[\nabla f_{\ell}(\U,\B_{\ell})] - E[\nabla f_{\ell_1}(\U,\B_{\ell_1})]\right)
\]
The first term can be bounded using standard concentration bounds. The second one requires carefully bounding $\|\B_\ell - \B_{\ell_1}\|_F$ by using the fact that both $\B_\ell, \B_{\ell_1}$ are close to $\G:=\U^\top \Xstar$.

We provide its proof and the  complete proof of our Theorem in Appendix \ref{proof_gdmin_byz}.
\end{proof}

\subsection{Complete Byz-AltGDmin algorithm} 

Combining Corollary \ref{u_init_gmom} and Theorem \ref{gdmin_thm_gmom}, and setting $\eta = 0.5/\sigmax^2$  and $\delta_0 = 0.1/\kappa^2$, we can show that, at iteration $t+1$, $\SD_F(\Ustar,\U_{t+1})\leq \delta_{t+1} = (1- 0.06/\kappa^2)^{t+1} 0.1/\kappa^2 $ whp.
Thus, in order for this to be $\le \eps$, we need to set $T = C \kappa^2 \log(1/\eps)$. 
Also, since we are using fresh samples at each iteration (sample-splitting), this also means that our sample complexity needs to be multiplied by $T$.

We thus have the following final result.

\begin{cor}
	Consider the complete Algorithm  \ref{byzaltgdmin}  
with sample-splitting.
Set  $T_{pow}=C\log (n \kappa )$,  $\eta = 0.5/\sigmax^2$, $T=  C \kappa^2 \log(1/\eps)$.  Assume that Assumption \ref{right_incoh} holds.
	If the total number of samples per column $m$, satisfies
$$m q\geq C \tL \kappa^4\mu^2(n+q)r^2 \log (1/\epsilon)$$ and $ m>C \kappa^2 \max(\log q, \log n) \log(1/\eps)$;  if at most $\tau L$ nodes are Byzantine with $\tau< 0.4\tL/L$, if the set of Byzantine nodes remains fixed for the initialization step power method (but can vary for the GDmin iterations);
	then, w.p. at least $1- T L n^{-10}$,
	$
	\SD_F(\Ustar,\U_T)\leq \eps, 
  $	
	and $\|\x_k - \xstar_k\| \le \eps \|\xstar_k\|$ for all $k \in [q]$.

The communication cost per node is order $\kappa^2 nr \log(\frac{n}{\eps})$.
The computational cost at any node is order $\kappa^2 \m n q r\log(\frac{n}{\eps})$ while that at the center is $n^2 \tL  \log^3(\tL r/\eps)$.
\end{cor}

The above result shows that, under exactly one assumption (Assumption \ref{right_incoh}), if each node has enough samples $\m$ ($\m$ is of order $(n+q)r^2 (\tL/L)$ times log factors); if the number of Byzantine nodes is less than $(0.4 \tL/L)$ times the total number of nodes, then our algorithm can recover each column of the LR matrix $\Xstar$ to $\eps$ accuracy whp. To our best knowledge, the above is the first guarantee for Byzantine resiliency for any type of low rank matrix recovery problems studied in a federated setting.

Observe that the above result needs total sample complexity that is only $\tL$ times that for basic AltGDmin \cite{lrpr_gdmin}.

\begin{algorithm}[t]
	\caption{\small{Byz-AltGDmin: Complete GMoM based algorithm}} 
	\begin{algorithmic}[1]\label{alg:gmom_init} \label{byzaltgdmin}
\STATE \textbf{Input:} Batch $\vartheta:\quad\{(\A_k)_\ell,\Y_\ell, k \in [q]\}$, $\ell\in[L]$  \\ 
\STATE \textbf{Parameters:} $T_{pow},T_\gm$,

\STATE {\bf \em Initialization using Subspace MoM: }

\STATE \underline{\textbf{Nodes $\ell=1,...,L$}}
		\STATE Compute $\alpha_\ell \leftarrow \frac{9 \kappa^2\mu^2}{\m q}\sum_{k} \|(\y_k)_\ell\|^2$.  	
		\STATE \underline{\textbf{Central Server}}
		\STATE $\alpha \leftarrow \text{Median}\{\alpha_{(\vartheta)}\}_{\vartheta=1}^{\tL}$, where $\alpha_{(\vartheta)} = \sum_{\ell=1}^\rho \alpha_{(\vartheta,\ell)} /\rho$ 
		\STATE \underline{\textbf{Central Server}}
		\STATE Let $\U_0 = \U_{rand}$ where $\U_{rand}$ is an $n \times r$ matrix with i.i.d standard Gaussian entries.
		\FOR {$\tau \in [T_{pow}]$ }
		\STATE \underline{\textbf{Nodes $\ell=1,...,L$}}
		\STATE Compute  $\V_{\ell} \leftarrow (\Xhat_0)_\ell^\top (\U_{(\vartheta)})_{\tau-1}$ for $\ell \in (\vartheta - 1)\rho +[\rho]$, $\vartheta \in [\tL]$. Push to center.
		\STATE \underline{\textbf{Central Server}}
		\STATE  Compute $\V_{(\vartheta)} \leftarrow \sum_{\ell=1}^{\rho}\V_{(\vartheta - 1)\rho + \ell}$
        \STATE Push $\V_{(\vartheta)}$ to nodes $\ell \in (\vartheta - 1)\rho +[\rho]$.
		\STATE \underline{\textbf{Nodes $\ell=1,...,L$}}
		\STATE Compute $\U_\ell \leftarrow \sum_{k}(\Xhat_0)_\ell\V_{(\vartheta)}$ for $\ell \in (\vartheta - 1)\rho +[\rho]$, $\vartheta \in [\tL]$. Push to center.
		
		\STATE \underline{\textbf{Central Server}}
		\STATE  Compute $\U_{(\vartheta)} \leftarrow QR(\sum_{\ell=1}^{\rho}\U_{(\vartheta - 1)\rho + \ell} )$
		\STATE 
Push $\U_{(\vartheta)}$ to nodes $\ell \in (\vartheta - 1)\rho +[\rho]$.
		\ENDFOR
		\STATE  Apply Algorithm \ref{NEW_PCA_1} on $\{U_{\vartheta}\}_{\vartheta=1}^{\tL}$ to get $\U_{out}$.
		\STATE Set $\U_0 \leftarrow \U_{out}$ 

\STATE {\bf \em AltGDmin Iterations: }

		\FOR {$t =1$ to $T$}
		\STATE \underline{\textbf{Nodes $\ell=1,...,L$}}
		\STATE Set $\U \leftarrow \U_{t-1}$
		\STATE $(\b_k)_{\ell} \leftarrow ((\A_k)_\ell\U)^{\dagger}(\y_k)_{\ell}$, $\ \forall \ k \in [q]$
		\STATE $(\x_k)_{\ell} \leftarrow \U (\b_k)_{\ell}$, $\ \forall \ k \in [q]$
		\STATE $ (\nabla f)_{\ell} \leftarrow \sum_{k \in [q]}  (\A_k)_\ell^\top ((\A_k)_\ell \U (\b_k)_{\ell} - (\y_k)_{\ell}) (\b_k)_{\ell}{}^\top$, $\ \forall \ k \in [q]$ 
		
		\STATE Push $\nabla f_\ell $ to center 
		
		\STATE \underline{\textbf{Central Server}}
		\STATE Compute $\nabla f_{(\vartheta)} \leftarrow \sum_{\ell\in\vartheta}\nabla f_\ell$

		\STATE $\nabla f_\gm \leftarrow approxGM\left( \{ vec(\nabla f_{(\vartheta)}), \ \vartheta \in [\tL] \} \setminus \{\vartheta: \|\nabla f_{(\vartheta)}\|_F > \thresh\} \right)$ \label{GM_line}
\\ (Use \cite[Algorithm 1]{cohen2016geometric}  with $T_\gm$ iterations on the set of $\nabla f_{(\vartheta)}$s whose Frobenius norm is below $\thresh$)

		\STATE Compute $\U^+ \leftarrow QR (\U_{t-1}-\frac{\eta}{\rho\m} \nabla f_\gm )$  
		\RETURN Set $\U_t \leftarrow \U^{+}$. Push $\U_t$ to nodes.
		\ENDFOR
\STATE {\bf Output} $\U_T$.
	\end{algorithmic}
\end{algorithm}

\section{Proofs for Sec \ref{res_sub_estim} and \ref{res_pca}}\label{proofs1}

All the proofs given below rely on the lemma for using GM for robust estimation borrowed from \cite{dist_adversarial}.
We give these lemmas in the section below, followed by two corollaries that will be used in our proofs.

\subsection{Using GM for robust estimation}
The goal of robust estimation is to get a reliable estimate of a vector quantity $\tz$ using $L$ individual estimates of it, denoted $\z_\ell$, when most of the estimates are good, but a few can be arbitrarily corrupted or modified by Byzantine attackers. A good approach to do this is to use the GM.
The following lemma, which is a minor modification of \cite[Lemma 2.1]{dist_adversarial}, studies this\footnote{\cite[Lemma 2.1]{dist_adversarial} does not provide an algorithm for approximating the Geometric Median; we combine their result with Claim \ref{cohen_1} to provide this}

\begin{lemma}\label{gm_new}
Consider $\{\z_1, \z_2, \dots, \z_\ell, \dots \z_L\}$ with $\z_\ell \subseteq \Re^n$. Let  $\z_\gm^*$ denote their GM and let
$\z_\gm$  denote their $(1+\eps_\gm)$ approximate GM estimate computed using \cite[Algorithm 1]{cohen2016geometric}.
	Fix an $\alpha\in(0,1/2)$.
	Suppose that the following holds for at least $(1-\alpha)L$ $\z_\ell$s:
\[
\norm{\z_\ell-\tz  }\leq \epsilon \| \tz\|.
\]
Let $C_\alpha := \frac{2(1-\alpha)}{1-2\alpha}$. Then, w.p. at least $1-c_\approxgm$,
	\begin{align*}
\norm{\z_\gm-\tz} 		 \leq C_{\alpha}\epsilon \| \tz\| + \epsilon_\gm\frac{\sum_{\ell=1}^{L}\norm{\z_\gm^*-\z_\ell}}{(1-2\alpha)L}
		\leq C_{\alpha}\epsilon \| \tz\| + \epsilon_\gm\frac{\max_{\ell \in [L]} \norm{\z_\ell}}{1-2\alpha}  
	\end{align*}
The number of iterations needed for computing $\z_\gm$ is  $T_\gm = C \log(\frac{L}{\eps_\gm})$, and the time complexity is $O\left(n L \log^3(\frac{L}{\eps_\gm})\right)$.
\end{lemma}
The second inequality follows because, using the exact GM definition, $\sum_\ell \norm{\z_\gm^*-\z_\ell} \le \sum_\ell \norm{\bm{0} - \z_\ell} = \sum_\ell \norm{\z_\ell}$ and $\sum_\ell \norm{\z_\ell} \le L \max_\ell \norm{\z_\ell}$. 
To understand this lemma simply, fix the value $\alpha$ to 0.4. Then $C_\alpha = 6$. We can also fix $\eps_\gm = \eps$. Then, it says the following. If at least 60\% of the $L$ estimates are $\eps$ close to $\tz$, then, the $(1+\eps)$ approximate GM, $\z_\gm$, is  $ 11 \eps \max( \norm{\tz}, \max_{\ell \in [L]} \|\z_\ell\|)$ close to $\tz$.
The next lemma follows using the above lemma and is a minor modification of \cite[Lemma 3.5]{dist_adversarial}. 
It fixes $\alpha=0.4$ and considers the case when most estimates are good with high probability (w.h.p.). We provide a short proof of it in Appendix \ref{gm_prob_proof}.

\begin{lemma}\label{gm_new_1}
	Let $\z_\ell \subseteq \Re^n$, for $\ell\in[L]$ and let $\z_\gm$ denote a $(1+\eps_\gm)$ approximate GM computed using \cite[Algorithm 1]{cohen2016geometric}. 
	For a $\tau< 0.4$, suppose that, for at least $(1-\tau) L$ $\z_\ell$'s,
	\[
	\Pr\{\norm{\z_\ell-\tz}\leq \eps \|\tz\| \} \geq 1-p
	\]
	Then, w.p. at least $1- c_\approxgm -\exp(-L\psi(0.4-\tau,p))$,
	\[
	\norm{\z_\gm-\tz} \leq  6 \eps \| \tz\| + 5 \eps_\gm  \max_{\ell \in  [L]}\norm{\z_\ell}  
	\]
where $\psi(a, b)=(1-a)\log \frac{1-a}{1-b}+a\log \frac{a}{b}$.
The number of iterations needed for computing $\z_\gm$ is  $T_\gm = C \log(\frac{L}{\eps_\gm})$, and the time complexity is $O\left(n L \log^3(\frac{L}{\eps_\gm})\right)$.
\end{lemma}

Suppose that, for a $\tau  < 0.4$, at least $(1-\tau) L$ $\z_\ell$s are ``good'' (are $\eps$ close to $\tz$) whp. Let $\eps_\gm=\eps$ and suppose that all $\z_\ell$'s, including the corrupted ones, are bounded in 2-norm by $\norm{\tz}$. Then,  the $(1+\eps)$-approximate GM is about $11 \eps \norm{\tz}$ close to $\tz$ with at least constant probability.
If the GM is approximated with probability 1, i.e., if $c_\approxgm = 0$, then, the above result says that, for $p$ small enough and large $L$, the reliability of the GM is actually higher than that of the individual good estimates. For example, for a $p < 0.01$, the probability is at least  $1 - p^{ L (0.4 - \tau)}$. The increase depends on $(0.4 -\tau)$ and $L$, e.g., if $\tau \ge 0.2$ and $L \ge 10$, then, the probability is at least $1 - p^{0.2 L} \ge 1 - p^2$.

\subsubsection{Corollary for Bounded $z_\ell$s}
In settings where all $\z_\ell$'s are bounded, we have the following corollary of Lemma \ref{gm_new_1}.
\begin{cor}
In the setting of Lemma \ref{gm_new_1}, if $\max_\ell \norm{\z_\ell} \le \norm{\tz}$, then $\norm{\z_\gm-\tz} \le 11 \eps \norm{\tz}$ with above probability.
The number of iterations needed is  $T_\gm = C \log(\frac{L}{\eps_\gm})$, and the time complexity is $O\left(n L \log^3(\frac{L}{\eps_\gm})\right)$.
\end{cor}
We use this for analyzing the Subspace Median algorithm in which the $\z_\ell$s are vectorized projection matrices from the different nodes.

\subsubsection{Corollary for Unbounded $z_\ell$s} \label{gm_unbounded}
When some $\z_\ell$s may not be bounded, we need an extra thresholding step.
Observe that, from the assumption in Lemma \ref{gm_new_1}, w.p. at least $1 - Lp$, the good $\z_\ell$s are bounded by $(1 + \eps) \|\tz\|$. Thus, to get a set of $\z_\ell$'s that are bounded in norm, while not eliminating any of the good ones, we can create a new set that only contains $\z_\ell$'s with norm smaller than threshold $\thresh = (1 + \eps) \|\tz\|$. In other words, we compute the GM of the set $\{\z_1,...,\z_{L}\}\setminus\{ \z_\ell:\norm{\z_\ell}> (1 + \eps) \|\tz\| \}$ as the input to the GM computation algorithm \cite[Algorithm 1]{cohen2016geometric}.  More generally, $\omega$ can be set to $C\|\tz\|$ for any $C > 1$. In practice, to set the threshold, we only need to have an estimate of the norm of the unknown quantity $\tz$ that we are trying to estimate.

We have the following corollary of Lemma \ref{gm_new_2} for this setting.

\begin{cor}\label{gm_new_2}
	Let $\z_\gm$ denote a $(1+\eps_\gm)$ approximate GM of $\{\z_1,...,\z_{L}\}\setminus\{ \z_\ell:\norm{\z_\ell}>  \thresh \}$, all vectors are in $\Re^n$. Set $\thresh = (1 + \eps) \|\tz\|$.
	For a $\tau< 0.4$, suppose that, for at least $(1-\tau) L$ $\z_\ell$'s,
	\[
	\Pr\{\norm{\z_\ell-\tz}\leq \eps \|\tz\| \} \geq 1-p
	\]
Then, w.p. at least $1- c_\approxgm - L p -  \exp(-L\psi(0.4-\tau,p))$,
	\[
	\norm{\z_\gm-\tz} \leq  6 \eps \| \tz\| + 5 \eps_\gm   (1 + \eps) \|\tz\| < 14 \max(\eps,\eps_\gm) \|\tz\| 
	\]
The number of iterations needed is  $T_\gm = C \log(\frac{L}{\eps_\gm})$, and the time complexity is $O\left(n L \log^3(\frac{L}{\eps_\gm})\right)$.
\end{cor}
We use this for analyzing ResPowMeth with $\z_\ell$ being the vectorized $\bPhi_\ell \U_\tau$. It is also used later for analyzing the GD step of the alternating GD and minimization (altGDmin) algorithm for solving the LRCS problem.

\subsection{Proof of Lemma \ref{key_lemma}} 
\label{proof_key_lemma}

Since $\SD_F(\U_\ell,\Ustar)=(1/\sqrt{2}) \norm{\pul - \pustar}_F$ \cite[Lemma 2.5]{spectral_dist}, thus, the lemma assumption implies that $ \max_{\ell \in \Jgood} \norm{\pul - \pustar}_F \le \sqrt{2}\delta$.

Observe that $\norm{\P_\U}_F\leq\sqrt{r}$ for any matrix $\U$ with orthonormal columns. Thus $\norm{\pul} \le \sqrt{r}$ for all $\ell$ including the Byzantine ones (recall that we orthonormalize the received $\Uhat_\ell$'s using QR at the center before computing $\pul$).
Hence,
 using GM Lemma \ref{gm_new_1}, 
	we have w.p. at least $1-c_\approxgm -\exp\left(-L \psi\left(0.4-\tau,p\right)\right)$
	\begin{equation}\label{eq:sd_GM_approx}
		\norm{\P_{GM}-\pustar}_F \leq 6\sqrt{2}\delta + 5\epsilon_\gm\sqrt{r}
	\end{equation}
Here $\P_{GM}= GM\{\pul,\ell \in [L]\}$. Thus, 
	\begin{align*}
		&  \max_{\ell \in \Jgood} \norm{\pul -\P_{GM}}_F \notag \\
		& \leq  \max_{\ell \in \Jgood} \norm{\pul - \pustar}_F + \norm{ \P_{GM} - \pustar}_F \notag \\
		& \leq  \sqrt{2}\delta + 6\sqrt{2} \delta + 5\eps_\gm\sqrt{r} = 7\sqrt{2}\delta + 5\eps_\gm\sqrt{r} 
	\end{align*}
w.p. at least $1-c_\approxgm-\exp(-L\psi(0.4-\tau,p))$.
	Next we bound the $\SD$ between $\P_{GM}$ and the node closest to it. This is denoted $\ell_{best}$ in the algorithm. 
	\begin{align*}
	\norm{\pulbst -\P_{GM}}_F & =	\min_{\ell}\norm{\pul - \P_{GM}}_F \\
		& \le \min_{\ell\in \Jgood}\norm{\pul - \P_{GM}}_F   \\
& \le \max_{\ell\in \Jgood}\norm{\pul - \P_{GM}}_F    \\
& \leq 7\sqrt{2}\delta + 5\epsilon_\gm\sqrt{r}
	\end{align*}
In this we used $\Jgood \subseteq [L]$ and hence the minimum value over all $L$ is smaller than that over all $\ell \in \Jgood$.
	We use this to bound the $\SD$ between $\U_{\ell_{best}}$ and $\Ustar$.
	\begin{align}
& \norm{\pulbst -\pustar}_F \notag \\
	& \le \norm{\pulbst -GM}_F + \norm{GM -\pustar}_F \notag \\
	& \leq 7\sqrt{2}\delta + 5\epsilon_\gm\sqrt{r} +  6\sqrt{2}\delta + 5\epsilon_\gm\sqrt{r} \notag \\
	& \leq  13\sqrt{2} \delta + 10\epsilon_\gm\sqrt{r}
\end{align}
Set $\eps_\gm = \delta\sqrt{2}/\sqrt{r}$. Thus,  we have that, w.p. at least $1-c_\approxgm-\exp(-L\psi(0.4-\tau,p))$,
\[\norm{\pulbst -\pustar}_F\leq  23 \sqrt{2}\delta\]
This then implies that $\SD_F(\U_{out},\Ustar) = \SD_F(\U_{\ell_{best}}, \Ustar) \leq 23\delta$ since $\U_{out} = \U_{\ell_{best}}$.

{\em Note: It is possible that $\ell_{best}$ is not a good node (we cannot prove that it is). This is why the above steps are needed to bound $\norm{\pulbst -\pustar}_F$.}

\subsection{Proof of Theorem \ref{th:ro_pca}, exact SVD at the nodes} \label{proof_main_res}
The version of Davis-Kahan $\sin \Theta$ theorem \cite{davis_kahan} stated next is taken from \cite[Corollary 2.8 ]{spectral_dist}.

\begin{claim}[Davis-Kahan  $\sin \Theta$ theorem \cite{davis_kahan}, \cite{spectral_dist}]\label{th:davis_kahan}
	Let $\bPhi^*,\bPhi$ be $n \times n$ symmetric matrices with $\Ustar\in\Re^{n\times r}$, $\U\in\Re^{n\times r}$ being the matrices of top $r$ singular/eigen vectors of $\bPhi^*,\bPhi$ respectively. Let $\sigmaonestar \geq ...\geq \sigma_n^*$ be the eigenvalues of $\bPhi^*$.
	If $\sigmarstar-\sigmarplusstar>0$ and $\norm{\bPhi-\bPhi^*}\leq\left(1-\frac{1}{\sqrt{2}}\right)(\sigmarstar-\sigmarplusstar)$
	then
	\[
\SD_F(\U,\Ustar)\leq \frac{2\sqrt{r}\norm{\bPhi-\bPhi^*}}{\sigmarstar-\sigmarplusstar}
\]
\end{claim}

Suppose that, for all $\ell \in \Jgood$,
	\begin{align*}
		\Pr\{\norm{\bPhi_\ell-\bPhi^*}&\leq  b_0 \} \geq 1-p
	\end{align*}
Using Claim \ref{th:davis_kahan}, if $b_0 < (1-1/\sqrt{2})\Delta$, this implies that, for all $\ell \in \Jgood$, w.p. at least $1-p$,
	\begin{align*}
		\SD_F(\U_\ell,\Ustar) 
		& \leq \frac{2\sqrt{r}b_0 }{\Delta}
	\end{align*}

Using Lemma \ref{key_lemma} with $\delta \equiv \frac{2\sqrt{r}b_0 }{\Delta}$, this then implies that,
 w.p. at least $1 - c_\approxgm - \exp(-L \psi(0.4-\tau,p) )$,   
	\begin{equation*}\label{eq:gm_new_use}
		\SD_F(\U_{out},\Ustar)\leq 23 \frac{2\sqrt{r}b_0 }{\Delta} = 46\sqrt{r}  \frac{b_0}{\Delta}
	\end{equation*}
To get the right hand side $\le \eps$ we need $b_0 \le \frac{\eps}{46 \sqrt{r}} \Delta$.

\subsection{Proof of Theorem \ref{th:ro_pca}: SVD at nodes computed using power method}\label{proof_main_res_part2}
This proof also needs to use Claim \ref{pm_lemma_1} given below (this is \cite[Theorem 1.1]{npm_hardt}) that analyzes each iteration of what the author calls ``noisy power method'' (power method that is perturbed by a noise/perturbation $\G_t$ in each iteration $t$. 

\begin{claim}\label{pm_lemma_1}[Noisy Power Method \cite{npm_hardt}]
	Let $\Ustar$ ($n\times r$) denote top $r$ singular vectors of a symmetric $n \times n$ matrix $\bPhi^*$, and let $\sigma_i$ denote it's $i-$th singular value.
Consider the following algorithm (noisy PM).
	\begin{enumerate}
		\item Let $\U_{rand}$ be an $n\times r$ matrix with i.i.d. standard Gaussian entries. Set $\U_{t=0} = \U_{rand}$.
		\item For $t=1$ to $T_{pow}$ do,
		\begin{enumerate}
			\item $\Uhat_{t} \leftarrow \bPhi^* \U_{t-1} + \G_{t}$
			\item $\Uhat_{t} \leftarrow QR(\Uhat_{t})$
		\end{enumerate}
	\end{enumerate}
	If at every step of this algorithm, we have
\begin{align*}
&	5\norm{\G_{t}}\leq \eps_{pow}(\sigmarstar-\sigmarplusstar), \\
& 5\norm{\Ustar{}^\top \G_{t}}\leq (\sigmarstar-\sigmarplusstar)\frac{\sqrt{r}-\sqrt{r-1}}{\gamma\sqrt{n}}
 \end{align*}
	for some fixed parameter $\gamma$ and $\eps_{pow}<1/2$. Then w.p. at least $1-\gamma^{-C_1} - \exp^{-C_2n}$, there exists a  $T_{pow} \ge C\frac{\sigmarstar}{\sigmarstar -\sigmarplusstar}\log (\frac{n\gamma}{\eps_{pow}})$ so that after $T_{pow}$ steps we have that \[\norm{(I-\U_{T_{pow}}\U_{T_{pow}}^\top) \Ustar }\leq \eps_{pow}\]
\end{claim}

 We state below a lower bound on $\sqrt{r}-\sqrt{r-1}$ based on Bernoulli's inequality.
	\begin{fact}\label{f:r_bound}
		Writing $\sqrt{r}-\sqrt{r-1}=\sqrt{r}\left(1-\sqrt{1-\frac{1}{r}}\right)$ and using Bernoulli's inequality $(1+r)^x\leq 1+xr$ for every real number $0\leq x\leq 1$ and $r\geq-1$ we have
	$\frac{1}{2\sqrt{r}}<\sqrt{r}-\sqrt{r-1}$
	\end{fact}

Suppose that, for all $\ell \in \Jgood$,
	\begin{align*}
		\Pr\{\norm{\bPhi_\ell-\bPhi^*}&\leq  b_0 \} \geq 1-p
	\end{align*}
Using Claim \ref{th:davis_kahan}, if $b_0 < (1-1/\sqrt{2})\Delta$, this implies that, for all $\ell \in \Jgood$, w.p. at least $1-p$,
	\begin{align}
		\SD_F(\U_\ell,\Ustar) 
		& \leq \frac{2\sqrt{r}b_0}{\Delta}
\label{kahan_bnd}
	\end{align}

Suppose that $\Uhat_\ell$ is an estimate of $\U_\ell$ computed using the power method. Next we use Claim \ref{pm_lemma_1} to help guarantee that $\SD_F(\Uhat_\ell,\U_l)$ is also bounded by $2\sqrt{r}b_0 /\Delta$.
Using Claim \ref{pm_lemma_1} with $\bPhi = \bPhi_\ell$, $\U = \U_\ell$, $\G_{\tau}=0$ for all $\tau$,  $\eps_{pow} = \frac{2b_0 }{\Delta}$, and $\gamma = n^{10}$,
we can conclude that if
$
T_{pow}>C\frac{\sigma_r(\bPhi_\ell)}{\sigma_r(\bPhi_\ell) -\sigma_{r+1}(\bPhi_\ell)}\log (\frac{n \cdot n^{10}}{\eps_{pow}}),
$
then $\SD_2(\Uhat_\ell, \U_\ell) \le \eps_{pow}=\frac{2b_0 }{\Delta}$ w.p. at least $1  -p- 1/n^{10}$. Here $\sigma_i = \sigma_i(\bPhi_\ell)$.
Using $\norm{\bPhi_\ell-\bPhi^*} \leq  b_0$ and Weyl's inequality, $\sigma_r -\sigma_{r+1} \ge \Delta - 2 b_0$ and $\sigma_r < \sigmarstar + b_0$.
Thus, if
\[
T_{pow} \ge C\frac{ \sigmarstar + b_0}{\Delta - 2b_0} \log ( n \frac{ \Delta}{ b_0})
\]
then
\[
\SD_2(\Uhat_\ell, \U_\ell) \le \eps_{pow} = \frac{2b_0}{\Delta}
\]
w.p. at least $1 -p - 1/n^{10}$. This then implies that $\SD_F(\Uhat_\ell, \U_\ell) \le \frac{2b_0 \sqrt{r}}{\Delta}$.

Combining this bound with the Davis-Kahan bound from \eqref{kahan_bnd}, we can conclude that, w.p. at least $1 -p - 1/n^{10}$,
\begin{align}
\SD_F(\Uhat_\ell,\Ustar) & \leq 2 \frac{2\sqrt{r}b_0}{\Delta} = 4 \sqrt{r} \frac{b_0}{\Delta}
\end{align}

Applying Lemma \ref{key_lemma} with $\delta \equiv  4 \sqrt{r}\frac{b_0}{\Delta}$, this then implies that,
 w.p. at least $1 - \exp(-L \psi(0.4-\tau,p + 1/n^{10}) )$,   
	\begin{equation}\label{eq:gm_new_use}
		\SD_F(\U_{out},\Ustar)\leq 23 \cdot  4 \sqrt{r} \frac{b_0}{\Delta} = 92 \sqrt{r}  \frac{b_0}{\Delta}
	\end{equation}
If we want the RHS of the above to be $\le \eps$, we need
\[
b_0 = \frac{\eps}{92 \sqrt{r}} \Delta
\]
and we need
$T_{pow} \ge C\frac{ \sigmarstar + b_0}{\Delta - 2b_0} \log ( n \frac{ \Delta}{ b_0}) $ with this choice of $b_0$.
By substituting for $b_0$ in the above expression, and upper bounding to simplify it, we get the following as one valid choice of $T_{pow}$
\[
T_{pow} = C (1 + 6\eps) \frac{\sigmarstar}{\Delta} \log ( n \frac{92\sqrt{r}}{\eps})
\]
This used $(1+\eps)(1-2 \eps)^{-1} < (1+\eps)(1 + 4 \eps) < 1+6 \eps$ for $\eps < 1$.
Since we are using $C$ to include all constants, and using $\eps < 1$, this further  simplifies to $T_{pow} = C  \frac{\sigmarstar}{\Delta} \log ( \frac{nr}{\eps}) $

\subsection{Proof of Theorem \ref{th:ro_pm}} \label{proof_respm}

We use Claim \ref{pm_lemma_1} with $\G_{t}=\bPhi^*\U - GM\{\bPhi_\ell\U\}_{\ell=1}^{L}$ and output $\U_{T_{pow}}\in\Re^{n\times r}$.
To apply it, we need $\|\G_t\|$ to satisfy the two bounds given in the claim. We use Lemma \ref{gm_new_2} to bound it.

Suppose that, for at least $\left(1-\tau\right) L$, $\bPhi_\ell$'s,
	\begin{align*}
		\Pr\{\norm{\bPhi_\ell-\bPhi^*}&\leq  b_0 \sigmax \} \geq 1-p
	\end{align*}
Since $\|\U\|_F = \sqrt{r}$, this implies
	\[
	\Pr\{	\norm{\bPhi_\ell\U-\bPhi^*\U}_F \leq  b_0 \sqrt{r}  \sigmax  \} \geq 1-p.
	\]
	We use this and apply Lemma \ref{gm_new_2} with $\z_\ell \equiv vec(\bPhi_\ell \U)$ and $\tz \equiv vec(\bPhi^* \U)$ so that $\|\tz\| = \|\bPhi^* \U\|_F \le \sigmax \sqrt{r}$. Setting
$\eps_\gm=b_0$ and applying the lemma, we have
	w.p. at least $1-c_\approxgm - Lp-\exp\left(-L\psi\left(0.4-\tau,p\right)\right)$
	\[
\norm{\G_t} \le \norm{\G_t}_F =	\norm{GM\{\bPhi_\ell\U\}_{\ell=1}^{L}-\bPhi^*\U}_F \leq  14  b_0 \sqrt{r}  \sigmax 
	\]
Recall that   $\sigmarstar-\sigmarplusstar \ge \Delta$.	We thus need $5\norm{\G_t}\leq \eps \Delta$ to hold. This will hold with high probability if $b_0 \sqrt{r}  \sigmax  \le \frac{\eps \Delta}{70}$.
Using Fact \ref{f:r_bound} and $\gamma=c$, for the second condition of Claim \ref{pm_lemma_1} to hold, we need
$\norm{\G_t} \leq\Delta\frac{1}{10c\sqrt{nr}}$. This then implies that we need $b_0\sqrt{r}  \sigmax  \le \frac{\Delta}{140c\sqrt{nr}}$.

Thus we can set $b_0 = \min\left( \frac{\eps}{70 \sqrt{r}}, \frac{1}{140c\sqrt{n}r}  \right)  \frac{\Delta}{\sigmax}$.

We also need $T_{pow} > C\frac{\sigmarstar}{\sigmarstar-\sigmarplusstar} \log (\frac{n\gamma}{\eps})$. This holds if we set $T_{pow} = C\frac{\sigmarstar}{\Delta} \log (\frac{nc}{\eps})$.

Hence 
w.p. at least $1-c_\approxgm Lp-\exp\left(-L\psi\left(0.4-\tau,p\right)\right)-c -e^{-C_2n}\geq 1-c_\approxgm - c-Lp-\exp\left(-L\psi\left(0.4-\tau,p\right)\right)$
\[\SD_F(\U_{out},\Ustar)\leq \epsilon\]

\subsection{Proof of Corollary \ref{th:ro_pca_A}}\label{proof_ro_pca_A}
The first part is a corollary of Theorem \ref{th:ro_pca} and \cite[Theorem 4.7.1]{versh_book} stated next. It gives a high probability bound on the error between an empirical covariance matrix estimate, $\hat{\Sig} = \D \D^\top/\q$, with the $\q$ columns of $\D$ being independent sub-Gaussian random vectors $\dstar_k$, and the true one, $\Sigmastar$.
\begin{claim}[\cite{versh_book}]\label{th:sigma_hpb}
Suppose that the matrix $\D$ is as defined in Sec. \ref{res_pca}. 
With probability at least $1-2\exp(-n)$,
	\[\norm{\hat{\Sig}-\Sigmastar}\leq CK^2\sqrt{\frac{n}{\q}}\norm{\Sigmastar}.\]
Here $K$ is the maximum sub-Gaussian norm of $\Sigmastar{}^{-1/2}\dstar_k$ over $k$.
\end{claim}
	Using Theorem \ref{th:ro_pca} with $\bPhi_\ell \equiv\hat{\Sig}_{\ell} = \D_\ell \D_\ell^\top / \q$, $\bPhi^*\equiv \Sigmastar$, in order to guarantee $\SD(\U_{out},\Ustar) \le \eps$ w.h.p.,
we need
	\[
\norm{\hat{\Sig}_{\ell}-\Sigmastar}\leq \frac{\epsilon\Delta}{92\sqrt{r}}
\]
By Claim \ref{th:sigma_hpb},
	\[\Pr\{\norm{\hat{\Sig}_{\ell}-\Sigmastar}\leq CK^2\sqrt{\frac{n}{\q}}\norm{\Sigmastar}\}\geq 1-2\exp(-n)\]
The above bound will be less than 	$\frac{\epsilon\Delta}{92\sqrt{r}}$ if
$
\q\geq \frac{92^2CK^4nr\norm{\Sigmastar}^2}{\Delta^2\epsilon^2}.
$

\subsection{Proof of Corollary \ref{th:sub_gmom}}
Corollary \ref{th:sub_gmom} is again a direct corollary of Theorem \ref{th:ro_pca} and \cite[Theorem 4.7.1]{versh_book}. We now apply both results on $\bPhi_{(\vartheta)}:= \sum_{\ell=1}^\rho \D_{(\vartheta,\ell)} \D_{(\vartheta,\ell)}^\top / (\q\rho)$, $\vartheta \in [\tL]$.
The reason this proof follows exactly as that for subspace median is because we assume that the set of Byzantine nodes is fixed across all iterations of this algorithm and the number of such nodes is lower by a factor of $L/\tL$. Consequently, for the purpose of the proof one can assume that no more than $\tau \tL$ mini-batches are Byzantine. With this, the proof remains the same once we replace $\q$  by $\q\rho$ and $L$ by $\tL$.

\section{Simulation Experiments} \label{sims}

All numerical experiments were performed using MATLAB on Intel(R)Xeon(R) CPU E3-1240 v5 @ 3.50GHz processor with 32.0 GB RAM.

\subsection{PCA experiments}

\subsubsection{Data generation}
We generated $\bPhi^* = \Ustar_{full} \S_{full} \Ustar_{full}{}^\top$, with $\Ustar_{full}$ generated by orthogonalizing an $n \times n$ standard Gaussian matrix; $\S_{full}$ is a diagonal matrix of singular values which are set as described below. This was generated once. The model parameters $n$, $r, q$, $L, L_{byz}$, and entries of $\S_{full}$ are set as described below in each experiment.

In all our experiments in this section, we averaged over 1000 Monte Carlo runs.
In each run, we sampled $q$ vectors from the Gaussian distribution, $\mathcal{N}(\mathbf{0}, \bPhi^*)$ to form the data matrix $\D$. This is split into $L$ columb sub-matrices, $\D_1, \D_2, \dots, \D_L$ with each containing $\q = q/L$ columns. $q,L$ are set so that $q/L$ is an integer.
Each run also generated a new $\U_{rand}$ to initialize the power method used by the nodes in case of SubsMed and used by the center in case of ResPowMeth. The same one was also used by the power methods for SubsMoM. Note: since SubsMed and SubsMoM run $\L$ and $\tL$ different power methods, ideally each could use a different $\U_rand$ and that would actually improve their performance. To be fair to all three methods, we generated $U_{rand}$ this way.

Let $L_{byz}=\tau L$. In all our experiments, we fixed $n = 1000$ and varied $r, q$, $L, L_{byz}$, and $\S_{full}$. In all experiments we used a large singular value gap (this ensures that a small value of $T_{pow}$ suffices).  We  experimented with three types of attacks described next.

\subsubsection{Attacks}
To our best knowledge, the PCA problem has not been studied for Byzantine resiliency, and hence, there are no known difficult attacks for it.
It is impossible to simulate the most general Byzantine attack. We focused on three types of attacks.
Motivated by reverse gradient (rev) attack \cite{rev_attack}, we generated the first one by colluding with other nodes to set $\U_{corrupt}$ as a matrix in the subspace orthogonal to that spanned by $\sum_{\ell}\hat{\U_\ell}$ at each iteration.
This is generated as follows. Let $\U =\sum_{\ell}\hat{\U_\ell}$ (in case of SubsMed, SubsMoM) and $\U = \sum_\ell\bPhi_\ell\U_t$ (for ResPowMeth). Orthonormalize it $\tilde{\U}=orth(\U)$ and let $\tilde{\M}=I-\tilde{\U}\tilde{\U}^\top$, obtain its QR decomposition $\tilde{\M} \stackrel{QR}{=} \U_{perp} R$ and set $\U_{corrupt}= (\thresh/\sqrt{r})  \U_{perp}(:,1:r)$. We call this \textit{Orthogonal attack}. Since SubsMed runs all its iterations locally, this is generated once for SubsMed, but it is generated at each iteration for ResPowMeth and SubsMoM.

The second attack that we call the  \textit{ones attack} consists of an $n\times r$ matrix of $-1$ multiplied by a large constant $C_{attack}$. The third attack that we call the  \textit{Alternating attack} is an  $n\times r$ matrix of alternating $+1$, $-1$ multiplied by a large constant $C_{attack}>0$. Values of $C_{attack}$ were chosen so that they do not get filtered out, essentially $0.9\thresh/\sqrt{nr}$.

\subsubsection{Algorithm Parameters}
For all geometric median (GM) computations, we used Weiszfeld's algorithm initialized using the average of the input data points. We set $T_\gm=10$. We vary $T_{pow}$.

\subsubsection{Experiments}
In all experiments, we compare ResPowMeth and SubsMed. In some of them, we also compare SubsMoM.  We also report results for the basic power method in the no attack setting.  To provide a baseline for what error can be achieved for a given value of $n,q,r$, we also report results for using ``standard power method'' in the no-attack setting; with this being implemented using power method with $T_{pow}$ iterations.
{\em Our reporting format is "$\max \SD_F(\text{mean} \SD_F), \text{mean time}$" in the first table and just "$\max \SD_F(\text{mean} \SD_F)$" in the others.} Here $\max \SD_F$ is the worst case error over all 1000 Monte Carlo runs, while $\text{mean} \SD_F$ is its mean over the runs.

In our first experiment, we let $n=1000$, $r=60$, $q= 1800$, $L=3, L_{byz}=1$, and we let $\S_{full}$ be a full rank diagonal matrix with first $r$ entries set to 15, the $r+1$-th entry to 1, and the others generated as $1-(1/n), 1-(2/n), \dots$.  Next, we simulated an approximately low rank $\Sigmastar$ by setting its  first $r$ entries set to 15, the $r+1$-th entry to 1, and the other entries to zero. We report results for both these experiments in Table \ref{exp1}.
As can be seen, from the first to the second experiment, the error reduces for both SubsMed and ResPowMeth, but the reduction is much higher for SubsMed.
Notice also that, for $Tpow = 1$, both ResPowMeth and SubsMed have similar and large errors with that of SubsMed being very marginally smaller. For $T_{pow}=10$, SubsMed has significantly smaller errors than ResPowMeth for reasons explained in the paper.  ResPowMeth has lower errors for the Orthogonal attack than for the other two; we believe the reason is that the Orthogonal attack changes at each iteration for ResPowMeth.

We also did some more experiments with (i) $L=3, L_{byz}=1$, $r=2,q=360$,  (ii) $L=6, L_{byz}=2$, $r=2,q=720$, and (iii) $L=6, L_{byz}=2$, $r=60,q=3600$.  All these results are reported in Table \ref{exp2}. Similar trends to the above are observed for these too.

In a third set of experiments, we used $L=18$, $r=60,q=3600$ and two values of  $L_{byz}$, $L_{byz}=2, L_{byz}=4$. For this one, we also compared SubsMoM with using $\tL=6$ minibatches. In the $L_{byz}=2$ case, SubsMoM has the smallest errors, followed by SubsMed. Error of ResPowMeth is the largest.
In the $L_{byz}=4$ case, ResPowMeth still has the largest errors. But in this case SubsMoM with $\tL=6$ also fails (when taking the GM of 6 points, 4 corrupted points is too large. SubsMed has the smallest errors in this case. We report results for these experiments in Table \ref{mom1}

\subsection{LRCS experiments}

In all experiments, we used $n=600$, $q=600$, $r=4$, $m=198$, and $L=18$ so that $\m=m/L=11$ and two values of $L_{byz}=1,2$. We simulated $\Ustar$ by orthogonalizing an $n\times r$ standard Gaussian matrix; and the columns $\bstar_k$ were generated i.i.d. from $\mathcal{N}(0,I_r)$. We then set $\Xstar=\Ustar\Bstar$. This was done once (outside Monte Carlo loop). For $100$ Monte Carlo runs, we generated matrices $\A_k,k\in[q]$ with each entry being i.i.d. standard Gaussian and we set $\y_k = \A_k \xstar_k$, $k \in [q]$.
In the figures we plot Error vs Iteration where $Error=\frac{\SD_F(\Ustar,\U)}{\sqrt{r}}$. We simulated the Reverse gradient (Rev) attack for the gradient step. In this case, malicious gradients are obtained by finding the empirical mean of the gradients from all nodes: $\nabla\leftarrow \frac{1}{L}\sum_{\ell=1}^{L}\nabla_\ell$ and set $\nabla_{mal}=-C\nabla$ where
$C = 10$. This forces the GD step to move in the reverse direction of the true gradient. We used step size $\eta=\frac{0.5}{\sigmax^2}$. We used Weiszfeld’s method to approximately compute geometric median.

We compare Byz-AltGDmin (Median) with Byz-AltGDmin (MoM) for both values of $L_{byz}$. We also provide results for the baseline algorithm - basic AltGDmin in the no attack setting. All these are compared in Figure \ref{byzaltgd}(b). We also compare the initialization errors in Figure \ref{byzaltgd}(a). As can be seen SubsMoM based initialization error is quite a bit lower than that with SubsMed. The same is true for the GDmin iterations.

\subsection{How to tune or set parameters for a real application}
Consider the PCA problem and Subspace Median.  Suppose that the user of the algorithm specifies the desired dimension $r$ and the desired final estimation error $\eps$. 
Our theoretical guarantee specifies that we need $T_{pow}=C\frac{\sigma_r^*}{\Delta}\log(\frac{n}{\eps})$, and $T_\gm=C\log(Lr/\eps)$. Given $r$ and $\eps$ (desired error), for setting $T_\gm$ all values are known. For $T_{pow}$ we need $\sigma_{r}^*$ and $\Delta$. To estimate these, at each node $\ell$, we compute the $r$-th and $(r+1)$-th singular values of $\bPhi_\ell = \D_\ell \D_\ell^\top/q_\ell$. Denote these by  $\hat{\sigma}_{\ell,r}$, and $\hat{\Delta}_\ell=\hat{\sigma}_{\ell,r}-\hat{\sigma}_{\ell,r+1}$ . We use $\max_\ell\{\hat{\sigma}_{\ell,r}\}$, $\min_{\ell}\{\hat{\Delta}_\ell\}$ in $T_{pow}$.
The constants $C$ in various expressions will typically need to be experimentally tuned for a given application. It should be noted that sufficiently large values of $T_\gm,T_{pow}$, e.g., setting both t0 10, works well for all algorithms without much change in final error. 
If the user does not specify $r$ we can  set $r$ using the well known 90 or 99\% energy threshold heuristic. We find $r_\ell$ as the smallest value of $r$ for which the sum of the top $r$ singular values is at least 90\% (or 99\% or similar) of the sum of all singular values of $\bPhi_\ell$. Instead of setting $\eps$ and $T_{pow}$, we can set a stopping criterion for the power method being implemented at each node: exit the loop if the estimates do not change much in subspace distance.

Consider LRCS.  For Byz-AltGDmin we set $\thresh$ as $\tilde{m}14\sqrt{r}\delta_{0}\max_\ell\{\hat{\sigma}_{\ell,1}\}$. The idea is to set the threshold $\thresh$ sufficiently large to ensure that non-Byzantine updates are not filtered out. For other parameters in Byz-AltGDmin $\tC=9\kappa^2\mu^2$, we set $\kappa=\max_{\ell}\{\kappa_\ell\}$, and take $\mu\geq 2$. We set $T=C\max_{\ell}\{\kappa_\ell^2\}\log(1/\eps)$, and $\eta=\frac{0.5}{\max_\ell\{\hat{\sigma}_{\ell,1}\}}$.

 \begin{table*}[t]
\begin{center}
\resizebox{0.9\linewidth}{!}{
	\subfloat[full rank $\Sigmastar$]{
\begin{tabular}{|c|c|c|}
			\hline
			\hline
			\textbf{Attacks}                      & \textbf{Methods}     & $T_{pow}=10$                          \\ \hline\hline
			\multirow{2}{*}{\textbf{Alternating}} & \textbf{SubsMed} & 0.375(0.348),0.680                      \\ \cline{2-3}
			& \textbf{ResPowMeth}   & 1.000(0.972),0.475                      \\ \hline\hline
			
			\multirow{2}{*}{\textbf{Ones}}        & \textbf{SubsMed} & 0.369(0.349),0.704                      \\ \cline{2-3}
			& \textbf{ResPowMeth}   & 0.999(0.990),0.513                      \\ \hline\hline
			\multirow{2}{*}{\textbf{Orthogonal}}  & \textbf{SubsMed} & 0.365(0.348),0.689                      \\ \cline{2-3}
			& \textbf{ResPowMeth}   & 0.999(0.366),0.500                      \\ \hline\hline
			\textbf{No Attack}                    & \textbf{Power(Baseline)}       & \multicolumn{1}{l|}{0.187(0.182),0.529} \\ \hline\hline
			
	\end{tabular}
}
\quad
	\subfloat[rank-$(r+1)$ $\Sigmastar$]{
\begin{tabular}{|c|c|c|c|}
			\hline\hline
			\textbf{Attacks}                      & \textbf{Methods}    & $T_{pow}=10$                           & $T_{pow}=1$     \\ \hline\hline
			\multirow{2}{*}{\textbf{Alternating}} & \textbf{SubsMed}    & 0.110(0.091),0.689                      & 0.999(0.614),0.326 \\ \cline{2-4}
			& \textbf{ResPowMeth} & 0.971(0.898),0.497                      & 1.000(0.991),0.049 \\ \hline\hline
			\multirow{2}{*}{\textbf{Ones}}        & \textbf{SubsMed}    & 0.111(0.091),0.669                      & 0.999(0.607),0.331 \\ \cline{2-4}
			& \textbf{ResPowMeth} & 0.992(0.952),0.477                      &  0.999(0.990),0.052 \\ \hline\hline
			\multirow{2}{*}{\textbf{Orthogonal}}  & \textbf{SubsMed}    & 0.106(0.091),0.672                      & 0.999(0.609),0.319\\ \cline{2-4}
			& \textbf{ResPowMeth} & 0.223(0.208),0.475                      & 0.999(0.993),0.048 \\ \hline\hline
			\textbf{No Attack}                    & \textbf{Power(Baseline)}      & \multicolumn{1}{l|}{0.063(0.050),0.505} & 0.999(0.605),0.050 \\ \hline\hline
	\end{tabular}
}
}
\caption{$n=1000$, $L=3$, $L_{byz}=1$, $r=60$, $q=1800$.
We report "$\max \SD_F(\text{mean} \SD_F), \text{mean time}$" in each column.
}
\label{exp1}
\end{center}
\end{table*}

\begin{table*}[t]
	\begin{center}
		\resizebox{0.9\linewidth}{!}{	\subfloat[$L=3, L_{byz}=1$, $r=2,q=360$]{\begin{tabular}{|c|c|c|l|}
						\hline	\hline
						\textbf{Attacks}                      & \textbf{Methods}    & $T_{pow}=10$                           &  $T_{pow}=1$     \\ \hline\hline
						\multirow{2}{*}{\textbf{Alternating}} & \textbf{SubsMed}    & 0.062(0.030)                      &  0.997(0.289) \\ \cline{2-4}
						& \textbf{ResPowMeth} & 0.177(0.084)                      &  0.999(0.424) \\ \hline\hline
						\multirow{2}{*}{\textbf{Ones}}        & \textbf{SubsMed}    & 0.080(0.030)                     &  0.989(0.236)  \\ \cline{2-4}
						& \textbf{ResPowMeth} & 0.196(0.087)                      &  0.972(0.311) \\ \hline\hline
						\multirow{2}{*}{\textbf{Orthogonal}}  & \textbf{SubsMed}    & 0.067(0.033)                      &  0.999(0.228) \\ \cline{2-4}
						& \textbf{ResPowMeth} & 0.125(0.066)                      &  0.999(0.375)  \\ \hline\hline
						\textbf{No Attack}                    & \textbf{Power(Baseline)}      & \multicolumn{1}{l|}{0.038(0.018)} &  0.968(0.211) \\ \hline\hline
			\end{tabular}}
			\quad
			\subfloat[$L=6, L_{byz}=2$, $r=2,q=720$]{
					\begin{tabular}{|c|c|c|l|}
						\hline	\hline
						\textbf{Attacks}                      & \textbf{Methods}    & $T_{pow}=10$                           &  $T_{pow}=1$     \\ 	\hline	\hline
						\multirow{2}{*}{\textbf{Alternating}} & \textbf{SubsMed}    & 0.045(0.020)                      &  0.986(0.227) \\ \cline{2-4}
						& \textbf{ResPowMeth} & 0.178(0.080)                      &  0.997(0.336) \\ 	\hline	\hline
						\multirow{2}{*}{\textbf{Ones}}        & \textbf{SubsMed}    & 0.048(0.020)                      &  0.999(0.275)\\ \cline{2-4}
						& \textbf{ResPowMeth} &  0.157(0.081)                      & 0.999(0.383) \\ 	\hline	\hline
						\multirow{2}{*}{\textbf{Orthogonal}}  & \textbf{SubsMed}    & 0.049(0.019)                      &  0.999(0.204)\\ \cline{2-4}
						& \textbf{ResPowMeth} & 0.102(0.057)                     &  0.999(0.339) \\	\hline	\hline
						\textbf{No Attack}                    & \textbf{Power}      & \multicolumn{1}{l|}{0.033(0.012)} &  0.975(0.203) \\ 	\hline	\hline
			\end{tabular}}}
		\end{center}
\qquad\qquad\qquad\qquad\qquad\qquad\qquad\qquad\qquad
			\subfloat[$L=6, L_{byz}=2$, $r=60,q=3600$]{	\begin{tabular}{|c|c|c|l|}
					\hline	\hline
					\textbf{Attacks}                      & \textbf{Methods}    & $T_{pow}=10$                           &  $T_{pow}=1$    \\ 	\hline	\hline
					\multirow{2}{*}{\textbf{Alternating}} & \textbf{SubsMed}    & 0.098(0.085)                      &  0.999(0.642) \\ \cline{2-4}
					& \textbf{ResPowMeth} & 0.992(0.853)                     &  1.000(0.988) \\ 	\hline	\hline
					\multirow{2}{*}{\textbf{Ones}}        & \textbf{SubsMed}    & 0.099(0.084)                      &  0.999(0.625) \\ \cline{2-4}
					& \textbf{ResPowMeth} &0.998(0.905)                      &  0.999(0.989)  \\ 	\hline	\hline
					\multirow{2}{*}{\textbf{Orthogonal}}  & \textbf{SubsMed}    & 0.103(0.084)                     &  0.999(0.610)\\ \cline{2-4}
					& \textbf{ResPowMeth} & 0.223(0.184)                     &  0.999(0.993) \\ 	\hline	\hline
					\textbf{No Attack}                    & \textbf{Power}      & \multicolumn{1}{l|}{ 0.043(0.036)} & 0.995(0.604) \\	\hline	\hline
		\end{tabular}}
\caption{Additional Experiments. We report "$\max \SD_F(\text{mean} \SD_F)$" in each column.
}\label{exp2}	
\end{table*}

\begin{table*}[t]
	\begin{center}
		\resizebox{0.8\linewidth}{!}{
			\subfloat[$L_{byz}=2$]{
				\begin{tabular}{|c|c|c|}
					\hline
					\textbf{Attacks}                      & \textbf{Methods}                         & $T_{pow}=10$                          \\ \hline\hline
					\multirow{3}{*}{\textbf{Alternating}} & \textbf{SubsMoM}                         & \multicolumn{1}{l|}{0.101(0.085)} \\ \cline{2-3}
					& \textbf{SubsMed}                         & 0.175(0.150)                      \\ \cline{2-3}
					& \textbf{ResPowMeth}                      & 0.522(0.463)                      \\ \hline\hline
					\multirow{3}{*}{\textbf{Ones}}        & \textbf{SubsMoM}                         & 0.098(0.085)                      \\ \cline{2-3}
					& \textbf{SubsMed}                         & 0.172(0.150)                      \\ \cline{2-3}
					& \multicolumn{1}{l|}{\textbf{ResPowMeth}} & \multicolumn{1}{l|}{0.542(0.502)} \\ \hline\hline
					\multirow{3}{*}{\textbf{Orthogonal}}  & \textbf{SubsMoM}                         & 0.104(0.086)                      \\ \cline{2-3}
					& \textbf{SubsMed}                         & \multicolumn{1}{l|}{0.172(0.151)} \\ \cline{2-3}
					& \textbf{ResPowMeth}                      & 0.223(0.191)                      \\ \hline\hline
					\textbf{No Attack}                    & \textbf{Power(Baseline)}                           & \multicolumn{1}{l|}{0.042(0.035)} \\ \hline
				\end{tabular}
			}
			\quad
			\subfloat[$L_{byz}=4$]{
				\begin{tabular}{|c|c|c|}
					\hline
					\textbf{Attacks}                      & \textbf{Methods}                         & $T_{pow}=10$                          \\ \hline\hline
					\multirow{3}{*}{\textbf{Alternating}} & \textbf{SubsMoM}                         & \multicolumn{1}{l|}{0.999(0.872)} \\ \cline{2-3}
					& \textbf{SubsMed}                         & 0.182(0.152)                      \\ \cline{2-3}
					& \textbf{ResPowMeth}                      & 0.987(0.894)                      \\ \hline\hline
					\multirow{3}{*}{\textbf{Ones}}        & \textbf{SubsMoM}                         & 1.000(1.000)                      \\ \cline{2-3}
					& \textbf{SubsMed}                         & 0.160(0.147)                     \\ \cline{2-3}
					& \multicolumn{1}{l|}{\textbf{ResPowMeth}} & \multicolumn{1}{l|}{0.999(0.948)} \\ \hline\hline
					\multirow{3}{*}{\textbf{Orthogonal}}  & \textbf{SubsMoM}                         & 1.000(1.000)                     \\ \cline{2-3}
					& \textbf{SubsMed}                         & \multicolumn{1}{l|}{0.183(0.151)} \\ \cline{2-3}
					& \textbf{ResPowMeth}                      & 0.216(0.179)                      \\ \hline\hline
					\textbf{No Attack}                    & \textbf{Power(Baseline)}                           & \multicolumn{1}{l|}{0.040(0.036)} \\ \hline
				\end{tabular}
			}
		}
		\caption{$L=18$, rank-$(r+1)$ $\Sigmastar$, $r=60$, $q=3600$, $T_\gm=10$, $\tL=6$ for SubsMoM.
We report "$\max \SD_F(\text{mean} \SD_F)$" in each column.
}\label{mom1}
	\end{center}
\end{table*}

\begin{figure*}[t]
	\begin{center}
		\subfloat[Initialization errors]{
\begin{tabular}{|c|c|c|}
				\hline
				\textbf{Method}  & $L_{byz}=1$  & $L_{byz}=2$  \\ \hline
				\textbf{SubsMed} & 0.716(0.665) & 0.717(0.667) \\ \hline
				\textbf{SubsMoM} & 0.477(0.457) & 0.475(0.459) \\ \hline
		\end{tabular}}
	\begin{subfigure}[Byz-AltGDmin]{0.45\textwidth}
				\includegraphics[width=4in]{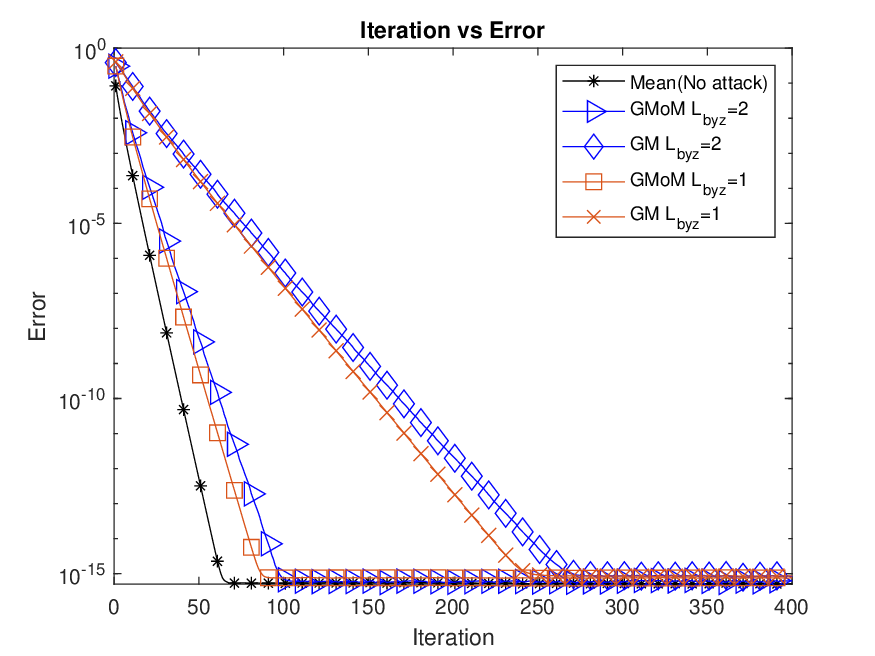}
		\end{subfigure}
	\end{center}
	\caption{Byz-AltGDmin (Median) vs Byz-AltGDmin (MoM) for $L_{byz}=1,2$; $L=18$}\label{byzaltgd}
\end{figure*}

\section{Conclusions, Extensions and Open Questions} \label{discuss}
Our work introduced a novel and well-motivated solution to Byzantine-resilient federated subspace estimation, and PCA, that is both communication-efficient and sample-efficient. We refered to this as ``Subspace-Median''. Its guarantee is provided in Theorem \ref{main_res} and Corollary \ref{th:ro_pca_A}.
We showed how the Subspace Median can be used to provably solve two practically useful problems: (i) Byzantine resilient federated PCA, and (ii) the initialization step of Byzantine-resilient horizontal federated LRCS.
We also developed Subspace Median-of-Means (MoM) extensions for both problems. These help improve the sample complexity at the cost of reduced Byzantine/outlier tolerance. For all these algorithms, Theorem \ref{main_res} helps prove sample, communication, and time complexity bounds for $\eps$-accurate subspace recovery. Extensive simulation experiments corroborate our theoretical results.
Our second important contribution is a provable communication-efficient and sample-efficient alternating GD and minimization (altGDmin) based solution to horizontally federated LRCS, obtained by using the Subspace Median to initialize the alternating GD and minimization (altGDmin) algorithm for solving it. 
Our proposed algorithms and proof techniques are likely to be of independent interest for many other problems. We describe some extensions next.

\subsection{Extensions}\label{extend}

One component that is missing in most existing work on Byzantine resilient federated GD, and stochastic GD, based solutions is how to initialize the GD algorithm in such a way that the problem becomes restricted strongly convex in the vicinity of the desired/true solution. 
Most existing works either assume  strongly convex cost functions or prove convergence to a local minimizer of a cost function. However, good initialization of the GD algorithm is a critical component for correctly solving a large number of practical problems. The spectral initialization approach has been extensively used for developing provably correct centralized iterative solutions to many non-convex optimization problems in signal processing and ML.  It involves computing the top, or top few, singular vectors of an appropriately defined matrix. This can be made Byzantine resilient and communication-efficient in a federated setting, by using the Subspace Median and Subspace MoM algorithms introduced in this work.
 Examples include LRCS, LR matrix completion, robust PCA using the LR plus sparse model, phase retrieval (PR), sparse PR, and low rank PR.

The overall approach developed here for modifying the altGDmin algorithm can also be widely used in other settings. In particular, vertically federated LRCS can be analyzed using easy extensions of our current work. It would require assuming that the Frobenius norm of the difference between column-sub matrices of $\Xstar$ that are sensed at the different nodes is bounded. This assumption is needed to ensure bounded heterogeneity of the different nodes' partial gradient estimates; a common assumption used in all past work on federated ML with heterogeneous nodes. Similar ideas can also extend for LR matrix completion, which also involves dealing with heterogeneous gradients. 
Vertically federated LRCS is the model that is relevant for federated sketching, and for multi-task representation learning when data for different tasks  is obtained at different nodes. 

Our Byzantine resilient PCA result can be generalized to extend it to various other PCA-based problems. Some examples are described in Remark \ref{gens} (PCA for non-i.i.d. data,  PCA for approximately LR datasets,
 PCA with missing data). Other examples include online PCA, subspace tracking, robust subspace tracking, differentially private PCA \cite{npm_hardt}.

\subsection{Open Questions}
In the current work, we treated the geometric median computation as a black box. Both for its accuracy and its time complexity, we relied on results from existing work. However, notice that in the Subspace Median algorithm (which is used in all other algorithms in this work), we need a ``median" of $r$-dimensional subspaces in $\Re^n$. These are represented by their $n \times r$ basis matrices. To find this though, we are computing the geometric median (GM) of vectorized versions of the subspace projection matrices $\P_\U:= \U \U^\top$ which are of size $n \times n$. Eventually we need to find the subspace whose projection matrix is closest to the GM. An open question is can we develop a more efficient algorithm to do this computation that avoids having to compute the GM of $n^2$ length vectors. We will explore the use of power method type ideas for modifying the GM computation algorithm in order to do this.
Alternatively, can we define a different notion of ``median'' for subspaces that can be computed more efficiently than Subspace Median.
Another related question is whether the computation can be federated to utilize the parallel computation power of the various nodes. In its current form, the entire GM computation is being done at the center.
A third open question is whether we can improve the guarantees for the Subspace Median of Means algorithms by using more sophisticated proof techniques, such as those used in \cite{yin2018byzantine}.

\appendices  \renewcommand\thetheorem{\Alph{section}.\arabic{theorem}}

\section{Proof of LRCS Initialization, Corollary \ref{u_init_gmom}}\label{init_lrccs_proof}

We prove this result for the $\tL=L$ setting below. The extension for the $\tL < L$ setting is straightforward. We explain this in Appendix \ref{gmom_extension_idea} below.

\subsection{Lemmas for proving Corollary \ref{u_init_gmom} for $\tL = L$} \label{proof_u_init}

We first state the lemmas from \cite{lrpr_gdmin} that are used in the proof and then provide the proof.
\begin{lemma}\label{lem:threshold}
	Define the set
	\[\mathcal{E}:=\left\{\tC(1-\epsilon_1)\frac{\norm{\Xstar}_F^2}{q}\leq \alpha_\ell\leq\tC(1+\epsilon_1)\frac{\norm{\Xstar}_F^2}{q}\right\}\]
	then $\Pr(\alpha\in\mathcal{E})\geq 1-L\exp(-\tilde{c}\m q\epsilon_1^2)$ where $\alpha = \text{Median}\{\alpha_\ell\}_{\ell=1}^{L}$
\end{lemma}

\begin{proof}
	Threshold computation: From \cite{lrpr_gdmin} Fact 3.7 for all $\ell\in\Jgood$
	\[\Pr\{\alpha_\ell\in\mathcal{E}\}\geq 1-\exp(-\tc\m q\eps_1^2)\] Since more than 75\% of $\alpha_\ell$'s are \textit{good} and the median is same as the 50th percentile for a set of scalars. This then implies $\text{Median}\{\alpha_\ell\}_{\ell=1}^{L}$ will be upper and lower bounded by \textit{good} $\alpha_\ell$'s. Taking union bound over  \textit{good} $\alpha_\ell$'s w.p. at least $1-L\exp(-\tc\m q\eps_1^2)=1-p_{\alpha}$
	\[\Pr\{\alpha\in\mathcal{E}\}\geq 1-L\exp(-\tc\m q\eps_1^2)\]
\end{proof}

\begin{lemma}[\cite{lrpr_gdmin}]\label{lm:init}
	Define $(\X_0)_\ell := \sum_{k} (\A_k)_\ell{}^\top ((\y_k)_\ell)_\trnc \e_k^\top, \ \y_{k,trnc}:=(\y_k \circ \bm{1}_{|\y_k| \le  \sqrt{\alpha}})$. Conditioned on $\alpha$, we have the following conclusions.
	\begin{enumerate}
		\item  Let $\zeta$ be a scalar standard Gaussian r.v.. Define,
		\[\beta_k(\alpha):=\E[\zeta^2\indic_{\{\norm{\xstar_k}^2\zeta^2\leq\alpha\}}]\]
		Then,
		\[\E[(\X_0)_\ell|\alpha] = \Xstar\D(\alpha)\]
		where $\D(\alpha)=diagonal(\beta_k(\alpha),k\in[q])$, i.e., $\D(\alpha)$ is a diagonal matrix of size $q\times q$ with diagonal entries $\beta_k$ defined above.
		\item Fix $0<\epsilon_1<1$. Then w.p. at least $1-\exp[(n+q)-c\eps_1^2\m q/ \mu^2\kappa^2]$
		\[\norm{(\X_0)_\ell  - \E[(\X_0)_\ell|\alpha]}\leq 1.1 \epsilon_1\norm{\Xstar}_F\]
		\item For any $\epsilon_1\leq 0.1$, $\min_k\E\left[\zeta^2\indic_{\left\{|\zeta|\leq\tC\frac{\sqrt{1-\eps_1}\norm{\Xstar}_F}{\sqrt{q}\norm{\xstar_k}}\right\}}\right]\geq 0.92$
	\end{enumerate}	
\end{lemma}

\begin{fact}\label{Dbound}
	For any $t > 0$, $\E[\zeta^2\indic_{\{\zeta^2\leq t\}}]\leq 1$, this then implies $\norm{\D(\alpha)}\leq 1$
\end{fact}

\subsection{Proof of Corollary \ref{u_init_gmom} for $\tL=L$} \label{proof_u_init_1}
	We will apply Theorem \ref{th:ro_pca} with $\bPhi_\ell \equiv (\X_0)_\ell(\X_0)_\ell^\top$, $\bPhi^* \equiv \E[(\X_0)_\ell|\alpha]\E[(\X_0)_\ell|\alpha]^\top = \Xstar\D(\alpha)^2\Xstar{}^\top$, $\eps=\delta_0$ and $\Delta= \sigma_{r}(\Xstar\D(\alpha)^2\Xstar{}^\top)-\sigma_{r+1}(\Xstar\D(\alpha)^2\Xstar{}^\top)$. For this we need to bound $\norm{(\X_0)_\ell(\X_0)_\ell^\top- \E[(\X_0)_\ell|\alpha]\E[(\X_0)_\ell|\alpha]^\top}$. We can write
	\begin{align}
		&(\X_0)_\ell(\X_0)_\ell^\top -\E[(\X_0)_\ell|\alpha]\E[(\X_0)_\ell|\alpha]^\top \notag\\
		&= (\X_0)_\ell((\X_0)_\ell - \E[(\X_0)_\ell|\alpha])^\top \notag \\
		&+ ((\X_0)_\ell - \E[(\X_0)_\ell|\alpha])\E[(\X_0)_\ell|\alpha]^\top\label{eq:quad}
	\end{align}
	To bound \eqref{eq:quad} we need the bounds on $(\X_0)_\ell$, $\E[(\X_0)_\ell|\alpha]$, and $(\X_0)_\ell - \E[(\X_0)_\ell|\alpha]$
	\begin{enumerate}
		\item From Lemma \ref{lm:init} part 2, letting $\eps_1 = \eps_3/\sqrt{r}$,
w.p. $1-\exp[(n+q)-c\eps_1^2\m q/ r \mu^2\kappa^2]$
		\begin{equation}\label{eq:init_1}
			\norm{(\X_0)_\ell  - \E[(\X_0)_\ell|\alpha]}\leq 1.1 \eps_3 \sigmax
		\end{equation}
		\item From Lemma \ref{lm:init} part 1, Fact \ref{Dbound}
		\begin{equation}\label{eq:init_2}
			\norm{\E[(\X_0)_\ell|\alpha]}\leq \sigmax
		\end{equation}
		\item Thus, for $\eps_3 < 0.1$,  w.p. $1-\exp[(n+q)-c\eps_3^2\m q/ r \mu^2\kappa^2]$
		\begin{align}
			\norm{(\X_0)_\ell}&=\norm{(\X_0)_\ell - \E[(\X_0)_\ell|\alpha]} +\norm{\E[(\X_0)_\ell|\alpha]}\notag\\
			&\leq 1.1 (1+\eps_3 )\sigmax < 1.3 \sigmax  \label{eq:init_3}
		\end{align}.
	\end{enumerate}
	
	From \eqref{eq:quad}, \eqref{eq:init_1},\eqref{eq:init_2} and \eqref{eq:init_3}, and $\E[(\X_0)_\ell|\alpha] = \Xstar\D(\alpha)$, we have w.p. at least $1-2\exp[(n+q)-c\eps_3^2\m q/ r \mu^2\kappa^2]$	
	\begin{align}
		&\norm{(\X_0)_\ell(\X_0)_\ell^\top - \Xstar\D(\alpha)^2\Xstar{}^\top} \leq  \notag \\
		& \norm{(\X_0)_\ell}\norm{(\X_0)_\ell - \Xstar\D(\alpha)} + \norm{(\X_0)_\ell - \Xstar\D(\alpha)}\norm{\D(\alpha)\Xstar} \notag \\
		& \leq 1.1 . 1.1 \eps_3 \sigmax^2 + 1.1 \eps_3 \sigmax^2 < 2.5 \eps_3 \sigmax^2
\label{bound_Phi_ell}
	\end{align}

To apply Theorem \ref{th:ro_pca} to get $\SD_F(\U_{out},\Ustar) < \delta_0$, we need $\norm{(\X_0)_\ell(\X_0)_\ell^\top-\Xstar\D(\alpha)^2\Xstar{}^\top} < \frac{\delta_0}{26\sqrt{r}}  \Delta$.
By Lemma \ref{lm:init} part 3, Fact \ref{Dbound}, and the fact that $\Xstar$ is rank $r$ we get a lower bound on $\Delta$
	\begin{equation}\label{eq:init_4}
		\Delta\geq 0.92^2\sigmin^2 -0 > 0.8 \sigmin^2
	\end{equation}
Using \eqref{eq:init_4} and \eqref{bound_Phi_ell}, the required condition for Theorem \ref{th:ro_pca} holds if
\[
2.5 \eps_3 \sigmax^2 \le 0.8 \frac{\delta_0}{26\sqrt{r}}  \sigmin^2
\]
This will hold if we set $\eps_3 = \frac{c}{\sqrt{r} \kappa^2} \delta_0 $.
With this choice of $\eps_3$, the bounds hold  w.p. at least $1-2\exp[(n+q)-c\delta_0^2\m q/ r^2 \mu^2\kappa^6]$

Thus, by Theorem \ref{th:ro_pca},	
	\[
	\Pr\{SD_F(\U_{out},\Ustar)\leq \delta_0|\alpha\}  \geq 1-c_0 - \exp(-L\psi(0.4-\tau,p + n^{-10}))
	\]
where $p = 2\exp[(n+q)-c \delta_0^2 \m q/ r^2 \mu^2\kappa^6]$.
	
	Following the same argument as given in proof of  \cite[Theorem 3.1]{lrpr_gdmin} and using Lemma \ref{lem:threshold} to remove the conditioning on $\alpha$, we get
	\begin{align*}
& \Pr\{SD_F(\U_{out},\Ustar)\leq \delta_0\} \\
& \geq 1-c_0 - \exp(-L\psi(0.4-\tau,p+ n^{-10})) -p_{\alpha}
\end{align*}
	where $p_{\alpha}=L\exp(-\tilde{c}\m q\delta_0^2/r^2\kappa^4)$.

If $\m q\geq C \kappa^6\mu^2(n+q)r^2/\delta_0^2$, then $p < e^{-c(n+q)}$, and $p_{\alpha} < e^{-c(n+q)}$.

Thus, the good event holds w.p. at least $1-c_0 - \exp(-L\psi(0.4-\tau, e^{-c(n+q)} + n^{-10} ) ) - e^{-c(n+q)}$.

\subsection{Proof of Corollary \ref{u_init_gmom} for a $\tL < L$}\label{gmom_extension_idea}
In this case, we apply Theorem \ref{th:ro_pca} on $\bPhi_{(\vartheta)} = \sum_{\ell=1}^\rho (\Xhat_0)_{(\vartheta,\ell)} (\Xhat_0)_{(\vartheta,\ell)}^\top / (\m \rho)^2$ and $\bPhi^* = \E[(\Xhat_0)_\ell |\alpha] \E[(\Xhat_0)_\ell |\alpha]^\top /\m^2$ for $\vartheta \in [\tL]$. To obtain the bounds needed to apply Theorem \ref{th:ro_pca}, we use the bounds from above.

\section{Proofs for LRCS AltGDmin iterations} \label{proof_gdmin_byz}

We prove this for the simple GM setting because that is notation-wise simpler. This is the $\tL=L$ setting. The extension to GMoM is straightforward once again.

{\em All expected values used below are expectations conditioned on past estimates (which are functions of past measurement matrices and measurements, $\A_k,\y_k$). For example, $\E[(\nabla_U f)_{\ell}]$ conditions on the values of $\U, \B_\ell$ used to compute it. This is also the reason why $\E[(\nabla_U f)_{\ell}]$ is different for different nodes; see Lemma \ref{new_term_lem}.
}

\subsection{Lemmas for proving Theorem \ref{gdmin_thm} for $\tL=L$: LS step bounds}

The next lemma bounds the 2-norm error between $(\b_k)_\ell$ and an appropriately rotated version of $\bstar_k$, $\g_k:= \U^\top \xstar_k = (\U^\top \Ustar) \bstar_k$; followed by also proving various important implications of this bound.
Here and below $\U$ denotes the subspace estimate at iteration $t$.

\begin{lemma}[Lemma 3.3 of \cite{lrpr_gdmin}]
	Assume that $\SD_F(\U,\Ustar) \le \delta_t $. 
Consider any  $\ell \in \Jgood$.
Let $\g_k:= \U^\top \xstar_k = (\U^\top \Ustar) \bstar_k$.

If $\delta_t \le 0.02 / \kappa^2$,
	and	if $\m \gtrsim \max(\log q, \log n, r)$, for $\epsilon_1<0.1$ then,
\\
w.p. at least, $1-\exp(\log q +r -c\epsilon_1^2\m)$
	\ben
	\item $\|(\b_k)_\ell - \g_k\| \le 1.7\epsilon_1 \delta_t \|\bstar_k\|$	
	\item $\|(\b_k)_\ell\|  \le 1.1 \|(\bstar_\ell)_k\|$	
	\item $\|\B_\ell - \G\|_F  \le 1.7\epsilon_1  \delta_t \sigmax $	
	\item $\|(\x_\ell)_k - \xstar_k\| \le 1.4  \delta_t \|\bstar_k\|$	
	\item $\|\X_\ell - \Xstar\|_F \le 1.4 \delta_t \sigmax$	
	\item $\sigma_{r}(\B_\ell) \ge  0.9 \sigmin$	
	\item $\sigma_{\max}(\B_\ell) \le  1.1 \sigmax$
	\een
	(only the last two bounds require the upper bound on $\delta_t$).
	\label{Blemma_new}	
\end{lemma}

All the lemmas given below for analyzing the GD step use Lemma \ref{Blemma_new} in their proofs.

\subsection{Lemmas for proving Theorem \ref{gdmin_thm} for $\tL=L$: GD step bounds} 
The main goal here is to bound $\SD_F(\U^+,\Ustar)$, given that $\SD_F(\U,\Ustar) \le \delta_t$. Here $\U^+$ is the subspace estimate at the next, $(t+1)$-th iteration.
We will show that $\SD_F(\U^+,\Ustar) \le (1 - (\eta \sigmax^2) \frac{c}{\kappa^2}) \delta_t$.
In our previous work \cite{lrpr_gdmin, lrpr_gdmin_2}, we obtained this by bounding the deviation of the gradient, $\nabla f = \sum_{k \in [q]} \nabla f_k$ from its expected value, $\E[\nabla f] = m (\X - \Xstar) \B^\top$ and then using this simple expression for the expected gradient to obtain the rest of our bounds. In particular notice that $\P_{\Ustar,\perp} \E[\nabla f] = m \P_{\Ustar,\perp} \U \B \B^\top$.

In this work, to use the same proof structure, we need a proxy for $\E[\nabla f]$. For this, we can use $\E[(\nabla f)_\ell]$ for any $\ell \in \Jgood$. We let $\ell_1 \in \Jgood$ be one such node. In what follows, we will use $\E[\nabla f_{\ell_1}] =  m (\X_{\ell_1} - \Xstar) \B_{\ell_1}^\top$ at various places.

\begin{lemma}(algebra lemma)\label{sd_pop}
Let
\[
\Err=\nabla f_\gm-\E[\nabla f_{\ell_1}(\U,\BGM)].
\]
Recall that $\U^+ = QR(\U - (\eta/\m) \nabla f_\gm)$.  We have
	\begin{align*}
		\SD_F(\Ustar,\U^+)
\leq \frac{ \| \I_{r}-\eta\B_{\ell_1} \B_{\ell_1}^\top \| \SD_F(\Ustar,\U)  + \frac{\eta}{\m} \| \Err\|_F }{ 1-\frac{\eta }{\m}\|\E[\nabla f_{\ell_1}(\U,\BGM)] \|  -\frac{\eta }{\m} \|\Err\|}
	\end{align*}
\end{lemma}

\begin{proof} See Appendix \ref{sd_pop_proof}. \end{proof}  

\begin{lemma}\label{diff_terms}
	Assume $\SD_F(\Ustar,\U)\leq \delta_t<\delta_0$. 
	\begin{enumerate}
		\item If $\eta \le 0.5/\sigmax^2$,
		{\small
			$
			\lambda_{\min}(\I_r - \eta \B_{\ell_1} \B_{\ell_1}^\top) = 1 - \eta \|\B_{\ell_1}\|^2 > 1 - \eta 1.3 \sigmax^2 > 0
			$
		}
		and so this matrix is p.s.d. and hence,
		{\small
			$
			\|\I_r - \eta \B_{\ell_1} \B_{\ell_1}^\top \| = \lambda_{\max}(\I_r - \eta \B_{\ell_1} \B_{\ell_1}^\top) = 1 - \eta \sigmamin(\B_{\ell_1})^2 \le 1 - \eta 0.8 \sigmin^2
			$
		}
		\item  For all $\ell \in \Jgood$,
		\[
		\E[\nabla f_{\ell}(\U,\B_\ell)] =\m (\X_\ell-\Xstar)\B_\ell^\top
		\]
		
		\item For all $\ell \in \Jgood$, 
		\[
		\norm{\E[\nabla f_{\ell}(\U,\B_\ell)]}_F \leq \m 1.6 \delta_t\sigmax^2
		\]
		
	\end{enumerate}
\end{lemma}

\begin{proof}
The first item follows using the bounds on $\sigmax(\B_\ell)$ and $\sigmin(\B_\ell)$ from Lemma \ref{Blemma_new}. Second item is immediate. Third item follows from item two and bounds on $\sigmax(\B_\ell)$, $\|\X_\ell - \Xstar\|_F$ given in Lemma \ref{Blemma_new}.
\end{proof}

The next lemma is an easy consequence of Lemmas \ref{new_term_lem} and Lemma \ref{gm_new_2}.
\begin{lemma}\label{err_bound_lemma}
	Let $p_1=\exp \Big((n+r)-c\epsilon_1^2\frac{\m q}{r\mu^2}\Big)+2\exp(\log q +r -c\epsilon_1^2\m)$. 
	If  $\tau< 0.4$, then, w.p. at least $1-Lp_1-\exp ( -L\psi(0.4-\tau,p_1) )$,	
	\[
	\norm{\Err}_F \leq 14 . 12.5 \m  \sigma_{\max}^{\ast 2}\epsilon_1\sqrt{r}\delta_t
	\]	
\end{lemma}

\begin{proof} See Appendix \ref{err_bound_lemma_proof}. \end{proof}

\subsection{Proof of Theorem \ref{gdmin_thm}}
The proof is an easy consequence of the above lemmas. Using the bounds from Lemma \ref{err_bound_lemma}, \ref{diff_terms} and the $\SD_F$ bound from Lemma \ref{sd_pop}, setting $\eps_1 = 0.3/175\sqrt{r}\kappa^2$, and using $\delta_t \le \delta_0 = 0.1/\kappa^2$ in the denominator terms, we conclude the following: if in each iteration, $\m q\geq C_1 \kappa^4\mu^2(n+r)r^2$, $\m>C_2\max(\log q, \log n)$, then, w.p. at least $1-Lp_1-\exp(-L\psi(0.4-\tau,p_1))$, where $p_1=\exp \Big((n+r)-c\frac{\m q}{r^2\kappa^4\mu^2}\Big)+2\exp(\log q +r -c\m/\kappa^4)$ 
\[
\SD_F(\Ustar, \U^+)  \le \frac{(1 - \frac{0.8 \eta \sigmax^2}{\kappa^2} ) \delta_t + \frac{0.3 \eta \sigmax^2}{\kappa^2} \delta_t}{1 - \frac{1.6 \cdot 0.1 \eta \sigmax^2}{\kappa^2} - \frac{0.1 \cdot 0.3 \eta \sigmax^2}{\kappa^2}}
\]
$\le (1 - (\eta \sigmax^2) \frac{c}{\kappa^2}) \delta_t: = \delta_{t+1}$
Applying this bound at each $t$ proves the theorem.

The numerical constants may have minor errors in various places.%

\subsection{Proof of Algebra lemma, Lemma \ref{sd_pop}} \label{sd_pop_proof}
Recall that
$\Err=\nabla f_\gm-\E[ \nabla f_{\ell_1} (\U,\B_{\ell_1})]$. Let $\P:=\I-\U^*\U^{*T}$.

GD step is given as
$
\Uhat^+= \U-\frac{\eta }{\m} \nabla f_\gm.
$

Adding and subtracting $\E[ \nabla f_{\ell_1} (\U,\B_{\ell_1})]  = \m(\X_{\ell_1}-\Xstar)\B_{\ell_1}^\top$, we get
\begin{align}
	\Uhat^+ 
	& = \U -\frac{\eta}{\m} \m (\U \B_{\ell_1}-\Xstar)\B_{\ell_1}^\top - \frac{\eta }{\m} \Err
\end{align}
Multiplying both sides by $\P:= \I-\Ustar \Ustar{}^\top$, 
\begin{align}
	\P\Uhat^+ 
	&= \P\U -\eta \P\U\B_{\ell_1}\B_{\ell_1}^\top -\frac{\eta }{\m}\P \Err \notag \\
	&= \P\U (\I_r - \eta \B_{\ell_1}\B_{\ell_1}^\top ) -\frac{\eta }{\m}\P \Err
	\label{eq:gd_pop_p}
\end{align}

Taking Frobenius norm and using $\norm{M_1M_2}_F\leq \norm{M_1}_F\norm{M_2}$ we get
\begin{equation}
	\| \P\Uhat^+ \|_F \leq \norm{\P\U}_F \norm{I_r - \eta \B_{\ell_1}\B_{\ell_1}^\top} + \frac{\eta }{\m}\norm{\P \Err}_F
\end{equation}

Now $\Uhat^+ \stackrel{QR}{=} \U^{+}R^{+}$ and since $\| M_1M_2 \|_F \leq \| M_1\|_F \| M_2\|$, this means that $\SD(\U^*,\U^+)\leq \| (I-\U^*\U^{*T})\Uhat^+\|_F\| (R^+)^{-1}\|$. Since $\| (R^+)^{-1}\|=1/\sigma_{min}(R^+)=1/\sigma_{min}(\Uhat^+)$, 
\begin{gather*}
	\| (R^+)^{-1}\|=\frac{1}{\sigma_{min}(\U-\frac{\eta }{\m}(\E[\nabla f_{\ell_1} (\U,\B_{\ell_1})] + \Err))} \\
	\leq \frac{1}{1-\frac{\eta }{\m}\norm{\E[\nabla f_{\ell_1}(\U,\B_{\ell_1})]} - \frac{\eta }{\m} \norm{\Err}}
\end{gather*}
Combining the last two bounds proves our result.

\subsection{Bounding $\|\nabla f_\ell(U,B_{\ell})  - \E[ \nabla f_{\ell_1} (U,B_{\ell_1})] \|_F$: Proof of Lemma \ref{new_term_lem}} \label{new_term_lem_proof}
From the proof of \cite[Lemma 3.5 item 1 ]{lrpr_gdmin} we can write w.p. at least $1-\exp \Big((n+r)-c\epsilon_1^2\frac{\m q}{r\mu^2}\Big)-2\exp(\log q +r -c\epsilon_1^2\m)$
\begin{equation}\label{eq:gd_diff}
	\norm{\nabla f_{\ell}-\E[\nabla f_{\ell}]}_F\leq 1.5\epsilon_1\sqrt{r}\delta_t\m\sigma_{\max}^{\ast 2}
\end{equation}

Using \eqref{eq:gd_diff} and Lemma \ref{diff_terms}, item 2,
\begin{align}
	&\norm{\nabla f_{\ell}-\E[\nabla f_{\ell_1}]}_F \leq  \notag \\
	& \norm{\nabla f_{\ell}-\E[\nabla f_{\ell}]}_F + \norm{\E[\nabla f_{\ell}]-\E[\nabla f_{\ell_1}]}_F \leq  \notag \\
	& 1.5\epsilon_1\sqrt{r}\delta_t\m\sigma_{\max}^{\ast 2} \notag \\
	&+ \norm{\m(\X_{\ell}-\Xstar)\B_\ell^{\top}-\m(\X_{\ell_1}-\Xstar)\B_{\ell_1}^{\top}}_F \label{eq:b_1}	
\end{align}
Using the bounds from Lemma \ref{Blemma_new},
\begin{align*}
	&\norm{\m(\X_{\ell}-\Xstar)\B_{\ell}^{\top}-\m(\X_{\ell_1}-\Xstar)\B_{\ell_1}^{\top}}_F\\
	&=\m\norm{\U(\B_\ell\B_{\ell}^{\top}-\B_{\ell_1}\B_{\ell_1}^{\top})-\Xstar(\B_{\ell}^{\top}-\B_{\ell_1}^{\top})}_F\\
	&=\m\norm{\U(\B_\ell\B_{\ell}^{\top}-\B_{\ell_1}\B_{\ell_1}^{\top}\pm \B_\ell B_{\ell_1}^{\top}) - \Xstar(\B_{\ell}-\B_{\ell_1})^{\top}}_F \\
	& = \m\norm{\U\B_{\ell}(\B_{\ell}^{\top}-\B_{\ell_1}^{\top})-\U(\B_{\ell_1}-\B_{\ell})\B_{\ell_1}^{\top}- \Xstar(\B_{\ell}^{\top}-\B_{\ell_1}^{\top})}_F \\
	&\leq \m(1.1\sigmax+1.1\sigmax + \sigmax)\norm{\B_\ell-\B_{\ell_1}}_F 	\\
	&= \m3.2\sigmax\norm{\B_\ell-\B_{\ell_1}\pm\G}_F\\
	&\leq \m3.2\sigmax(\norm{\B_\ell-\G}_F + \norm{\B_{\ell_1}-\G}_F)  \leq \m11\sigmax^2\epsilon_1\sqrt{r}\delta_t
\end{align*}
Using this in \eqref{eq:b_1}
\begin{align}
	\norm{\nabla f_{\ell}-\E[\nabla f_{\ell_1}]}_F
	& \leq 1.5\epsilon_1\sqrt{r}\delta_t\m \sigmax^2+ 	 \m11\sigmax^2 \epsilon_1\sqrt{r}\delta_t \notag \\
	& \leq \m12.5 \sigmax^2 \epsilon_1\sqrt{r}\delta_t
\end{align}
w.p. at least $1-\exp \Big((n+r)-c\epsilon_1^2\frac{\m q}{r\mu^2}\Big)-2\exp(\log q +r -c\epsilon_1^2\m)$.

\subsection{Bounding $Err$: Proof of Lemma \ref{err_bound_lemma}} \label{err_bound_lemma_proof}
Recall that $\Err = \nabla f_\gm - \E[ \nabla f_{\ell_1} (\U,\B_{\ell_1})]$.
This bound follows from the Lemma \ref{gm_new_2} and Lemma \ref{new_term_lem}. 
We apply  Lemma \ref{gm_new_2}  with $z_\ell \equiv \nabla f_\ell$ and $\tilde{z} \equiv  \E[ \nabla f_{\ell_1} (\U,\B_{\ell_1})]$, $\alpha = 0.4$, $\tau < 0.4$, $\epsilon\equiv 7.8\epsilon_1 =\eps_\gm$ and $\thresh$ set to a constant $C$ times an upper bound on $\|\E[ \nabla f_{\ell_1}]\|_F$. From Lemma \ref{diff_terms}, $\|\E[ \nabla f_{\ell_1}]\|_F \le 2 \m \delta_t \sigmax^2 \le 2 \m \delta_0 \sigmax^2 $. The Theorem needs $\delta_0 = c/\kappa^2$. Thus, we can set $\thresh = C \m \sigmin^2$.

To apply Lemma \ref{gm_new_2}, we need a high probability bound on
$
\max_{\ell \in \Jgood} \| \ \nabla f_\ell - \E[ \nabla f_{\ell_1} (\U,\B_{\ell_1})]\|_F. 
$

By  Lemma  \ref{new_term_lem} and union bound and using $|\Jgood|=(1-\tau)L \le L$, we can show that 
\begin{align}
	\max_{\ell \in \Jgood} \|\nabla f_\ell(\U,\B_\ell)  - \E[ \nabla f_{\ell_1} (\U,\B_{\ell_1})] \|_F
	&\le \m12.5\sigma_{\max}^{\ast 2}\epsilon_1\sqrt{r}\delta_t
\end{align}
w.p. at least $1- L \left( \exp \Big((n+r)-c\epsilon_1^2\frac{\m q}{r\mu^2}\Big) + 2 \exp(\log q +r -c\epsilon_1^2\m) \right):=1-p_1$.

Thus, using Lemma \ref{gm_new_2} w.p. at least $1-Lp_1-\exp ( -L\psi(0.4-\tau,p_1) )$,	
\[
\|\Err\|_F \le 14\m12.5\sigma_{\max}^{\ast 2}\epsilon_1\sqrt{r}\delta_t
\]

\section{One step Analysis of ResPowMeth} \label{direct_respow}
If we want to analyze ResPowMeth directly, we need to bound $\SD(\Uhat, \Ustar)$ at each iteration. Consider its first iteration.

Let $\P_{\Ustar,\perp}=\I-\Ustar\U^{\ast\top}$, $GM=GM\{\bPhi_\ell\U_{rand}\}_{\ell=1}^{L}$.
Then,
\[
\SD_F(\Ustar,\hat{\U})\leq \| \P_{\Ustar,\perp} GM\|_F \| (R^+)^{-1}\| \le  \frac{\| \P_{\Ustar,\perp} GM\|_F}{\sigma_{min}(GM)}
\]
Here $GM\stackrel{QR}{=} \hat{\U}R^{+}$. This follows since  $\| (R^+)^{-1}\|=1/\sigma_{min}(R^+)=1/\sigma_{min}(GM)$.

To bound both numerator and denominator, we use the fact that $GM$ is an approximation of $\bPhi^* \U_{rand}$.

Suppose that
\begin{align*}
	\norm{\bPhi_\ell-\bPhi^*}&\leq  b_0.
\end{align*}
Using \cite[Theorem 4.4.5]{versh_book}, $\|\U_{rand}\|\le 1.1 (\sqrt{n} + \sqrt{r})$, and so,
\begin{align*}
\norm{\bPhi_\ell\U_{rand}-\bPhi^*\U_{rand}}&\leq  b_0 \|\U_{rand}\| \le  2.2  b_0 \sqrt{n}  
\end{align*}
where we used $r \le n$.
Using this and applying Lemma \ref{gm_new_2} with $\eps_\gm=b_0 $, we have
w.p. at least $1-c_0 -Lp-\exp\left(-L\psi\left(0.4-\tau,p\right)\right)$
\[
\norm{GM -\bPhi^*\U_{rand}} \leq  31 b_0 \sqrt{n} 
\]
Then,
\begin{align}
	&\norm{\P_{\Ustar,\perp} GM}_F\notag \\&=\norm{\P_{\Ustar,\perp}(GM-\bPhi^*\U_{rand}+\bPhi^*\U_{rand})}_F\notag \\
	&=\norm{\P_{\Ustar,\perp}(GM-\bPhi^*\U_{rand})}_F \notag  + \norm{\P_{\Ustar,\perp} \bPhi^*\U_{rand})}_F  \notag \\
	&\leq  \norm{GM-\bPhi^*\U_{rand}}_F + \sigma_{r+1}^* \norm{\U_{rand}} \sqrt{r} \notag \\
	&=31 b_0 \sqrt{n} +  \sigma_{r+1}^* \cdot 2.2  \sqrt{n} \cdot \sqrt{r}   \label{eq:pu_init_new}
\end{align}
where we used $\norm{\P_{\Ustar,\perp} \bPhi^*\U_{rand})}_F  \le \norm{\P_{\Ustar,\perp} \bPhi^*} \norm{\U_{rand})}_F  \le  \sigma_{r+1}^* \norm{\U_{rand}} \sqrt{r}$.
Also,
\begin{align}
	&\sigma_{min}(GM) \\
	&\geq \sigma_{min}(\bPhi^*\U_{rand}) - \norm{\bPhi^*\U_{rand}-GM} \notag \\
	&\geq \sigma_{min}(\bPhi^*\U_{rand}) - 31b_0 \sqrt{n}    \notag \\
	&\geq \sigma_{min}(\P_{\Ustar} \bPhi^*\U_{rand}) - \norm{\P_{\Ustar,\perp} \bPhi^*\U_{rand}}  - 31b_0 \sqrt{n}  \notag \\
	&\geq \sigma_{min}(\P_{\Ustar} \bPhi^*\U_{rand}) - 2.2 \sigma_{r+1}^* \sqrt{n}  - 31b_0 \sqrt{n}
\label{eq:sig_min}
\end{align}
where we used Weyl's inequality and $\bPhi^*\U_{rand} = \P_{\Ustar} \bPhi^*\U_{rand} + \P_{\Ustar,\perp} \bPhi^*\U_{rand}$
Finally,
\begin{align}
	\sigma_{min}(\P_{\Ustar} \bPhi^*\U_{rand})& = \sigma_{min}(\Ustar\Sig\Ustar{}^\top\U_{rand})  \notag \\
	&\geq \sigma_{\min}(\Ustar)\sigmarstar\sigma_{min}\left(\U^{*\top}\U_{rand}\right) \label{eq:t}
\end{align}

We bound $\sigma_{min}\left(\U^{*\top}\U_{rand}\right)$ using \cite[Theorem 1.1]{smallest_singular} which helps bound the minimum singular value of square matrices with i.i.d. zero-mean sub-Gaussian entries.  $i,j$-th entry of $\U^{*\top}\U_{rand}$ is the inner product of $i$-th column of $\Ustar$ and $j-$th column of $\U_{rand}$. $\Ustar$ has orthonormal columns and hence each entry of $\U^{*\top}\U_{rand}$ is mean-zero, unit variance Gaussian r.v. Thus, by  \cite[Theorem 1.1]{smallest_singular}, w.p., at least $1-(C\epsilon) - \exp^{-cr}$
\begin{align*}
	\sigma_{min}(\U^{*\top}\U_{rand})&\geq \epsilon(\sqrt{r}-\sqrt{r-1}) \\
	& = \epsilon\sqrt{r}\left(1-\sqrt{1-\frac{1}{r}}\right) \\
	& \geq \epsilon\sqrt{r}\left(1-\left(1-\frac{1}{2r}\right)\right) \\
	& = \epsilon\frac{1}{2\sqrt{r}}
\end{align*}
In the above, we used Bernoulli inequality $(1+x)^n\leq 1+nx$, where $0\leq n\leq 1$, $x\geq-1$ for $\sqrt{1-\frac{1}{r}}$. Use $\eps=0.1$.

Thus,  w.p. at least $1- 0.1 - \exp^{-cr}$,
\[
\sigma_{min}(GM) \ge \sigmarstar 0.1 \frac{1}{2\sqrt{r}} -  2.2 \sigma_{r+1}^* \sqrt{n}  - 31b_0 \sqrt{n}
\]

Together this implies
\begin{align}
\SD_F(\Ustar,\hat{\U})
& \le \frac{31 b_0 \sqrt{n} +  \sigma_{r+1}^* \cdot 2.2   \sqrt{n} \cdot \sqrt{r}}{ \sigmarstar 0.1 \frac{1}{2\sqrt{r}} -   2.2 \sigma_{r+1}^* \sqrt{n}  - 31b_0 \sqrt{n} } \\
& \le \frac{62 b_0 \sqrt{n r} +  4.4 \sigma_{r+1}^* \sqrt{n} \cdot r}{ \sigmarstar 0.1 -   4.4 \sigma_{r+1}^* \sqrt{nr}  - 62b_0 \sqrt{nr} }
\end{align}
To get this bound below $\eps_1$, we need $b_0 \le c \eps_1 / \sqrt{nr}$ and we need $\sigma_{r+1}^* < c/(\sqrt{n}r)$. Thus even to get $\eps_1 =0.99 $ (any value less than one), we need $b_0$ to be of order $1/\sqrt{nr}$.

\section{Geometric Median computation algorithms} \label{gm_append}

The geometric median cannot be computed exactly. We describe below two algorithms to compute it.
The first is the approach developed in Cohen et al \cite{cohen2016geometric}. This comes with a near-linear computational complexity bound. However, as we briefly explain below this is very complicated to implement and needs too many parameters. No numerical simulation results have been reported using this approach even in \cite{cohen2016geometric} itself and not in works that cite it either (to our best knowledge).

The practically used GM approach is Weiszfeld's algorithm \cite{beck2015weiszfeld}, which is a form of iteratively re-weighted least squares algorithms. It is simple to implement and works well in practice. However it either comes with an asymptotic guarantee, or with a finite time guarantee for which the bound on the required number of iterations is not easy to interpret. This bound depends upon the chosen initialization for the algorithm.
Because of this, we cannot provide an easily interpretable bound on its computational complexity.

\subsection{Cohen et al \cite{cohen2016geometric}'s algorithm: Nearly Linear Time GM}
\begin{algorithm}[ht]
	\caption{$\mathtt{AccurateMedian}(\epsilon_{GM})$}	
	\label{alg:nearlylinearmedian}	\label{cohen}
	\hspace*{\algorithmicindent} \textbf{Input}: points $z_1,...,z_L\in\mathbb{R}^{d}$	
	\hspace*{\algorithmicindent}	\textbf{Input}: desired accuracy $\epsilon_{GM}\in(0,1)$
	\begin{algorithmic}[1]
		\STATE\emph{Compute a 2-approximate geometric median and use it to center}\\	
		\hspace*{\algorithmicindent}Compute $x^{(0)}:=\frac{1}{L}\sum_{i\in[L]}z_i$ and $\widetilde{f}_{\ast}:=f(x^{(0)})$\\
		\COMMENT{Here $f(x)=\sum_{i\in[L]}\norm{x-z_i}_2$}\\
		\hspace*{\algorithmicindent}Let $t_{i}=\frac{1}{400\widetilde{f}_{\ast}}(1+\frac{1}{600})^{i-1}$, $\tilde{\epsilon}_{*}=\frac{1}{3}\epsilon_{GM}$,
		and $\tilde{t}_{*}=\frac{2L}{\tilde{\epsilon}_{*}\cdot\tilde{f}_{*}}$	
		\hspace*{\algorithmicindent}Let $\epsilon_{v}=\frac{1}{8}(\frac{\tilde{\epsilon}_{*}}{7L})^{2}$
		and let $\epsilon_{c}=(\frac{\epsilon_{v}}{36})^{\frac{3}{2}}$ 		
		\STATE $x^{(1)}=\mathtt{LineSearch}(x^{(0)},t_{1},t_{1},0,\epsilon_{c})$\\		
		\hspace*{\algorithmicindent}\emph{Iteratively improve quality of approximation}\\		
		\hspace*{\algorithmicindent}Let $T_{GM}=\max_{i\in\mathbb{Z}}t_{i}\leq\tilde{t}_{*}$		
		\FOR {$i \in [1,T_{GM}]$}	
		\item[]		
		\hspace*{\algorithmicindent}\emph{Compute $\epsilon_{v}$-approximate minimum eigenvalue
			and eigenvector of $\nabla^2 f_{t_{i}}(x^{(i)})$}			
		\STATE $(\lambda^{(i)},u^{(i)})=\mathtt{ApproxMinEig}(x^{(i)},t_{i},\epsilon_{v})$
		\item[]		
		\hspace*{\algorithmicindent}\emph{Line search to find $x^{(i+1)}$ such that $\norm{x^{(i+1)}-x_{t_{i+1}}}_{2}\leq\frac{\epsilon_{c}}{t_{i+1}}$}		
		\STATE$x^{(i+1)}=\mathtt{LineSearch}(x^{(i)},t_{i},t_{i+1},u^{(i)},\epsilon_{c})$		
		\ENDFOR
		\STATE \textbf{Output: }$\epsilon_{GM}$-approximate geometric median $x^{(T_{GM}+1)}$.
	\end{algorithmic}	
\end{algorithm}
The function $\mathtt{ApproxMinEig}$ in Algorithm \ref{alg:nearlylinearmedian} calculates an approximation of the minimum eigenvector of $\nabla^2 f_{t}(x)$. This approximation is obtained using the well-known power method, which converges rapidly on matrices with a large eigenvalue gap. By leveraging this property, we can obtain a concise approximation of $\nabla^2 f_{t}(x)$.
The running time of $\mathtt{ApproxMinEig}$ is $\mathcal{O}\left(Ld\log\left(\frac{L}{\epsilon_{GM}}\right)\right)$. This time complexity indicates that the algorithm's execution time grows linearly with $L$ and $d$ and logarithmically with $L/\epsilon_{GM}$.
The function $\mathtt{LineSearch}$ in Algorithm \ref{alg:nearlylinearmedian} performs a line search on the function $g_{t,y,v}(\alpha)$, as defined in Equation \ref{lem:quotient-func}. The line search aims to find the minimum value of $g_{t,y,v}(\alpha)$, subject to the constraint $|\mathbf{x} - (\mathbf{y} + \alpha \mathbf{v})|_2 \leq \frac{1}{49t}$, where $\mathbf{x}$ is the variable being optimized.
\begin{equation}
	g_{t,y,v}(\alpha)=\min_{\norm{x-(y+\alpha v)}_{2}\leq\frac{1}{49t}}f_{t}(x)\label{lem:quotient-func}
\end{equation}
To evaluate $g_{t,y,v}(\alpha)$ approximately, an appropriate centering procedure is utilized. This procedure allows for an efficient estimation of the function's value.
The running time of $\mathtt{LineSearch}$ is $\mathcal{O}\left(Ld\log^2\left(\frac{L}{\epsilon_{GM}}\right)\right)$. The time complexity indicates that the algorithm's execution time grows linearly with $L$ and $d$, while the logarithmic term accounts for the influence of $\frac{L}{\epsilon_{GM}}$ on the running time.

\subsection{Practical Algorithm:  Weiszfeld's method}

Weiszfeld's algorithm, Algorithm \ref{WZ}, provides a simpler approach for approximating the Geometric Median (GM). It is easier to implement compared to Algorithm \ref{alg:nearlylinearmedian}. It is an iteratively reweighted least squares algorithm. It iteratively refines the estimate by giving higher weights to points that are closer to the current estimate, effectively pulling the estimate towards the dense regions of the point set. The process continues until a desired level of approximation is achieved, often determined by a tolerance parameter, $\epsilon_{GM}$. While the exact number of iterations needed cannot be determined theoretically (as we will see from its guarantees below), the algorithm typically converges reasonably quickly in practice. 

We provide here the two known guarantees for this algorithm. 

\begin{theorem}\label{th:Weiszfeld's_asymp}[Corollary 7.1 \cite{beck2015weiszfeld}]
	Suppose that there is no optimal $\z\in\mathcal{A}=\{\z_1,...,\z_L\}$ such that it minimizes $\sum_{\ell=1}^{L}\left\| \z-\z_\ell \right\|$. Let $\{z_t\}_{t\geq 0}$ be the sequence generated by Weiszfeld's Algorithm \ref{WZ} with $z_0$ as given in the initialization of Algorithm \ref{WZ}. Then, for any $t\geq 0$, we have $\z_t\notin\mathcal{A}$ and $\z_t\rightarrow\z^*$ as $t\rightarrow\infty$. Here $z^*$ is the true GM.
\end{theorem}

\begin{algorithm}[ht]
	\caption{ Weiszfeld's Algorithm}\label{WZ}	
	\hspace*{\algorithmicindent} \textbf{Input} $\mathcal{A}=\{z_1,z_2,...,z_L\}$  \\
	\hspace*{\algorithmicindent} \textbf{Parameters} $T$, $\epsilon_{GM}$ \\
	\hspace*{\algorithmicindent} \textbf{Output} $z_{GM}$	
	\begin{algorithmic}
		\STATE \underline{\textbf{Initialization}}
		\STATE $z_0 = z_p + t_p d_p$, where
		\item[] $p\in\arg\min\{f(z_i):1\leq i\leq L\}$, $f(z)=\sum_{i=1}^{L}\left\| z-z_i \right\|$ and $d_p=\frac{R_p}{\norm{R_p}}$, $t_p=\frac{\norm{R_p}-1}{\mathcal{L}(z_p)}$
		\[\mathcal{L}(z)= \begin{cases}
			\sum_{i=1}^{L}\frac{1}{\lVert z-z_i\rVert} &\quad\text{if}\quad z\notin\mathcal{A}\\
			\sum_{i=1,i\neq j}^{L} \frac{1}{\lVert z_j-z_i\rVert}&\quad\text{if}\quad z=z_j\quad (1\leq j\leq L) \\
		\end{cases}\]\\
		\[R_{j}=\sum_{i=1,i\neq j}^{L}\frac{z_j-z_i}{\lVert z_i-z_j\rVert}\]
		\item[]
		\STATE \underline{\textbf{Iterative step}}
		\item[]
		\STATE $z_{t+1} = \left( \sum_{i=1}^{L}\frac{z_i}{\left\lVert z_i-z_t \right\rVert} \right)\bigg/\left( \sum_{i=1}^{L}\frac{1}{\left\lVert z_i-z_t \right\rVert} \right)$
		\item[]
		\STATE \underline{\textbf{Terminating Condition}}
		\item[]
		\STATE \begin{enumerate}
			\item $t>T$ \text{Upper bound on number of iterations}
			\item $\left\lVert z_{t+1}-z_t \right\rVert_2 < \epsilon_{GM}$
		\end{enumerate}
	\end{algorithmic}
\end{algorithm}	

\begin{theorem}\label{th:Weiszfeld's}[Theorem 8.2 \cite{beck2015weiszfeld}]
	Let $\{z_t\}_{t\geq 0}$ be the sequence generated by Weiszfeld's Algorithm \ref{WZ} with $z_0$ as given in the initialization of Algorithm \ref{WZ}. Then, for any $t\geq 0$, we have
	\[f(z_t)-f^*\leq \frac{M}{2t}\norm{z_0-z^*}^2\]
	where $M=\frac{2\mathcal{L}(z_p)L^2}{(\norm{R_p}-1)^2}$. Here $z^*$ is the true GM.
\end{theorem}

The first result above is asymptotic. 
The second one, Theorem \ref{th:Weiszfeld's}, gives convergence rate of $\mathcal{O}(M/t)$ where $M$ is as defined in the theorem. It is not clear how to upper bound $M$ only in terms of the model parameters ($d,L$ or $z^*$). Consequently, the rate of convergence is not clear.
Moreover, the expression for $z_0$ (initialization) is too complicated and thus it is not clear how to bound $\norm{z_0-z^*}^2$. Consequently, one cannot provide an expression for the iteration complexity that depends only on the model parameters.

  \subsection{Proof of Lemma \ref{gm_new}}
$\Jgood=\{\ell: \norm{\z_\ell-\tz}\leq \eps\norm{\tz}\}$ and define $z^\ast:=GM(z_1,...,z_{L})$ as exact Geometric median.

 	For $\z_\ell$ with $\ell \in\Jgood$, we have
 	\begin{align}
 		\norm{\z_{\gm}-\z_\ell} &= \norm{\z_{\gm}- \tz+ \tz-\z_\ell} \notag \\
 		&\geq \norm{\z_{\gm}- \tz} - \norm{\z_\ell- \tz} \notag \\
 		&\geq\norm{\z_{\gm}- \tz} - 2\epsilon\norm{\tz} + \norm{\z_\ell- \tz} \label{eq_gm_1}
 	\end{align}

 Moreover, by triangle inequality for $\z_\ell\notin\Jgood$, we have
 \begin{align}
 	\norm{\z_{\gm}-\z_\ell}
 	\geq  \norm{\z_\ell- \tz} - \norm{\z_{\gm}- \tz} \label{eq_gm_2}
 \end{align}

Summing \eqref{eq_gm_1}, \eqref{eq_gm_2} we get
\begin{align*}
	\sum_{\ell=1}^{L} \norm{\z_{\gm}-\z_\ell} &\geq \sum_{\ell=1}^{L}\norm{\z_\ell- \tz} \\
	& + (2|\Jgood|-L)\norm{\z_{\gm}- \tz} - 2|\Jgood|\epsilon\norm{\tz}
\end{align*}
By definition of $\z_\gm$ (approximate GM),  $\sum_{\ell=1}^{L}\norm{\z_{\gm}-\z_\ell}\leq (1+\epsilon_{\gm})\sum_{\ell=1}^{L}\norm{\z^\ast-\z_\ell}$. Hence,
\begin{align*}
	\sum_{\ell=1}^{L}\norm{\z_\ell- \tz} + (2|\Jgood|-L)\norm{\z_{\gm}- \tz} - 2|\Jgood|\epsilon\norm{\tz} \\\leq (1+\epsilon_{\gm})\sum_{\ell=1}^{L}\norm{\z^\ast-\z_\ell}
\end{align*}
Since $\z^\ast$ is the minimizer of $\min_{\z\in \Re^{n}}\sum_{\ell=1}^{L}\norm{\z-\z_\ell}$, so
\[\sum_{\ell=1}^{L}\norm{\z^\ast-\z_\ell}\leq \sum_{\ell=1}^{L}\norm{\z_\ell- \tz}.\]
Using this to lower bound the first term on the LHS of above,
\begin{align*}
	&\sum_{\ell=1}^{L}\norm{\z^\ast-\z_\ell}+ (2|\Jgood|-L)\norm{\z_{\gm}- \tz} - 2|\Jgood|\epsilon\norm{\tz} \\
	&\leq (1+\epsilon_{\gm})\sum_{\ell=1}^{L}\norm{\z^\ast-\z_\ell} \\
	\end{align*}

Arranging the terms and using the fact $|\Jgood|\geq (1-\alpha)L$  we get
\begin{align*}
	&\norm{\z_{\gm}- \tz} \\
	&\leq \frac{2|\Jgood|\epsilon\norm{\tz}}{2|\Jgood|-L} + \epsilon_{\gm}\frac{\sum_{\ell=1}^{L}\norm{\z^\ast-\z_\ell}}{2|\Jgood|-L} \\
	& \leq \frac{2(1-\alpha)\epsilon\norm{\tz}}{1-2\alpha} + \epsilon_{\gm}\frac{\max_{\ell\in[L]}\norm{\z_\ell}}{1-2\alpha}
\end{align*}
Using Claim \ref{cohen_1} ( with constant probability $1-c_\approxgm$ Algorithm \ref{cohen} obtains $(1+\eps_\gm)-$approximate geometric median $\z_\gm$ in order $T_\gm=C\log\left(\frac{L}{\eps_\gm}\right)$) implies that with probability $1-c_\approxgm$
\begin{align*}
	&\norm{\z_\gm-\tz} \leq C_{\alpha}\epsilon \| \tz\| + \epsilon_\gm\frac{\sum_{\ell=1}^{L}\norm{\z^*-\z_\ell}}{(1-2\alpha)L}\\
	&\leq C_{\alpha}\epsilon \| \tz\| + \epsilon_\gm\frac{\max_{\ell \in [L]} \norm{\z_\ell}}{1-2\alpha}
\end{align*}
where $C_\alpha := \frac{2(1-\alpha)}{1-2\alpha}$.

\subsection{Proof of Lemma \ref{gm_new_1}}\label{gm_prob_proof}

Given \[\Pr\{\norm{\z_\ell-\tz}\leq \epsilon\norm{\tz}\} \geq 1-p\]

then

\begin{align*}
	&\Pr\left\{\sum_{\ell=1}^{L}\mathbb{1}_{\{\norm{\z_\ell-\tz}\leq \epsilon\norm{\tz}\}}\geq L(1-\alpha)+L_{byz}\right\}\\
	&\geq \Pr\{T\geq L(1-\alpha)+L_{byz}\}
\end{align*}
where $T\sim \text{Binomial}(L,1-p)$ (First-order stochastic domination)

By Chernoff's bound for binomial distributions, the following holds:

\[\Pr\{T\geq L(1-\alpha)+L_{byz}\}\geq 1-\exp(-L\psi(\alpha-\tau,p))\]

where $\tau=\frac{L_{byz}}{L}$

This then implies w.p. at least $1-\exp(-L\psi(\alpha-\tau,p))$,
\[\sum_{\ell=1}^{L}\mathbb{1}_{\{\norm{\z_\ell-\tz}\leq \epsilon\norm{\tz}\}}\geq L(1-\alpha)+L_{byz}\geq L(1-\alpha)\]
where $\alpha\in(\tau,1/2)$. Using Lemma \ref{gm_new}
\[\norm{\z_{\gm}-\tz} \leq C_{\alpha}\epsilon\norm{\tz} + \epsilon_{\gm}\frac{\max_{1\leq \ell\leq L}\norm{\z_\ell}}{1-2\alpha} \]

w.p. at least $1-c_\approxgm-\exp(-L\psi(\alpha-\tau,p))$. Fixing $\alpha=0.4$ we get the result.

\subsection{Proof of Corollary \ref{gm_new_2}}

$\Jgood$ denotes the set of good node (nodes whose estimates $\z_\ell$ satisfy $\norm{\z_\ell-\tz}\leq \eps \|\tz\|$) with the stated probability.  First we need to show that,  with high probability, none of the entries of $\Jgood$ are thresholded out.
Using given condition 	$\Pr\{\norm{\z_\ell-\tz}\leq \eps \|\tz\| \} \geq 1-p$ and union bound, we conclude that,  w.p. at least $1-(1-\tau)Lp$, $\max_{\ell \in \Jgood} \|\z_\ell\| \le (1+\eps) \|\tz\|= \thresh$. This means that, with this probability, none of the $\Jgood$ elements are thresholded out.

For the set $\{\z_1,...,\z_{L}\}\setminus\{ \z_\ell:\norm{\z_\ell}> (1 + \eps) \|\tz\| \}$ we apply Lemma \ref{gm_new_1}. Since  $\left|\{\z_1,...,\z_{L}\}\setminus\{ \z_\ell:\norm{\z_\ell}> (1 + \eps) \|\tz\| \}\right|=L' \leq L$,

$(1-\tau)L\leq L'\leq L$ implies $\tau'\leq\tau<0.4$ hence condition of Lemma \ref{gm_new_1} is satisfied using Lemma \ref{gm_new_1} w.p. at least $1-c_\approxgm-\exp(-L'\psi(0.4-\tau',p ))-(1-\tau)Lp\geq 1-c_\approxgm-\exp(-L\psi(0.4-\tau,p ))-Lp $,
\[\norm{\z_{\gm}-\tz} \le  6 \eps \| \tz\| + 5 \eps_\gm   (1 + \eps) \|\tz\| < 14 \max(\eps,\eps_\gm) \|\tz\|\]

\bibliographystyle{IEEEbib}
\bibliography{tipnewpfmt_kfcsfullpap, byz}

\end{document}